\documentclass[a4paper,twocolumn,longbibliography,accepted=2025-07-29]{quantumarticle}
\pdfoutput=1 
\usepackage[numbers, sort&compress]{natbib}
\usepackage{graphicx} 
\usepackage{dirtytalk} 
\usepackage{bm}
\usepackage{amsmath}
\usepackage{mathrsfs}
\usepackage{amsthm}
\usepackage{braket}
\usepackage{amsfonts,amssymb,amsthm, amscd, bbm,braket}
\usepackage{graphicx}   
\usepackage{subcaption}  
\usepackage{amsbsy} 
\usepackage[bold]{hhtensor} 
\usepackage{natbib}
\usepackage{algorithm}
\usepackage[noend]{algpseudocode}
\usepackage{algpseudocode}

\usepackage{multirow}
\usepackage{tcolorbox}
\usepackage{verbatim} 

\newcommand{\RR}{\mathbb{R}}
\newcommand{\CC}{\mathbb{C}}

\newcommand{\Tr}{\mbox{Tr}}

\renewcommand\bra[1]{{\langle{#1}|}}
\renewcommand\ket[1]{{|{#1}\rangle}}

\newtheorem{corollary}{Corollary}
\newtheorem{theorem}{Theorem}
\newtheorem{lemma}{Lemma}
\newtheorem{proposition}{Proposition}

\newtheorem{hypothesis}{Hypothesis}

\usepackage[breaklinks=true]{hyperref}

\hypersetup{
  colorlinks   = true, 
  urlcolor     = blue, 
  linkcolor    = blue, 
  citecolor   = red 
}

\begin{document}

\title{
Estimating quantum Markov chains using coherent absorber post-processing and pattern counting estimators}

\author{Federico Girotti}
\affiliation{School of Mathematical Sciences, University of Nottingham, United Kingdom}
\affiliation{Centre for the Mathematics and Theoretical Physics of Quantum Non-Equilibrium Systems,
University of Nottingham, Nottingham, NG7 2RD, UK}
\affiliation{Department of Mathematics, Polytechnic University of Milan, Milan, Piazza L. da Vinci 32, 20133, Italy}

\author{Alfred Godley}
\affiliation{School of Mathematical Sciences, University of Nottingham, United Kingdom}
\affiliation{Centre for the Mathematics and Theoretical Physics of Quantum Non-Equilibrium Systems,
University of Nottingham, Nottingham, NG7 2RD, UK}

\author{M\u{a}d\u{a}lin Gu\c{t}\u{a}}
\affiliation{School of Mathematical Sciences, University of Nottingham, United Kingdom}
\affiliation{Centre for the Mathematics and Theoretical Physics of Quantum Non-Equilibrium Systems,
University of Nottingham, Nottingham, NG7 2RD, UK}

-----------------------------------------------------------------------------------------------------------
\begin{abstract}
We propose a two step strategy for 
estimating one-dimensional dynamical parameters of a quantum Markov chain, which involves quantum post-processing the output using a coherent quantum absorber and a ``pattern counting'' estimator computed as a simple additive functional of the  outcomes trajectory produced by sequential, identical measurements on the output units. We provide strong theoretical and numerical evidence that the estimator achieves the quantum Cram\'{e}r-Rao bound in the limit of large output size.

Our estimation method is underpinned by an asymptotic theory of translationally invariant modes (TIMs) built as averages of shifted tensor products of output operators, labelled by binary patterns. For large times, the TIMs form a bosonic algebra and the output state approaches a joint coherent state of the TIMs, whose amplitude depends linearly on the mismatch between system and absorber parameters. Moreover, in the asymptotic regime, the TIMs capture the full quantum Fisher information of the output state. While directly probing the TIMs' quadratures seems impractical, we show that the standard sequential measurement is an effective joint measurement of all the TIMs number operators; 
indeed, we show that counts of different binary patterns extracted from the measurement trajectory have the expected joint Poisson distribution. Together with the displaced-null methodology of \cite{GiGoGu} this provides a computationally efficient estimator which only depends on the total number of patterns. This opens the way for similar estimation strategies in continuous-time dynamics, expanding the results of \cite{DayouCounting}.

\end{abstract}


\maketitle
\section{Introduction}

Quantum statistical inference \cite{Helstrom1976,Holevo2011,Hayashi2005,Paris2008,TothReview,RafalReview,Tomo2,Albarelli2020,Sidhu_2020} provides the mathematical framework for enhanced metrology 
\cite{ Metrology1,Fujiwara2008,Giovannetti2011, Escher11,Metrology2,Girolami14,Smirne16,Seveso2017,Haase18,Metrology4,Rossi20,Sisi21}, imaging \cite{Tsang16,Tsang21,Lupo20,Fiderer21,Oh21},  waveform 
and noise 
estimation \cite{TWC11,Berry2015,Ng16,Norris16,Shi23,Sung19,Tsang23}, 
and quantum sensing applications \cite{Dowling,Degen2017,Pezze18,Marciniak2022,Zwick2023} including  time keeping \cite{Robinson24}, magnetometry \cite{Jones2009, Jan2021,Brask2015,ARPG17}, biomedical sensing \cite{Aslam2023}, 
thermometry \cite{Correa2015, Mehboudi2019}, gravitational wave detection \cite{GW1,GW2,GW3,GW4}.

The cornerstone of quantum estimation is the quantum Cram\'{e}r-Rao bound (QCRB)  \cite{Holevo2011,Helstrom1976,Belavkin76, QCR1} 
which places a fundamental restriction on the precision in estimating unknown  parameters of a quantum state. 
For one-dimensional parameters the bound is attainable in the limit of many copies, by measuring a specific observable called the symmetric logarithmic derivative, which is the quantum analogue of the classical score function. 
However, when dealing with complex models involving correlated states of many-body systems,
optimal measurements may be hard to compute and implement in practice.
Therefore, it is particularly important 
to devise \emph{realistic} measurement schemes which allow the estimation of unknown parameters with close to optimal precision, by means of \emph{computationally efficient} estimators.

In this paper we provide a general measurement and data processing protocol for estimating an arbitrary  dynamical parameter of a \emph{quantum Markov chain} (QMC), which satisfies the above criteria. A QMC is a discrete time model of an open quantum system, in which the system interacts successively with a sequence of identically prepared ``noise units'' representing the environment, cf. Figure \ref{fig:intro} a). The setup is similar to Haroche's photon-box one-atom maser \cite{Haroche} and to that used in quantum collision models \cite{Ciccarello22}, and provides a physical mechanism for generating versatile many-body states such  as matrix product states \cite{PerezGarciaWolfCirac,SchonSolanoVerstraeteCiracWolf} and finitely correlated states \cite{FannesNachtergaeleWerner,FannesNachtergaeleWerner2}. By discretising time, QMCs can  be used to model continuous-time dynamics of a Markovian open system coupled with Bosonic input-output channels \cite{AttalPautrat,VerstraeteCirac,Gough04}.

\begin{figure*}[!ht]
    \centering
   \includegraphics[width=\linewidth]{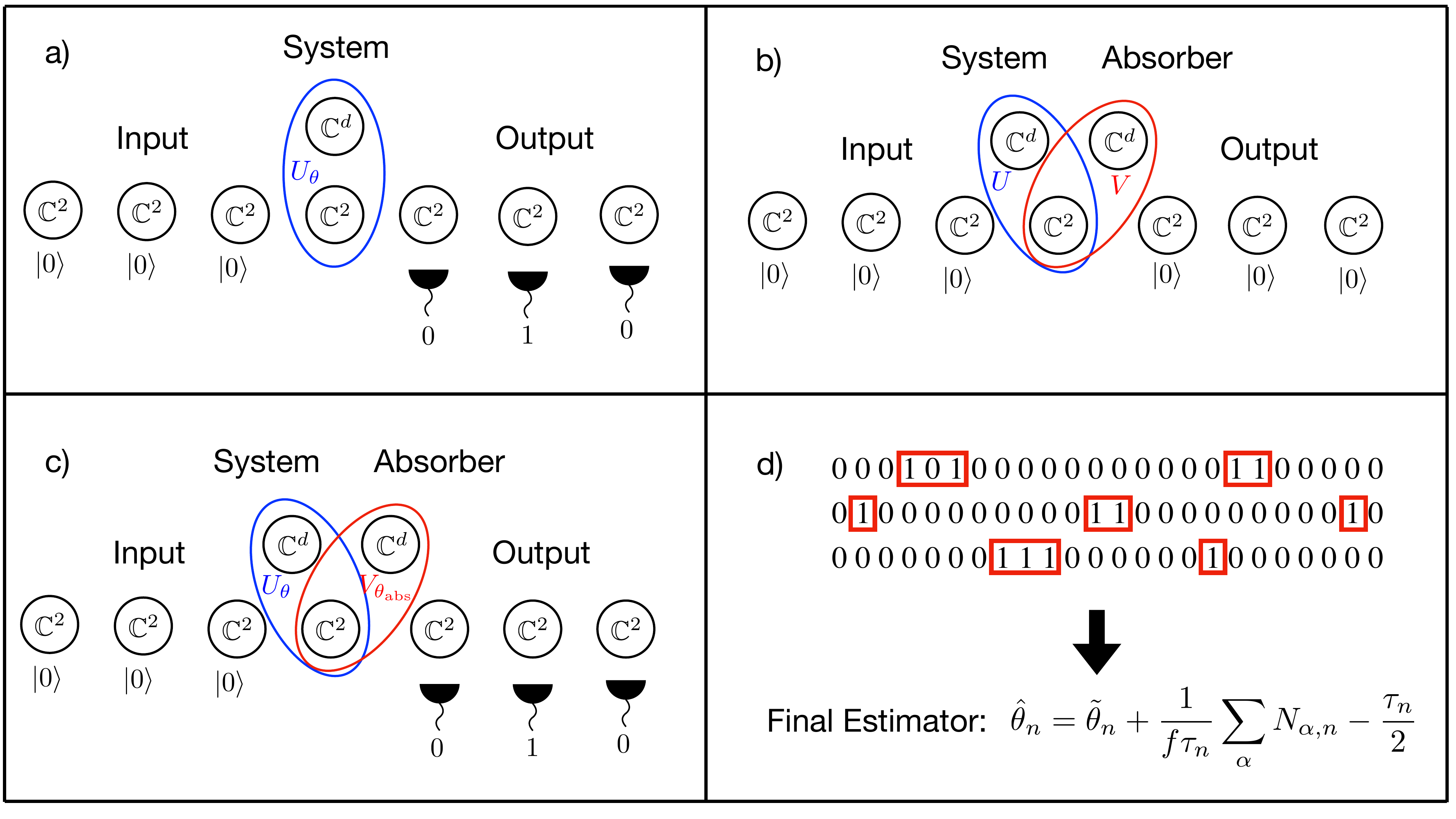}
    \caption{Basic elements of the pattern counting estimator. Panel a) A quantum Markov chain as a system interacting sequentially with the environment via a parameter dependent unitary $U_\theta$. The first stage estimator $\tilde{\theta}_n $ is obtained by performing a standard sequential measurement on the output and equating empirical and expected counts. Panel b) Post-processing the output using a coherent absorber. When system and absorber parameters match, the output is identical to the input (vacuum) Panel c) After the first estimation stage the absorber is fixed at a value $\theta_{\rm abs} = \tilde{\theta}_n -\delta_n$ where $\tilde{\theta}_n$ is the preliminary estimator and $\delta_n$ is the parameter shift required by the displaced-null measurement theory \cite{GiGoGu}. The output generated by the system and absorber dynamics with unitary $V_{\theta_{\rm abs}}U_\theta$ is measured sequentially in the standard basis. Panel d) Given a measurement trajectory, excitation patterns are identified as binary sequences starting and ending with a $1$ separated by long sequences of $0$s. The final estimator is a correction to the preliminary estimator which depends only on the total number of patterns $\sum_{\alpha} N_{\alpha, n}$, the QFI $f$ at $\tilde{\theta}_n$ and 
    the displacement parameter $\tau_n$.}
    \label{fig:intro}
\end{figure*}

For clarity, it is useful to distinguish between two mainstream approaches to parameter estimation in quantum open systems. In the setting of \cite{Huelga97,Benatti14, Smirne16,Sekatski2017,DDCS,ZhouZhangPreskillJuang,GoreckiZhouSisiJiangDemko}, the quantum system undergoes a noisy evolution depending on an unknown parameter, and the experimenter tries to extract information about the parameter by repeatedly applying instantaneous direct measurements and control operations while the system is evolving. In contrast, in this paper we adopt the setting commonly used in quantum optics and input-output theory 
\cite{GardinerZoller,WisemanMilburn,GJ09,CKS17} where the experimenter does not have direct access to the system but can measure the output field of an environment channel coupled to the system. This allows the experimenter to track the conditional state of the system by means of stochastic filtering equations \citep{Belavkin94,DCM,WM93,BvHJ07,Carmichael} and control it using feedback.
As these techniques require full knowledge of the system's dynamical parameters, it is important to devise tools for estimating such parameters from the stochastic trajectory of the measurement record. 

Since the early works \cite{Mab96,GW01}, many aspects of continuous-time estimation have been investigated, including adaptive estimation \cite{Berry02,Wiseman04} filtering methods \cite{Ralph11,CG09,Six15}, Heisenberg scaling \citep{GutaMacieszczakGarrahanLesanovsky,ARPG17,Genoni18,Ilias22}, sensing with error correction \citep{Plenio16}, Bayesian estimation \citep{GM13,Negretti13,KM16,Ralph17,Zhang19}, quantum smoothing \citep{Tsang09,T09,Tsang10,Guevara15}, 
estimation of linear systems \citep{GY16,Genoni17,Levitt17,LGN18}, central limit and large deviations theory for trajectories \cite{Garrahan10, GutavanHorssen15,Burgarth15,Garrahan19}, concentration bounds for time averaged observables \cite{BR21,GiGaGu23}, estimation with feedback control \cite{Fallani22}. However, while the the ultimate precision limit  can be expressed in terms of the output \emph{quantum Fisher information} (QFI) \cite{Guta2011,Molmer14,GutaCB15,Guta_2015,Guta_2017,Genoni17}, and can in principle be attained by measuring a certain output observable, standard  measurement protocols such as counting and homodyne typically do not achieve this limit. Therefore, attaining this bound with realistic measurements has been an important open problem in the field.

Two recent papers \cite{Godley2023,DayouCounting} have addressed this problem by introducing the idea of \emph{quantum post-processing} of the output state using a \emph{quantum coherent absorber} \cite{StannigelRAblZoller}. For a given QMC reference dynamics, the absorber takes the system's output as its own input, and is characterised by the property that it ``reverts'' the action of the system so that the absorber's output is a trivial product state (vacuum in continuous-time dynamics), cf. Figure \ref{fig:intro} b). In the statistical estimation framework, the absorber is set to a particular reference parameter of the QMC dynamics, which is kept fixed throughout the protocol. A small deviation of the true parameter from this value will lead to non-trivial output statistics which can be used to estimate the deviation as illustrated in Figure \ref{fig:intro} c). In \cite{Godley2023}, two of the present authors proved that the QCRB can be achieved by performing sequential, adaptive measurements on the output units, after the interaction with the coherent absorber. In addition, the adaptive measurement can be implemented efficiently in a Markovian fashion. The work \cite{DayouCounting} deals with the same problem but in the continuous-time setting, and proposes to perform a standard counting measurement (instead of an adaptive one) in conjunction with post-processing using the coherent absorber. In both papers the final estimator was computed from the measurement trajectory using the maximum likelihood method, which becomes computationally expensive for long trajectories. 

The strategy proposed here is similar to that of \cite{DayouCounting}, albeit in discrete rather than continuous time, but strengthens it in several important aspects. Firstly, we employ the technique of \emph{displaced-null measurements} \cite{GiGoGu} to 
provide a precise recipe for choosing the coherent absorber parameter. As we explain below, this is an important technical detail, as the intuitive choice of absorber parameter \emph{fails} to achieve the QCRB in the limit of long trajectories. Secondly, our strategy employs a two step adaptive procedure which allows us to compute the final estimator as a simple linear transformation of the total number of ``pattern counts'' which can be easily extracted from the measurement trajectory, cf. Figure \ref{fig:intro} d). This circumvents the computational cost associated with other estimators such as maximum likelihood, and the lack of theoretical results on their performance, specific to this context. Thirdly, we provide strong theoretical evidence that the final estimator achieves the QCRB in the limit of large times. This is based on a novel asymptotic representation of the output in terms of \emph{translationally invariant modes} which obey canonical commutation relations and whose state is shown to satisfy the \emph{quantum local asymptotic normality} property \cite{LAN1,LAN3,LAN5,LAN6}. This allows us to reduce the QMC estimation problem to a more familiar Gaussian estimation one, and cast the sequential output measurement as a counting measurement of the Gaussian modes.

We now give a brief summary of the two stages estimation strategy proposed here and the related mathematical results. We consider a QMC whose dynamics depends on a one-dimensional parameter 
$\theta$ which we aim to estimate by measuring the output state produced after $n$ time steps. In the first stage we run the QMC dynamics with the unknown parameter $\theta$ for 
$1\ll\tilde{n}\ll n$ time steps and measure the noise units in a fixed basis, cf. Figure \ref{fig:intro} a). From the total counts statistics we construct a rough estimator $\tilde{\theta}_n$ by matching the empirical frequency to its expected value. In general this estimator is neither optimal nor unbiased, but its mean square error has the standard $1/\tilde{n}$ scaling \cite{GiGaGu23}. Its role is to reduce the range of the 
unknown parameter to a shrinking region, and allow the second stage measurement to be optimal in this region, see \cite{GiGoGu} for details in the case of independent samples. In the second stage we run the system \emph{and} absorber QMC for the remaining $n-\tilde{n}$ time steps and measure the output in the standard basis. If the absorber parameter matched the true system parameter $\theta$, this measurement would produce a string of $0$s (corresponding to no counts in continuous-time) cf. Figure \ref{fig:intro} b). Therefore, it would seem natural to choose the absorber parameter to be $\tilde{\theta}_n$, our best guess at the unknown parameter $\theta$. However, this choice is unsuitable since for small deviations $\Delta_n = \theta-\tilde{\theta}_n$, the counting statistics depends quadratically on $\Delta_n$, which prevents the estimation of $\theta$ at standard $1/n$ rate. This non-identifiability issue is explained in detail in \cite{GiGoGu}, which also provides the solution to this problem. We deliberately set the absorber parameter at $\theta_{\rm abs}= \tilde{\theta}_n -\delta_n$, which is away from the best guess by a small ``displacement'' $\delta_n\downarrow 0$ chosen to be larger than the uncertainty $|\Delta_n|$. This allows us to unambiguously identify $\theta$ from counts statistics. Stage two of the estimation procedure is illustrated in Figure \ref{fig:intro} c).

We turn now to the question of estimating $\theta$ from the counts trajectory $\omega= (\omega_1,\omega_2,\dots, \omega_n)$ of the second stage measurement (setting $n-\tilde{n}$ to $n$ for simplicity). Since $\theta-\theta_{\rm abs}$ is vanishingly small (even with the extra displacement), $\omega$ will typically contain a small number of $1$s while most of the outcomes will be $0$, cf. Figure \ref{fig:intro} d). This allows us split the trajectory into long sequences of $0$s and in between them, binary ``excitation patterns'' starting and ending with a $1$. For each pattern

$\alpha$ (e.g. $1$, $11$, $101$ etc.) we count the number of occurrences $N_{\alpha,n}$. In Theorem \ref{thm:trajs} we show that in the limit of large $n$ the counts $N_{\alpha,n}$ become independent Poisson variables whose intensities are $\lambda_\alpha u^2$ where $u= \sqrt{n}(\theta-\theta_{\rm abs})$ is the ``local parameter'' and $\lambda_\alpha$ is a model dependent coefficient which can be computed explicitly. Moreover the total Fisher information of the Poisson variables is equal to the output QFI, which shows that the pattern counts statistics capture the full information of the output state. Using this asymptotic behaviour, we construct a simple estimator $\hat\theta_n$ (cf. equations \eqref{eq:theta.hat.final} and \eqref{eq:u.hat.final}) which is linear in the total pattern count, and we argue why it should achieve the QCRB in the limit of large $n$. To summarise, the two step procedure provides a computationally and statistically efficient estimation method which involves only standard basis measurements and a minimal amount of ``quantum post-processing'' implemented by the coherent absorber.

For a more in-depth understanding of why the  excitation pattern counts have asymptotically Poisson distributions, we refer to sections \ref{sec:TIM} and \ref{sec:CLTPoisson} where we develop a theory of \emph{translationally invariant modes} (TIMs) of the output. These modes turn out to capture all statistical information about the unknown parameter, and can be measured simultaneously and optimally by performing the sequential standard output measurement. For each excitation pattern $\alpha=(\alpha_1,\dots ,\alpha_k)\in \{0,1\}^k $ we define the creation operator $A^*_\alpha(n)$ on the output chain of length $n$. This consists of a running average 
$$
A^*_\alpha(n)=\frac{1}{\sqrt{n}}\sum_{i=1}^{n-k+1}\sigma^\alpha_i
$$
where $\sigma^\alpha_i$ is the tensor product of the type 
$\sigma^\alpha = \sigma^{\alpha_1}\otimes \dots \otimes \sigma^{\alpha_k} $ where $\sigma^0= \mathbf{1}$ and $\sigma^1 = \sigma^+ = |1\rangle \langle 0|$, with first tensor acting on position $i$ of the output chain. In Proposition \ref{prop:fock.states} and Corollary  \ref{prop:creation.action} we show that asymptotically with $n$, by applying the creation operators to the reference (vacuum) state $|0\rangle^{\otimes n}$, we obtain Fock-type states with different excitation pattern numbers. The creation and annihilation operators satisfy the Bosonic commutation relations with each excitation pattern being an independent mode. For large $n$, a Fock state is a superposition of basis states consisting of long sequences of $0$s interspersed with the corresponding patterns appearing in any possible order. One of our key results, Theorem \ref{thm:limdistribution} shows that when the gap between system and absorber parameters scales as $\theta - \theta_{\rm abs}=u/\sqrt{n}$, the quadratures of the excitation pattern modes satisfy the  Central Limit Theorem and the corresponding joint state is a product of coherent states whose amplitudes are linear in the local parameter $u$. The total QFI of this multimode coherent state is equal to the output QFI, showing that the TIMs contain all statistical information about the dynamics. We also prove separately, that the number operators of the TIMs have asymptotic Poisson distributions, as expected for a coherent state. Together with the result of Theorem \ref{thm:trajs}, this completes a circle of ideas, which played a crucial role in formulating our estimation strategy. In a nutshell, when looking at the output from the perspective of the TIMs, one deals with a simple Gaussian estimation problem. Using the displaced-null method, we can achieve the QCRB by measuring the number operators of the TIMs and such a measurement can be implemented by simple sequential counting measurements followed by the extraction of pattern counts from the measurement trajectory $\omega$.

The paper is organised as follows. In section \ref{sec:review.null-measurements} we give a brief review of quantum estimation theory and the displaced-null measurement technique developed in \cite{GiGoGu}. In section \ref{sec:QMC} we introduce the notion of QMC and the estimation problem, together with the idea of quantum post-processing using a coherent absorber. In section \ref{sec:TIM} we define the translationally invariant modes of the output and establish their Fock space properties. In section \ref{sec:CLTPoisson} we show that the restriction of the output state to the TIMs is a coherent state whose amplitude is linear in the local parameter and whose QFI is equal to the output QFI (cf. Theorem \ref{thm:limdistribution} and Corollary \ref{coro:fisher}). 
In section \ref{sec:limit.th.counting} we establish that the excitation pattern counts obtained from the sequential output measurement have asymptotically Poisson distribution (cf. Theorem \ref{thm:trajs}). In section \ref{sec:estimator} we formulate our measurement and 
estimation strategy and define the "pattern counts" estimator. Finally in section \ref{sec:numerics} we present results of a simulation study confirming the earlier theoretical results.

\section{Quantum estimation and the displaced null measurement technique}
\label{sec:review.null-measurements}

In this section we give a brief overview  of the quantum parameter estimation theory \cite{Hayashi2005, Paris2008,TothReview,RafalReview,Tomo2,Albarelli2020,Sidhu_2020} used in this paper, with an emphasis on asymptotic theory and the displaced-null measurement technique developed in our previous work \cite{GiGoGu}. In particular, we explain  why this method is asymptotically optimal, by employing the Gaussian approximation technique called local asymptotic normality 
\cite{LAN1,LAN2,LAN3,LAN4,LAN5,LAN6,Yamagata13,Fujiwara20,Fujiwara22}. Later on, this picture will guide our intuition when dealing with the Markov estimation problem. For our purposes it suffices to discuss the case of one-dimensional parameters, and we refer to \cite{GiGoGu} for the multi-dimensional setting.

Let $\rho_\theta\in M(\mathbb{C}^d)$ be a family of quantum states depending smoothly on a one-dimensional parameter 
$\theta$. Consider a measurement described by a positive operator valued measure 
$\{M_1,\dots, M_k\}$ and let $X$ be the measurement outcome with probability distribution $p_\theta(X=i) = {\rm Tr}(\rho_\theta M_i)$. The quantum Cram\'er-Rao bound (QCRB) \cite{Helstrom1976,Holevo2011,Belavkin76,YuenLax76,QCR1,QCR2} states that the variance of any unbiased estimator $\hat{\theta} =\hat{\theta}(X)$ is lower bounded as
$$
{\rm Var}(\hat{\theta}) =\mathbb{E}_\theta(\hat{\theta}-\theta)^2 \geq 
F_\theta^{-1}
$$
where $F_\theta$ is the quantum Fisher information (QFI) defined as 
$F_\theta = {\rm Tr} (\rho_\theta\mathcal{L}_\theta^2)$ with $\mathcal{L}_\theta$ the symmetric logarithmic derivative (SLD) operator which satisfies 
$$\frac{d\rho_\theta}{d\theta} =\frac{1}{2}\left(\mathcal{L}_\theta \rho_\theta + \rho_\theta \mathcal{L}_\theta\right).
$$
In general, the QCRB is not achievable when only a single copy of $\rho_\theta$ is available. However, the bound \emph{is} attainable in the asymptotic limit of large number of samples by the following two step adaptive procedure \cite{GillMassar}. Given $n$ copies of $\rho_\theta$ one can use a small proportion of the samples (e.g. $\tilde{n} = n^{1-\epsilon}$ for a small $\epsilon >0$) to compute a preliminary (non-optimal) estimator 
$\tilde{\theta}_n$ of $\theta$; reasonable estimators will concentrate around $\theta$ such that 
$|\tilde{\theta}_n-\theta|=O (n^{-1/2+\epsilon})$ with high probability, which will be assumed throughout. In the second step, one measures the SLD operator 
$\mathcal{L}_{\tilde{\theta}_n}$ on each of the remaining copies. If $X_1,\dots , X_{n^\prime}$ are the outcomes of these measurements (with $n^\prime = n-\tilde{n}$) then the estimator
\begin{equation}
\label{eq:two.step.SLD}
\hat{\theta}_n:= \tilde{\theta}_n+ \frac{1}{F_{\tilde{\theta}_n}n^\prime}\left(\sum_{i=1}^{n^\prime} X_i\right)
\end{equation}
is \emph{asymptotically optimal} in the sense that 
\begin{equation}
\label{eq:asymptotic.achievability.QCRB}
n\mathbb{E}_\theta(\hat{\theta}_n -\theta)^2 \to F_\theta^{-1}
\end{equation} 
in the limit of  large $n$ and in addition $\hat{\theta}_n$ is \emph{asymptotically normal}, i.e. 
$\sqrt{n}(\hat{\theta}_n -\theta)$ converges in distribution to the normal $N(0, F_\theta^{-1})$. 

However, for certain models including that considered in this paper,  measuring the SLD may not be feasible experimentally. Instead, we will  use a different method called \emph{displaced-null measurement} \cite{GiGoGu}, which aims to estimate the parameter of a \emph{pure state models} 
$\rho_\theta = |\psi_\theta\rangle \langle \psi_\theta|$ 
by measuring each copy in a basis that contains the vector $|\psi_{\tilde{\theta}}\rangle$ with $\tilde{\theta}$ close to the true parameter 
$\theta$. The Fisher information of such a measurement is known to converge to the QFI as $\tilde{\theta}$ approaches $\theta$ \cite{NullQFI1,NullQFI2,NullQFI3}. This suggests that the QCRB can be achieved asymptotically by using a two-step strategy similar to the SLD case: one first obtains a preliminary estimator $\tilde{\theta}_n$ and then measures each copy in a basis containing the vector $|\psi_{\tilde{\theta}_n}\rangle$. However, it turns out that this ``null measurement" strategy fails due to the fact that for small deviations from $\tilde{\theta}_n$, the outcome probabilities depend on $(\theta-\tilde{\theta}_n)^2$ and one cannot distinguish between left and right deviations from $\tilde{\theta}_n$, cf. \cite{GiGoGu} for the precise mathematical statement. This non-identifiability issue can be sidestepped by deliberately changing the reference parameter from $\tilde{\theta}_n$ to 
$\theta_0:=\tilde{\theta}_n - \delta_n$ where $\delta_n =n^{-1/2+3\epsilon}$, so that $\theta= \theta_0 +(u+\tau_n)/\sqrt{n}$ with 
$\tau_n=n^{3\epsilon}$. The choice of $\tau_n$ is not unique but we refer to \cite{GiGoGu} for the general requirements. Since 
$|\theta -\tilde{\theta}_n| =O(n^{-1/2+\epsilon}) \ll \delta_n$, it means that 
$\theta$ lies on the right side of $\theta_0$ and can be unambiguously identified from the outcomes of a measurement in a basis $\{|e_0\rangle, \dots , |e_{d-1}\rangle \}$ such that 
$|e_0\rangle \equiv 
|\psi_{\theta_0}\rangle$. In addition, as $\theta_0$ approaches $\theta$ in the limit of large $n$, the displaced-null measurement exhibits the optimality properties of the ``null measurements" without sharing their non-identifiability issues.

Let $X_1,\dots , X_{n\prime}\in \{0, 1,\dots , d-1\}$ 
be the independent outcomes of basis $\{|e_0\rangle, \dots , |e_{d-1}\rangle \}$ measurements performed on the remaining $n^\prime = n-\tilde{n}$ systems, and let $N_{j,n}$ denote the counts of the outcome $j=0,\dots d-1$. The displaced-null estimator based on the two-stage measurement strategy is defined as follows
$$
\hat{\theta}_n := \tilde{\theta}_n +\hat{u}_n/\sqrt{n} 
$$
with local parameter estimator
$$\hat{u}_n= 
\frac{2}{\tau_n f}\sum_{j=1}^{d-1}N_{j,n} - \frac{\tau_n}{2}
$$
where $f= 4 \|\dot{\psi}_{\tilde{\theta}_n}\|^2$ is the QFI at $\tilde{\theta}_n$. The expression of $\hat{u}_n$ is derived in section \ref{sec:estimator} by taking into account that the pattern counts are asymptotically normal (conditional to the preliminary estimator $\tilde{\theta}_n$) and finding the unbiased linear combination with the smallest variance, which turns out to be related to the total number of patterns as above. 
The estimator $\hat{\theta}_n$ is asymptotically optimal in the sense of equation \eqref{eq:two.step.SLD} and asymptotically normal.

In appendix \ref{sec.LAN&DNM} we give more insight into this method by analysing its  properties using the theory of local asymptotic normality \cite{LAN1,LAN2,LAN3,LAN4,LAN5,LAN6}. In a nutshell, for large $n$ the original model becomes equivalent to a Gaussian one consisting of a multi-mode coherent state whose amplitude depends linearly on the parameter, while the SLD and displaced-null strategies translated into measuring a quadrature and respectively the modes number operators of the shifted state. While this i.i.d. setup is different from the Markovian one studied in this paper, the overall asymptotic picture is similar and reader may find the i.i.d. case useful in guiding the intuition.

\section{Quantum Markov chains and post-processing  using coherent absorbers}
\label{sec:QMC}

We start this section by reviewing the problem of  estimating dynamical parameters of quantum Markov chains (QMC). We then introduce the notion of quantum coherent absorber, which will play a key role in designing an optimal sequential measurement strategy.

A quantum Markov chain consists of a system interacting successively with a chain of independent "noise units" (the input) modelling the environment. In this paper the system's space is taken to be $\mathcal{H}_s \cong\mathbb{C}^d$ while the "noise units" are two dimensional systems prepared in the state $|0\rangle$ where $\{|0\rangle, |1\rangle\}$ is the standard basis in $\mathcal{H}_n\cong\mathbb{C}^2$. We expect that the theory developed here works for general finite dimensional inputs, but we restrict here to this minimal setup which can be used to represent a discretised version of a continuous-time Markovian model with a single Bosonic field \cite{PautratAttal}.

At each time step  the system interacts with the input unit via a unitary
$U$ on $\mathcal{H}_s\otimes \mathcal{H}_n$. If the system is initially prepared in a state $|\varphi\rangle$, the joint state of system and noise units (output) after $n$ times steps is 
\begin{eqnarray}\label{eq:total.state}
|\Psi_n\rangle &=& U_n |\varphi\otimes 0^{\otimes n}\rangle 
\\
&=& 
U^{(n)}\cdot\dots \cdot U^{(2)}\cdot U^{(1)}|\varphi\otimes 0^{\otimes n}\rangle\in 
\mathcal{H}_s\otimes \left(\mathcal{H}_n\right)^{\otimes n}
\nonumber
\end{eqnarray}
where $U^{(i)}$ is the unitary acting on the system and the $i$-th noise unit. By expanding the state \eqref{eq:total.state} with respect to the standard product basis in the output we have
\begin{equation}\label{eq.Psi_n}
|\Psi_n\rangle = 
\sum_{i_1,\dots , i_n \in \{0,1\}} 
K_{i_n}\dots K_{i_1} |\varphi\rangle \otimes |i_{1}\rangle\otimes \dots\otimes |i_{n}\rangle
\end{equation}
where $K_i= \langle i|U|0\rangle$ are Kraus operators acting on $\mathcal{H}_s$.

From equation \eqref{eq:total.state} it follows that the reduced system state of the system at time $n$ is given by 
$$
\rho^{\rm sys}_n := 
{\rm Tr}_{\rm out}
(|\Psi_n\rangle \langle \Psi_n|) = 
\mathcal{T}_*^n(\rho^{\rm sys}_{\rm in}), \quad \rho^{\rm sys}_{\rm in} = |\varphi\rangle\langle \varphi|,
$$
where the partial trace is taken over the output noise units, and 
$
\mathcal{T}_*:
L_1(\mathcal{H}_s)\to L_1(\mathcal{H}_s)
$ 
is the Markov transition operator 
(Schr\"{o}dinger picture)
$$
\mathcal{T}_*: \rho \mapsto \sum_{i\in \{0,1\}} 
K_i \rho K_i^*
$$
whose dual (Heisenberg picture) will be denoted by $\mathcal{T}$. Here $L_1 (\mathcal{H}_s)$ denotes the trace-class operators on $\mathcal{H}_s$, which in this case is isomorphic to $M(\mathbb{C}^d)$.

On the other hand, the reduced state of the output is 
\begin{eqnarray}
\rho^{\rm out}_n &:=& 
{\rm Tr}_{\rm sys} (|\Psi_n\rangle\langle \Psi_n|)
\label{eq:output.state}\\
&=& 
\sum_{{\bf i},{\bf j}\in \{0,1\}^n} 
\langle \varphi |
K^*_{\bf j}K_{\bf i} 
|\varphi\rangle \cdot
|{\bf i}\rangle
\langle {\bf j}|
\nonumber
\end{eqnarray}
where $K_{\bf i}:= K_{i_n}\dots K_{i_1}$ for ${\bf i}= (i_1,\dots i_n)$.

\begin{hypothesis} 
Throughout the paper we will assume that the dynamics is \emph{primitive} in the sense that $\mathcal{T}_*$ has a unique stationary state $\rho^{\rm ss}>0$ so that $ \mathcal{T}_*(\rho^{\rm ss})= \rho^{\rm ss}$ and it is \emph{aperiodic}, i.e. the only eigenvalue of $\mathcal{T}_*$ with unit absolute value is $1$.
\end{hypothesis}

In order to make the presentation of our results more accessible, we will show how the theory we develop applies to a simple example; we will consider a two dimensional system whose reduced dynamics (precession around $z$ axis with bit-flip noise) is given by the following Lindblad evolution:
\begin{equation} \label{eq:ctex}
\frac{d\rho_t}{dt}=-i\omega[\sigma_z,\rho_t]+\gamma(\sigma_x\rho_t \sigma_x-\rho_t),\end{equation}
where $\omega, \gamma>0$ and $\sigma_x$, $\sigma_z$ is the usual notation for $z$ and $x$ Pauli matrices. Both the frequency $\omega$ or the coupling strength $\gamma$ are natural parameters to estimate. In the input-output formalism, the joint system and output state is given by a continuous matrix product state; considering a small enough time discretization $dt$, one can approximate it by a state of the form as in equation \eqref{eq.Psi_n} where the Kraus operators are given by 
$$K_0=\sqrt{1-\gamma_t} e^{-i\omega_t \sigma_z}, \quad K_1=e^{-i\omega_t \sigma_z}\sqrt{\gamma_t}\sigma_x,$$
$\gamma_t=\gamma dt$, $\omega_t=\omega dt$
Both the continuous and the discrete time evolutions have the maximally mixed state $\mathbf{1}/2$ as unique stationary state.

\emph{Estimation of dynamical parameters}

We investigate the following quantum estimation problem: assuming that the dynamics depends smoothly on an unknown parameter 
$\theta\in \mathbb{R}$, we would like to estimate $\theta$ by performing measurements on the output state 
$\rho_n^{\rm out}$ generated after a number $n$ of interaction steps. In particular, we are interested in designing measurement strategies which achieve the highest possible precision, at least in the limit of large times. 

Let $\theta\mapsto U_\theta$ be a smooth map describing how the dynamics depends on an unknown parameter $\theta$, which is assumed to belong to an open bounded interval $\Theta$ of $\mathbb{R}$. We use similar notations $|\Psi_{\theta,n} \rangle, K_{\theta,i}, \mathcal{T}_\theta$ to denote the dependence on $\theta$ of the system-output state, Kraus operators, transition operator, etc. Two sequences of quantum statistical models indexed by time are of interest here: the system-output state 
$\mathcal{S}\mathcal{O}_n:= \{|\Psi_{\theta,n}\rangle :\theta\in \Theta\}$ defined in equation \eqref{eq:total.state} and the output state 
$\mathcal{O}_n:= \{\rho^{\rm out}_{\theta,n} : \theta\in \Theta\}$ defined in equation \eqref{eq:output.state}. While the former is more informative than the latter and easier to analyse, we are particularly interested in estimation strategies which involve only measurements on the output, hence the importance of the model $\mathcal{O}_n$. The following Theorem \cite{Guta_2015} shows that for primitive dynamics the QFI of both models scale linearly with $n$ with the \emph{same} rate, so having access to the system does not change the asymptotic theory. To simplify the expression of the QFI rate \eqref{eq:QFI.rate} we assume the following "gauge condition"
\begin{equation}
\label{eq:gauge}
\sum_j {\rm Tr} 
(\rho_\theta^{\rm ss}\dot{K}_{\theta,j}^*
K_{\theta,j}) =0.
\end{equation}
The condition \eqref{eq:gauge} means 
that $\sum_j 
\dot{K}_{\theta,j}^*
K_{\theta,j}$ belongs to the subspace $\{X: {\rm Tr}(\rho_\theta^{\rm ss}X) =0\} \subseteq B(\mathcal{H}_s)$ of bounded operators, on which 
the resolvent $\mathcal{R}_\theta:= ({\rm Id} - \mathcal{T}_\theta)^{-1}$ is well defined as the Moore-Penrose inverse. The condition can be satisfied by choosing the complex phase of the Kraus operators appropriately, or equivalently the phase of the standard basis in the noise unit space $\mathcal{H}_n$.
\begin{theorem}
\label{Th.QFI}
Consider a primitive discrete time Markov chain whose unitary $U_\theta$ depends smoothly on $\theta\in \Theta\subset \mathbb{R}$, and assume that condition \eqref{eq:gauge} holds true. The QFI $F^{\rm s+o}_n(\theta)$ of the system and output  state $|\Psi_{\theta,n}\rangle $ and the QFI $F^{\rm out}_n(\theta)$ of the output state $\rho^{\rm out}_{\theta, n}$ scale linearly with $n$ with the same rate: 
\begin{eqnarray}
\label{eq:QFI.rate}
&&\lim_{n\to\infty } \frac{1}{n} F^{\rm s+o}_\theta(n) =\lim_{n\to\infty } \frac{1}{n} F^{\rm out}_\theta(n) = f_\theta \\ 
&&=4\sum_{i=1}^k  {\rm Tr}\left[\rho^{\rm ss}_\theta\dot{K}_{\theta, i}^*\dot{K}_{\theta,i}\right] 
\nonumber\\ &&
+8\sum_{i=1}^k {\rm Tr} \left[{\rm Im} ( K_{\theta, i} \rho^{\rm ss}_\theta \dot{K}_{\theta,i}^* ) \cdot \mathcal{R}_\theta( {\rm Im} \sum_j\dot{K}_{\theta,j}^*
K_{\theta,j})\right]
\nonumber
\end{eqnarray}
where 
$\mathcal{R}_\theta$ is the Moore-Penrose inverse of ${\rm Id}- \mathcal{T}_\theta$.
\end{theorem}
In the following, we will always assume that $f_\theta>0$ for every $\theta \in \Theta$. In general, the classical Fisher information associated to simple repeated measurements (measuring the same observable on each output unit) does not achieve the QFI rate $f_\theta$. However, the class of available measurements can be enlarged by unitarily "post-processing" the output before performing a standard measurement, so that  effectively one measures the original output in a rotated basis. While this shifts the difficulty from measurement to "quantum computation", it turns out that the post-processing can be implemented with minimal computational cost by employing  the concept of a coherent absorber introduced in \cite{StannigelRAblZoller}. Indeed \cite{Godley2023} demonstrated that the QFI rate \emph{is} achievable by combining post-processing by a coherent absorber with a simple adaptive sequential measurement scheme. Furthermore, \cite{DayouCounting}, argued that one can also achieve the QFI by performing  simple counting measurements in the output, without the need for adaptive measurements. Our goal is to revisit this scheme and to provide a new, computationally and statistically effective estimation strategy.

\subsection{Quantum postprocessing with a coherent absorber}
\label{sec:CA}

The working of the coherent absorber 
is illustrated in Figure \ref{fig:intro} b).  Consider a QMC with a fixed and known  unitary $U$. After interacting with the system, each output noise unit interacts with a separate $d$ dimensional system $\mathcal{H}_a \cong \mathbb{C}^d$ (the coherent absorber), via a unitary $V$. The system and absorber can now be regarded as a single doubled-up system which interacts with the input via the unitary $ W:= V U$ on $\mathcal{H}_s\otimes \mathcal{H}_a \otimes \mathcal{H}_n$, where $U$ acts on first and third tensors and $V$ on second and third. The defining feature is that the system plus absorber have a pure stationary state. One can arrange this by requiring 
$$
VU: 
|\chi^{\rm ss}\rangle \otimes |0\rangle \mapsto 
|\chi^{\rm ss}\rangle \otimes |0\rangle
$$
where $|\chi^{\rm ss}\rangle\in \mathcal{H}_s\otimes \mathcal{H}_a$ is a purification of the system stationary state, i.e. $\rho^{\rm ss} = {\rm  Tr}_{\rm abs} ( |\chi^{\rm ss}\rangle\langle \chi^{\rm ss}|)$. This 
implies that in the stationary regime the output is decoupled from system and absorber, and is in the "vacuum" state $|0\rangle^{\otimes n}$. 
We briefly recall a few expressions related to the construction of $V$ that will be useful later on, and refer to Lemma 4.1 in \cite{Godley2023} for more details; for clarity, in the following we will use the labels $S,A,N$ to indicate system, absorber and noise unit.
Let us consider a spectral representation of $\rho^{\rm ss}$ and the corresponding purification:
$$\rho^{\rm ss}=\sum_{i=1}^{d}\lambda_i \ket{i_S} \bra{i_S}, \quad \ket{\chi^{\rm ss}}=\sum_{i=1}^{d}\sqrt{\lambda_i} \ket{i_S} \otimes \ket{i_A}.$$
For simplicity we assume that the eigenvalues $\lambda_i$ are strictly positive and are ordered in decreasing order;  
one can check that the following vectors are orthonormal:
\[
\ket{v_i}=\sum_{k=0}^{1}\sum_{j=1}^{d} \sqrt{\frac{\lambda_j}{\lambda_i}}\bra{i_S}K_k \ket{j_S} \ket{j_A} \otimes \ket{k_N}, \quad i=1,\dots, d.
\]
For any choice of $v_{d+1},\dots, v_{2d}$ such that $\{v_1,\dots, v_{2d}\}$ is an orthonormal basis for $\CC^d \otimes \CC^2$ (the Hilbert space corresponding to the absorber and the ancilla), a suitable choice for $V$ is given by
\[
V=
\mathbf{1}_S \otimes \left (
\sum_{i=1}^{d} |i_A \otimes 0_N\rangle \otimes \bra{v_i}+ 
\sum_{i=d+1}^{2d} 
|i_A \otimes 1_N\rangle \otimes \bra{v_i} 
\right ) .
\]
Note that $V$ is not uniquely defined: there is freedom in the spectral resolution of $\rho^{\rm ss}$ if there are degenerate eigenvalues and in picking $v_{d+1},\dots , v_{2d}$. 
The Kraus operators corresponding to the reduced dynamics $W= VU$ of the system and the absorber together are given by the following expression:
\begin{eqnarray*}
\tilde{K}_k: &=&
\langle k| W|0\rangle =
\sum_{l=0}^{1} \bra{k_N} V  \ket{l_N} \bra{l_N} U \ket{0_N}\\
&=& \sum_{l=0}^{1}  K_l \otimes V_{kl}
\end{eqnarray*}
where $\mathbf{1}_{S} \otimes V_{kl}:=\bra{k_N} V  \ket{l_N}$ and $K_l\otimes \mathbf{1}_{A}:=\bra{l_N} U \ket{0_N}$ are the Kraus operators of $V$ and 
$U$, respectively. 
The following Lemma prescribes the structure of the blocks $V_{kl}$. 
We first define the ``recovery'' channel \cite{Petz86,Petz88,Barnum02} with Kraus operators $K^\prime_i= \sqrt{\rho} K_i^*\sqrt{\rho^{-1}}$ and note that they satisfy the normalisation condition and the recovery channel has $\rho$ as invariant state.

\begin{lemma} \label{lemm:amkr}
The absorber operators $V_{kl}$ are of the following form. The blocks $V_{0l}$ are determined as 
$
V_{0l} = K_l^{\prime T} 
$
where the transpose is taken with respect to the eigenbasis of $\rho$. Assuming that 
$\mathbf{1} - |V_{00}|^2$ and $\mathbf{1} - |V_{01}|^2$ are strictly positive, then the $V_{1l}$ blocks are determined up to an overall arbitrary unitary $u$
$$
V_{10} = u |V_{10}| , V_{11}= uw |V_{11}|
$$
where   
$
|V_{10}| = \sqrt{\mathbf{1} - 
|V_{00}|^2}
$, 
$|V_{11}| = \sqrt{\mathbf{1} - |V_{01}|^2 }
$ are fixed, as well as the unitary 
$w= -|V_{10}|^{-1} V_{00}^* V_{01} 
|V|_{11}^{-1}
$.

\end{lemma}

{\it Proof.} 
From the definition of $|v_i\rangle$ we have
\begin{align*}
 V_{0k}  &=
 \sum_{i,j=1}^{d} \sqrt{\frac{\lambda_j}{\lambda_i}} \bra{j_S} K^*_k \ket{i_S} \ket{i_A} \bra{j_A} = K_k^{\prime T}.
\end{align*}
which proves the first statement.
From the fact that $V$ is unitary we obtain 
\begin{eqnarray*}
V_{00}^*V_{00} +V_{10}^*V_{10} &=& \mathbf{1}
\\
V_{01}^*V_{01} +V_{11}^*V_{11} &=& \mathbf{1}
\end{eqnarray*}
from which we get 
$$
|V_{10}| = \sqrt{\mathbf{1} - |V_{00}|^2}, \quad 
|V_{11}| = \sqrt{\mathbf{1} - |V_{01}|^2}
$$
which means that the absolute values of $V_{10}, V_{11}$ are fixed. 
Let $V_{10}= U_0 |V_{10}|$ and $V_{11}= U_1 |V_{11}|$ be their polar decompositions. Then from
$$
V_{00}^*V_{01} + V_{10}^*V_{11} =0
$$
we get $V_{10}^*V_{11}= -V_{00}^*V_{01}$ and 
$$
U_0^* U_1= |V_{10}|^{-1} V_{10}^*V_{11} |V_{11}|^{-1} = 
-|V_{10}|^{-1} V_{00}^*V_{01} |V_{11}|^{-1}
$$
which is a fixed unitary $w$. This proves the claim.

\qed

Later on we will require that the system and absorber transition operator $\tilde{\mathcal{T}}(\cdot) = \sum_k \tilde{K}^*_k \cdot \tilde{K}_k$ is primitive, in addition to the 
system's transition operator $\mathcal{T}$ satisfying the same property. At the moment we are not able to establish what is the connection between the two properties. However we performed extensive numerical simulations with randomly chosen QMC dynamics and corresponding absorbers which indicate that ``generically'' with respect to the original dynamics, for every primitive $\mathcal{T}$ there exists a corresponding absorber such that 
$\tilde{\mathcal{T}}$ is primitive. In fact, a stronger statement seems to hold, which is that the spectral gap of $\tilde{\mathcal{T}}$ is smaller or equal to that of $\mathcal{T}$ and one can always choose the absorber such that the two are equal. For more details on the absorber theory we refer to the recent paper \cite{Tsang24}.

Returning to the parameter estimation setting where $U= U_\theta$ depends on the unknown parameter $\theta$, we note that one cannot implement a coherent absorber which precisely matches 
the system dynamics. Instead, one can implement the absorber for an approximate value $\theta_0$ of $\theta$ and try to estimate the offset 
$\theta-\theta_0$ by measuring the output. This setting is closely related to that of displaced null measurements discussed in section \ref{sec:review.null-measurements}. Indeed, the joint state of system, absorber and output is pure for all $\theta$ and at $\theta=\theta_0$ it is of the product form 
$|\chi^{\rm ss}_{\theta_0} \rangle \otimes|0\rangle^{\otimes n}$, assuming that system and absorber are initially in the stationary state. Therefore, repeated standard basis measurements on the output units constitute a null measurement (in conjunction with a final appropriate measurement on system and absorber).

Once again, let us consider the discretisation of the simple continuous time example in equation \eqref{eq:ctex}; using Lemma \ref{lemm:amkr}, one obtains the following Kraus operators for the system and the absorber together (the first factor corresponds to the system and the second one to the absorber):
\begin{align*}
    \widetilde{K}_0=e^{-i\omega_t \sigma_z} \otimes e^{i\omega_t \sigma_z}(\gamma_t\mathbf{1}+(1-\gamma_t)\sigma_x \otimes \sigma_x), \\ \widetilde{K}_1=\sqrt{\gamma_t(1-\gamma_t)}e^{-i\omega_t \sigma_z}\otimes W (\mathbf{1}-\sigma_x \otimes \sigma_x),
\end{align*}
where $W$ is an arbitrary unitary operator.
One can check that the purification $\ket{\chi^{\rm ss}}=\frac{1}{\sqrt{2}}(\ket{00}+\ket{11})$ of the unique invariant state for the system satisfies $\widetilde{K}_0\ket{\chi^{\rm ss}}=\ket{\chi^{\rm ss}}$ and $\widetilde{K}_1\ket{\chi^{\rm ss}}=0$ (which implies that it is an invariant state for the dynamics of system and absorber together); 
In the continuous time counterpart, this would correspond to the following Hamiltonian $\widetilde{H}$ and jump operator $\widetilde{L}$ for the system and absorber dynamics:
$$\widetilde{H}=\omega(\sigma_z \otimes \mathbf{1}-\mathbf{1}\otimes \sigma_z), \quad \widetilde{L}=\sqrt{\gamma}\mathbf{1} \otimes W(\mathbf{1}-\sigma_x\otimes \sigma_x).
$$
By choosing $W= -\sigma_x$ we recover the absorber parameters prescribed in \cite{StannigelRAblZoller} and the system-absorber dynamics can be described as cascaded input-output dynamics.

The exact procedure for determining $\theta_0$ will be described in section \ref{sec:estimator} and follows the important displacement prescription outlined in section \ref{sec:review.null-measurements}. For the moment it suffices to say that $\theta_0$ will be informed by the outcome of a preliminary estimation stage involving simple (non-optimal) measurements on the output (without post-processing), and it will converge to $\theta$ in the limit of large $n$.  
While for $\theta_0=\theta$ the output state is the "vacuum", for $\theta\neq \theta_0$ the output could be seen as carrying a certain amount of "excitations" which increases with the parameter mismatch $|\theta-\theta_0|$. 

In section \ref{sec:TIM} we show how these "excitations" can be given a precise meaning by fashioning the output Hilbert space into a Fock space carrying modes labelled by certain "excitation patterns". In section 
\ref{sec:CLTPoisson} we show that from this perspective, the output state converges to a joint coherent state of the excitation pattern modes whose displacement depends linearly on $\theta-\theta_0$. This will allow us to devise a simple "pattern counting" algorithm for estimating $\theta$ in section \ref{sec:estimator}.

\section{Translationally invariant modes in the output}
\label{sec:TIM}

In this section we introduce the concept of translationally invariant modes (TIMs) of a spin chain. We show that in the limit of large chain size, certain translationally invariant states acquire the characteristic Fock space structure, and that the corresponding creation and annihilation operators satisfy the bosonic commutation relations. This construction will then be used in analysing the stationary Markov output state in section \ref{sec:CLTPoisson}.

Let $\left(\mathbb{C}^2\right)^{\otimes n}$ be a spin chain of length $n$ and let $|\Omega_n\rangle: = |0\rangle^{\otimes n}$ be the reference ``vacuum'' state. For every pair of integers $(k,l)$ with $1\leq l\leq k$ we define an \emph{excitation pattern} of length 
$k$ and number of excitations $l$ to be an ordered sequence
$\alpha := (\alpha_1, \dots \alpha_k) \in \{0,1\}^k$ such that $\alpha_1=\alpha_k =1$ and $\sum_{i=1}^k \alpha_i =l$. 
For instance, for $k=1,2$ the only patterns are $1$ and respectively $11$ while for $k=3$ the possible patterns are 
$111$ and $101$.

For each pattern $\alpha$ of length $k$ we define `creation and annihilation' operators
$$
A^*_{\alpha}(n) = \frac{1}{\sqrt{n}} \sum_{i=1}^{n-k+1} \sigma^{\alpha}_i, 
\quad A_{\alpha}(n)=\frac{1}{\sqrt{n}}\sum_{i=1}^{n-k+1}\sigma_i^{\alpha*}
$$
where $\sigma_i^{\alpha}=\prod_{j=0}^{k-1}\sigma_{i+j}^{\alpha_{j+1}}$, with $\sigma^0 := \mathbf{1}$, 
$\sigma^1:= \sigma^+ := |1\rangle\langle 0|$ and the index $i$ denotes the position in the chain.

In particular for $\alpha=1$ we have
\[A^*_1(n)=\frac{1}{\sqrt{n}}\sum_{i=1}^{n}\sigma_i^+, \quad A_1(n)=\frac{1}{\sqrt{n}}\sum_{i=1}^{n}\sigma_i^-.\]

We further define the ``canonical coordinates'' and ``number operator'' of the ``mode'' $\alpha$ as
\begin{eqnarray}
Q_{\alpha}(n)&=&\frac{A_{\alpha}(n)+A^*_{\alpha}(n)}{\sqrt{2}},\label{eq:quad.Q}\\
P_{\alpha}(n)& =& \frac{A_{\alpha}(n)-A^*_{\alpha}(n)}{\sqrt{2}i},\label{eq:quad.P}\\ 
N_{\alpha}(n)&=&A^*_{\alpha}(n)A_{\alpha}(n).
\label{eq:number.op}
\end{eqnarray}
We now introduce ``Fock states'' obtained by applying 
creation operators to the vacuum. 

Let $\mathcal{P}$ denote the ordered set of all patterns, where the order is the "natural" one inherited after identifying patterns with integer numbers using the binary representation. Let ${\bf n}:\mathcal{P}\to \mathbb{N}$ be pattern counts 
${\bf n} = (n_\alpha)_{\alpha\in \mathcal{P}}$ such 
that all but a finite number of counts are zero, and let 
${\bf n}!:= \prod_{\alpha\in\mathcal{P}}n_\alpha!$ and $|{\bf n}|:=\sum_{\alpha}n_{\alpha}$ the total number of patterns.

We define the approximate Fock state associated to the set of 
counts ${\bf n}$ as
\begin{equation}
\label{eq:Fock.state}
|{\bf n} 
; n\rangle:= 
\frac{1}{\sqrt{\bf n}!} 
\prod_{\alpha\in\mathcal{P}} 
A^*_{\alpha}(n)^{n_\alpha}
|\Omega_n\rangle , \quad n \geq 1,
\end{equation}
where the product is ordered according to the order on $\mathcal{P}$. For any fixed $n$ these vectors are not normalised or orthogonal to each other, and indeed they are not linearly independent since the Hilbert space is finite dimensional; however, Proposition \ref{prop:fock.states} shows 
that the expected Fock structure emerges in the limit of large 
$n$. Let us first illustrate this with a simple example. 
The state containing 2 patterns $\alpha =1$ is given by 
\begin{eqnarray*}
|n_1= 2;n\rangle &:=& 
\frac{1}{\sqrt{2}} \left(A_{1}^*(n) \right)^2 |\Omega_n \rangle \\
&=&\frac{1}{\sqrt{2}n} \sum_{i\neq j=1}^n 
|0\dots 0
10 \dots 01 0\dots 0\rangle
\end{eqnarray*}
where $i\neq j$ indicate the positions of the excitations, while the state containing a single $11$ pattern is 
\begin{eqnarray*}
|n_{11}=1;n\rangle &:=& 
A_{11}^*(n) |\Omega_n \rangle \\
&=&\frac{1}{\sqrt{n}} \sum_{i=1}^{n-1} 
|0\dots 0
11 0\dots 0\rangle
\end{eqnarray*}
Now it is easy to check that as $n\to \infty$
\begin{eqnarray*}
\langle n_1= 2;n| n_1= 2;n \rangle &= &
\frac{1}{2n^2}\frac{4n(n-1)}{2}\to 1,\\  
\langle n_{11}=1;n| n_{11}=1;n \rangle &= & \frac{n-1}{n} \to 1,\\
\langle n_1= 2;n| n_{11}=1;n \rangle &=& \frac{2(n-1)}{2\sqrt{n}n}= O\left(\frac{1}{\sqrt{n}} \right).
\end{eqnarray*}
This is generalised in the following Proposition which establishes the familiar structure of the bosonic Fock space in the limit of large $n$.

\begin{proposition}
\label{prop:fock.states}
Let $|{\bf n};n \rangle$ be the "Fock states" defined in equation \eqref{eq:Fock.state}. In the limit of large $n$ the "Fock states" become normalised and are orthogonal to each other
$$
\lim_{n\to\infty}\langle {\bf n};n| {\bf m} ;n\rangle =
\delta_{{\bf n}, {\bf m}} .
$$
Moreover, the order of the creation operators in \eqref{eq:Fock.state} becomes irrelevant in the limit of large $n$.
\end{proposition}

The proof of Proposition \ref{prop:fock.states} can be found in Appendix \ref{app:TIM}. From Proposition \ref{prop:fock.states} we obtain the following Corollary which shows that the action of the creation operators on the "Fock states" converges to that of a bosonic creation operator in the limit of large $n$.

\begin{corollary}
\label{prop:creation.action}
Let $\beta$ be a pattern and let $\bm{\delta}^{(\beta)}$ be the counts set with $\bm{\delta}^{(\beta)}_\alpha = \delta_{\alpha, \beta}$. Let $|{\bf n}; n\rangle$ be a "Fock state" as defined in equation \eqref{eq:Fock.state}.
In the limit of large $n$ the action of creation and annihilation operators $A^*_\beta(n)$ and $A_\beta(n)$ satisfy
\begin{eqnarray}
A^*_\beta(n)|{\bf n} ;n\rangle &=&
\sqrt{n_\beta +1} |{\bf n} + \bm{\delta}^\beta; n\rangle+o(1)\\
A_\beta(n)|{\bf n} ;n\rangle &=&
\sqrt{n_\beta } |{\bf n} - \bm{\delta}^\beta; n\rangle+o(1)
\end{eqnarray}
\end{corollary}
The proof of Corollary \ref{prop:creation.action} can be found in Appendix \ref{app:TIM}.

\section{Limit distribution of quadratures and number operators.}
\label{sec:CLTPoisson}

In this section we analyse the structure of the output state obtained by post-processing the output of a QMC with a coherent absorber, as described in section \ref{sec:CA}. Motivated by the fact that the statistical uncertainty scales as 
$1/\sqrt{n}$
we choose the absorber parameter 
$\theta_0$ to be fixed and known, and write the system parameter as
$\theta= \theta_0+u/\sqrt{n}$,
where $u$ is an unknown local parameter. For a more in depth motivation we refer to local asymptotic normality theory in appendix \ref{sec.LAN&DNM}.

Since the output state becomes stationary for long times, it is natural to focus on the state of the translationally invariant modes introduced in section \ref{sec:TIM}. In Theorem \ref{thm:limdistribution}
we show that asymptotically with $n$, the restricted state of these modes is a joint coherent state whose amplitude depends linearly on $u$, i.e. a Gaussian shift model. Moreover, Corollary \ref{coro:fisher} shows that the QFI of this model is equal to the QFI rate $f_{\theta_0}$ of the output state characterised in Theorem \ref{Th.QFI}. This means that the TIMs capture the entire QFI of the output. Together with Theorem \ref{thm:trajs} of section \ref{sec:limit.th.counting}, these  results will be the theoretical underpinning the  estimator proposed in section \ref{sec:estimator}.

We consider the system and absorber together as an open system with space $\mathcal{H}_{sa}\cong  \mathbb{C}^D$ with $D= d^2$, interacting with the noise units. The corresponding unitary is $W_\theta= V_{\theta_0}U_\theta$ where the absorber parameter is a fixed value $\theta_0$ (which later will be determined based on a preliminary estimation procedure) and $\theta$ is the unknown parameter to be estimated. We distinguish between the system Kraus operators $K_{\theta,i} =
\langle i|U_\theta |0\rangle\in B(\mathcal{H}_s)$ and the system and absorber Kraus operators 
$\tilde{K}_{\theta,i} = \langle i| W_\theta|0\rangle \in B(\mathcal{H}_{sa})$, and similarly between the system transition operator $\mathcal{T}_\theta$ and the system and absorber transition operator $\tilde{\mathcal{T}}_\theta (X) =\sum_{i=0}^1 \tilde{K}_{\theta,i}^*X\tilde{K}_{\theta,i} $. 
 
We will be interested in the probabilistic and statistical properties of the output state $\tilde{\rho}^{\rm out}_{\theta}$ of the system and absorber dynamics, for parameters $\theta$ in the neighbourhood of a given $\theta_0$. For clarity we state the precise mathematical properties we assume throughout.

\begin{hypothesis}\label{hyp:2}
The following properties of the system-absorber QMC are assumed to be true:
\begin{enumerate}
\item $\tilde{\mathcal{T}}_\theta$ has a unique invariant state $\tilde{\rho}^{\rm ss}_\theta>0$ and is aperiodic for $\theta$ in a neighborhood of $\theta_0$;
\item The Kraus operators $\tilde{K}_{\theta,i}$ and the stationary state $\tilde{\rho}^{\rm ss}_{\theta}$ are analytic functions of $\theta$ around $\theta_0$;
\item At $\theta_0$ the stationary state is pure 
$\tilde{\rho}^{\rm ss}_{\theta_0}=
\ket{\chi^{\rm ss}_{\theta_0}}\bra{\chi^{\rm ss}_{\theta_0}}$ and 
$\tilde{K}_{\theta_0,i} |\chi^{\rm ss}_{\theta_0}\rangle=(1-i)|\chi^{\rm ss}_{\theta_0}\rangle$ for $i=0,1$.
\end{enumerate}
\end{hypothesis}

To formulate the result we use the 
local parametrisation $\theta= \theta_0 + u/\sqrt{n}$
where $u$ is to be seen as a local parameter to be estimated from the output of length $n$. To simplify the notation, we denote the derivatives at $\theta_0$ as 
$$\dot{K}_i := \left. \frac{dK_{\theta,i}}{d\theta} \right|_{\theta_0},\quad \dot{\tilde{K}}_i := \left. \frac{d\tilde{K}_{\theta,i}}{d\theta} \right|_{\theta_0},
$$ and drop the subscript $\theta_0$ in $K_{\theta_0,i} =: K_i,$ $ \tilde{K}_{\theta_0,i} =: \tilde{K}_i$, etc.; we also use the local parameter instead of $\theta$ e.g. 
$K_{\theta_0+u/\sqrt{n},i} =: K_{u,i}$. 

Note that properties $1.$ and $2.$ in Hypothesis \ref{hyp:2} imply that $\dot{\tilde \rho}^{\rm ss}=\tilde{\mathcal{R}}_*\dot{\tilde{\mathcal{T}}}_*(\tilde{\rho}^{\rm ss})$ where $\tilde{\mathcal{R}}$ is the Moore-Penrose inverse of $ {\rm Id} - \tilde{\mathcal{T}}$. In addition, $\tilde{K}_0\ket{\chi^{\rm ss}}= \tilde{K}^*_0\ket{\chi^{\rm ss}} = \ket{\chi^{\rm ss}}$.
Indeed by construction $\tilde{K}_0|\chi^{\rm ss}\rangle =|\chi^{\rm ss}\rangle$ and 
$\tilde{K}_1|\chi^{\rm ss}\rangle =0$, and by applying  
$\tilde{K}_0^*\tilde{K}_0+ \tilde{K}_1^*\tilde{K}_1=\mathbf{1}$ to $|\chi^{\rm ss}\rangle$ we get $\tilde{K}_0^*|\chi^{\rm ss}\rangle = |\chi^{\rm ss}\rangle$.

Since for large $n$ the dynamics reaches stationarity, we consider the output state corresponding to the system starting in the stationary state $\tilde{\rho}^{\rm ss}_u$
\begin{equation}\label{eq:output.stat.local}
\tilde{\rho}^{\rm out}_{u,n}=
\sum_{{\bf i}, {\bf j}\in \{0,1\}^n}
{\rm Tr} \left[ 
\tilde{\rho}^{\rm ss}_u 
\tilde{K}^*_{u,{\bf j}}
\tilde{K}_{u,{\bf i}} \right]\ket{\bf i}\bra{\bf j}
\end{equation}
To formulate our results below we need to introduce several superoperators acting on the system and absorber space. For  $x\in B(\mathcal{H}_{sa})$, we define 
\begin{equation}\label{eq:aop}
{\cal A}_i(x)=\begin{cases} \tilde{\mathcal{T}}(x) & i=0 \\
\tilde{K}^*_1x\tilde{K}_0 & i=1 \\
\tilde{K}^*_0x\tilde{K}_1 & i=-1, \\\end{cases}\end{equation}
and
$$
\dot{{\cal A}}_1(x)=\dot{\tilde{K}}_1^*x\tilde{K}_0 + \tilde{K}_1^*x\dot{\tilde{K}}_0
$$

Furthermore, for every pattern  $\alpha$ of length $l$, we denote
\[{\cal A}^\alpha={\cal A}_{\alpha_1} \cdots {\cal A}_{\alpha_l} \text{ and } \tilde{{\cal A}}^{\alpha}=\dot{{\cal A}}_1{\cal A}_{\alpha_2} \cdots {\cal A}_{\alpha_l}.
\]
With a slight abuse of notation we will denote the expectation with respect to a density matrix $\rho$ as 
$\rho(X) = {\rm Tr}(\rho X)$.

\begin{theorem} 
\label{thm:limdistribution}
Let $\theta_0$ be a fixed parameter and assume that the dynamics satisfies the assumptions in Hypothesis \eqref{hyp:2}. Let $\theta=\theta_0 +u/\sqrt{n}$ be the system parameter with fixed local parameter $u$. Let $\alpha$ be a fixed pattern and let $z=\beta+i\gamma\in \mathbb{C}$ with $|z|=1$. Then the following convergence in distribution hold in the limit of large $n$ with respect to the output state $\tilde{\rho}^{\rm out}_{u,n} $ as defined in equation \eqref{eq:output.stat.local}.
\begin{itemize}
    \item[i)]
    \textbf{The quadratures \eqref{eq:quad.Q} and \eqref{eq:quad.P} of the TIM mode $\alpha$ satisfy the joint Central Limit Theorem}
\[\beta Q_\alpha(n)+\gamma P_\alpha(n) \xrightarrow{{\cal L}} N(u \mu_{\alpha,z}, 1/2),\]
where $N(\mu, V)$ denotes the normal distribution with mean $\mu$ and variance $V$ and  $\mu_{\alpha,z}= \sqrt{2}\Re (\bar{z} \mu_\alpha)$ with
\begin{equation}
\label{eq:mu_alpha}
\mu_{\alpha}= \tilde{\rho}^{\rm ss}((\dot{\tilde{\mathcal{T}}}\tilde{\mathcal{R}}{\cal A}^{\alpha}+\tilde{{\cal A}}^{\alpha})(\mathbf{1})).
\end{equation}
\item[ii)] \textbf{The number operator \eqref{eq:number.op} of the TIM mode $\alpha$ satisfies the Poisson limit:}
\[
N_\alpha(n) \xrightarrow{{\cal L}} {\rm Poisson}(u^2\lambda_\alpha),
\]
where
$\lambda_\alpha:=|\mu_\alpha|^2$.

\end{itemize}
\end{theorem}

The proof of Theorem \ref{thm:limdistribution} can be found in Appendix \ref{app:thmlimdistr}. 

We now provide a simpler expression for the parameters $\mu_\alpha$ in terms of Kraus operators and their first derivatives. As a by-product we show that the sum of all (limiting) QFIs of the Gaussian modes 
$(Q_\alpha(n), P_\alpha(n))$ is the QFI rate of the output state. This means that the TIMs capture all the QFI of the output state.

\begin{lemma} \label{lem:alternative}
    Let $\alpha$ be any excitation pattern and let $\mu_\alpha$ 
    be the constant defined in \eqref{eq:mu_alpha}.
    Then $\mu_\alpha$ can also be expressed as
\begin{eqnarray}
    \mu_\alpha &=& 
    \langle \tilde{K}_{\alpha_{|\alpha|}} \cdots \tilde{K}_{\alpha_1}(\mathbf{1}-\tilde{K}_0)^{-1}\dot{\tilde{K}}_0 \chi^{\rm ss} | \chi^{\rm ss}\rangle
    \nonumber\\
    &+&\langle \tilde{K}_{\alpha_{|\alpha|}} \cdots \tilde{K}_{\alpha_2}\dot{\tilde{K}}_1 \chi^{\rm ss}|\chi^{\rm ss}\rangle.
    \label{eq:mu_alpha2}
    \end{eqnarray}    
    Moreover with
    $\lambda_\alpha=|\mu_\alpha|^2$, one has
    \begin{equation}
    \label{eq:sum.lambda.alpha}
    \sum_{\alpha} \lambda_\alpha=\|(\tilde{K}_1(\mathbf{1}-\tilde{K}_0)^{-1}\dot{\tilde{K}}_0+\dot{\tilde{K}}_1)\chi^{\rm ss}\|^2.
    \end{equation}
\end{lemma}

The proof of Lemma \ref{lem:alternative} can be found in Appendix \ref{app:thmlimdistr}.

\begin{corollary} \label{coro:fisher}
Asymptotically with $n$, the total QFI of the TIMs is equal to the QFI rate of the output state \eqref{eq:QFI.rate}, that is
\[
4\sum_{\alpha}\lambda_\alpha=4 \|(\tilde{K}_1(\mathbf{1}-\tilde{K}_0)^{-1}\dot{\tilde{K}}_0+\dot{\tilde{K}}_1)\chi^{\rm ss}\|^2 =f_{\theta_0} .
\]
\end{corollary}
The proof of Corollary \ref{coro:fisher} can be found in Appendix \ref{app:fisher}. 

The upshot of this section is that (asymptotically in $n$) the statistical information of the output state is concentrated in the TIMs and the state's restriction to the TIM Bosonic algebra is a coherent state. Formally, in order to optimally estimate the parameter, one would only need to measure the appropriate quadrature of the Gaussian shift model, as explained in section \ref{sec:review.null-measurements}. However, it is not obvious how to perform such a measurement, and the theoretical insight does not seem to help on the practical side. Surprisingly, it turns out that the standard sequential counting measurement is an effective joint measurement of all the TIMs' number operators! This will be the main result of the next section, which in conjunction with the displaced null strategy discussed in section \ref{sec.LAN&DNM}, provides the ingredients of a counting-based estimation strategy.

\section{Limit theorem for counting trajectories}
\label{sec:limit.th.counting}

In this section we continue to investigate the probabilistic properties of the output state and consider the distribution of the stochastic process obtained measuring the output units sequentially in the canonical basis $\{\ket{0}, \ket{1}\}$. We consider the system and absorber dynamics with 
fixed absorber parameter $\theta_0$ and system parameter 
$\theta= \theta_0+u/\sqrt{n}$ for a fixed local parameter $u$. The  output state is given by equation \eqref{eq:output.stat.local}. The probability of observing a sequence $\omega =(\omega_1,\dots , \omega_n)\in \{0,1\}^{n}$ as the outcome of the first $n$ measurements is given by
\begin{equation}
\label
{eq:outcome.distribution}
\nu_{u,n}(\omega):=\rho^{\rm ss}_{u}\left ({\cal B}_{u,\omega_1} \cdots {\cal B}_{u,\omega_n}(\mathbf{1})\right )\end{equation}
where 
$$
{\cal B}_{u,j}(x)=\begin{cases} \tilde{K}^*_{u,1} x\tilde{K}_{u,1} & j=1 \\
\tilde{K}^*_{u,0}x\tilde{K}_{u,0} & j=0 \\\end{cases}.
$$
In order to state the main result of this section, we need to introduce a collection of events: first of all we define
$$B_{\bf 0}(n)=\{(0,\dots,0)\} \subset \{0,1\}^n.
$$
Let $\gamma$ be an arbitrary but fixed real number satisfying $0<\gamma<1$. Let ${\bf m} = \{m_\alpha\}_{\alpha\in \mathcal{P}}$ be a set of pattern excitation counts where all occupation numbers are zero except $(m_{\alpha^{(1)}}, \dots,m_{\alpha^{(k)}})$; we define $B_{\bf m}(n)$ as the set of all binary sequences of length $n$ containing $m_{\alpha^{(1)}}$ copies of $\alpha^{(1)}$s, up to $m_{\alpha^{(k)}}$ copies of $\alpha^{(k)}$s in any order, and such that between two consecutive patterns there are at least $n^\gamma$ $0$s. We remark that $B_{\bf m}(n)$ is the empty set for every $n$ strictly smaller than 
$\sum_{i=1}^{k}m_{\alpha^{(i)}} |\alpha^{(i)}|+(k-1)n^\gamma$ ($|\alpha|$ is the length of the pattern $\alpha$).

For instance, let us consider the set of pattern excitation counts $\mathbf{m}=(2_{(11)},1_{(101)})$, given by $2$ excitations of the type $(11)$ and $1$ of type $(101)$. The set $B_{\mathbf{m}}(n)$ is empty up to $n \geq 7+2\cdot n^\gamma$, then it contains all the strings of length $n$ of the following form:
\begin{align*}
    0\cdots 0\mathbf{101}\underbrace{0\cdots 0}_{n_1}\mathbf{11}\underbrace{0\cdots 0}_{n_2}\mathbf{11}0\cdots 0,\\
    0\cdots 0\mathbf{11}\underbrace{0\cdots 0}_{n_1}\mathbf{101}\underbrace{0\cdots 0}_{n_2}\mathbf{11}0\cdots 0,\\
    0\cdots 0\mathbf{11}\underbrace{0\cdots 0}_{n_1}\mathbf{11}\underbrace{0\cdots 0}_{n_2}\mathbf{101}0\cdots 0,
\end{align*}
where $0\cdots 0$ stays for a sequence of all $0$s and $n_2,n_3\geq n^\gamma$.

The following Theorem shows that the distribution of the 
pattern counts ${\bf m}$ converges to a product of Poisson distributions with the same intensities as those of the number operators of the TIM in Theorem \ref{thm:limdistribution}. This means that performing a standard output measurement and extracting the pattern counts provides an effective joint measurement of the number operators of the TIMs. This finding is essential in constructing an optimal estimator in section \ref{sec:estimator}.

\begin{theorem} \label{thm:trajs}
For every positive constant $C>0$ and finite collection of excitation patterns counts ${\bf m}=(m_{\alpha^{(1)}},\dots ,m_{\alpha^{(k)}})$,
the following limit is equal to zero:
\begin{equation} \label{eq:countlim}
\lim_{n \rightarrow \infty}\sup_{|u|<C}\left |\nu_{u,n}\left (B_{\bf m}(n)\right )-e^{-\lambda_{tot}u^2}\prod_{i=1}^{k} \frac{
\left(\lambda_{\alpha^{(i)}}u^2\right)^{m_{\alpha^{(i)}} }
}{m_{\alpha^{(i)}}!}\right |.
\end{equation}
where 
    \begin{eqnarray} 
    \lambda_{tot}& :=&-\frac{1}{2}\bra{\chi^{\rm ss}}(2\dot{{\cal B}}_{0*}{\cal R}_{0*} \dot{{\cal B}}_{0*}+\Ddot{{\cal B}}_{0*})(\ket{\chi^{\rm ss}}\bra{\chi^{\rm ss}})\ket{\chi^{\rm ss}}
    \nonumber\\
    &=&-\Re(\langle \chi^{\rm ss}, (2\dot{\tilde{K}}_0(\mathbf{1}-\tilde{K}_0)^{-1}\dot{\tilde{K}}_0+\Ddot{\tilde{K}}_0) \chi^{\rm ss} \rangle ).
    \label{eq:lambdatot}
    \end{eqnarray}
Moreover,
\[
\sum_{\alpha}\lambda_\alpha = \lambda_{tot}.
\]
\end{theorem}

The Proof of Theorem \ref{thm:trajs} can be found in Appendix \ref{app:thmtrjs}.

The previous result has some relevant consequences. Let us define the ``pattern extraction'' function 
which associates to each trajectory $\omega\in \{0,1\}^n$ a set of pattern counts $ \{ N_{\alpha,n}(\omega): \alpha\in \mathcal{P}\}\in \mathbb{N}^\mathcal{P}$, which is uniquely determined by the condition that $\omega$ is a maximal union of contiguous patterns separated by sequences of $0$s of length at least $n^\gamma$ with a fixed $0<\gamma<1$; moreover, let us consider the stochastic process given by the infinite collection of independent random variables $\{N_\alpha: \alpha \in \mathcal{P}\}$ where $N_{\alpha}$ is a Poisson random variable with parameter $\lambda_\alpha u^2$.
\begin{corollary} \label{coro:lconv}
  For every $u \in \mathbb{R}$ the law of the stochastic process $\{ N_{\alpha,n}: \alpha\in \mathcal{P}\}$ under the measure $\nu_{u,n}$ converges to the one of $\{N_{\alpha}:\alpha \in \mathcal{P}\}$. Moreover, for every $\alpha \in \mathcal{P}$, $p\geq 1$ one has
  $$\lim_{n\rightarrow +\infty}\mathbb{E}[N_{\alpha,n}^p]=\mathbb{E}[N_{\alpha}^p].$$
\end{corollary}
The Proof of Corollary \ref{coro:lconv} can be found in Appendix \ref{sec:lconvapp}.

However, in the following we will be interested in local parameters with growing size, i.e. $|u| \leq n^{\epsilon^\prime}$ for some $0 <\epsilon^\prime <1/2$. In this case we can show the following result.
\begin{proposition} \label{prop:grweps}
 For $0 <\epsilon^\prime <1/6$ and for every finite collection of excitation patterns counts ${\bf m}$ the following holds true:
\begin{equation} \label{eq:grweps}
\lim_{n \rightarrow +\infty}\sup_{|u| \leq n^{\epsilon^\prime} } \left | \frac{\nu_{u,n}\left (B_{\mathbf{m}}(n)\right )}{ e^{-\lambda_{tot}u^2}\prod_{i=1}^{k} \frac{\lambda_{\alpha^{(i)}}^{m_i}u^{2m_i}}{m_i!}}-1 \right |=0.
\end{equation}  
\end{proposition}
The Proof of Proposition \ref{prop:grweps} can be found in Appendix \ref{app:thmtrjs} as well. Upgrading this result to a  weak convergence one similar to Corollary \ref{coro:lconv} remains the subject of future research.
\section{Pattern counting estimator}
\label{sec:estimator}

In this section we describe our adaptive estimation scheme which exploits the asymptotic results presented in sections \ref{sec:CLTPoisson} and \ref{sec:limit.th.counting}. The scheme involves four key ingredients: 
\begin{itemize}
    \item[i)] 
    perform a simple output measurement (no absorber) to compute a preliminary estimator;
    
    \item[ii)] set the absorber parameter by using the displaced-null measurement technique developed in \cite{GiGoGu} and run the system-absorber dynamics for the remainder of the time;
    \item[iii)] perform a sequential counting measurement in the output and extract counts for the TIMs from the outcomes trajectory;
    \item[iv)]
    construct a simple estimator expressed in terms of total counts of patterns for different TIM modes.
    
\end{itemize}

The first step of the adaptive protocol is to use $\tilde{n} := n^{1-\epsilon}\ll n$ output units to produce a rough preliminary estimator 
$\tilde{\theta}_n$. This can be done by performing a repeated standard basis measurement on the output (without using an absorber). 
Typically, the estimator
$\tilde{\theta}_n$ will have variance scaling with the standard rate $\tilde{n}^{-1}= n^{-1+2\epsilon}$, and $\tilde{n}^{1/2}(\tilde{\theta}_n-\theta)$ will satisfy the central limit theorem and a concentration bound ensuring that $|\tilde{\theta}_n-\theta| =O(n^{-1/2+\epsilon})$ with high probability. 
For instance, one can define 
$\tilde{\theta}_n$ by equating the empirical counting rate with the theoretical rate $n_\tau:= {\rm Tr}(\rho^{\rm ss}_{\tau} K_{\tau, 1}^*K_{\tau,1})$, see \cite{GutaCB15,GiGaGu23}. 

In the second step we set the absorber at a parameter value $\theta_0$ and run the system-absorber quantum Markov chain for the reminder of the time $n^\prime = n-\tilde{n}$. 
The naive choice for $\theta_0$ is our best guess $\tilde{\theta}_n$ about $\theta$, based on the first stage measurement. However, with this choice, the counting measurement suffers from the non-identifiability issue described in section \ref{sec:review.null-measurements}. This can be resolved by further displacing the absorber parameter by an amount 
$\delta_n:= \tau_n/\sqrt{n} $ where $\tau_n= n^{3\epsilon}$ so that 
$\theta_{\rm abs}:= 
\tilde{\theta}_n - \delta_n$. As usual we write 
$\theta = \tilde{\theta}_n+u_n/\sqrt{n}$ where $u_n$ is a local parameter satisfying $|u_n|\leq n^{\epsilon}$, so that 
$\theta= 
\theta_{\rm abs}
+(u_n+\tau_n)/\sqrt{n}$.

In the third step we perform \emph{standard basis measurements} in the output of the modified dynamics which includes the absorber. For simplicity, in the discussion below we ignore the fact that we have $n^\prime = n-\tilde{n}$ rather than $n$ output units, which does not affect the error scaling of the estimator. Let $\omega= (\omega_1,\dots , \omega_n)$ be the measurement outcome with distribution \eqref{eq:outcome.distribution}, where the local parameter $u_n$ is replaced by 
$u_n+\tau_n$ to take into account the displacement.

We now describe the construction of the estimator from the outcomes of this last stage measurement. The estimator will be built using the ``pattern extraction'' function $ \{ N_{\alpha,n}(\omega): \alpha\in \mathcal{P}\}\in \mathbb{N}^\mathcal{P}$ that we defined in the previous section; we recall that it
associates to each trajectory $\omega\in \{0,1\}^n$ a set of pattern counts, which is uniquely determined by the condition that $\omega$ is a maximal union of contiguous patterns separated by sequences of $0$s of length at least $n^\gamma$ with a fixed $0<\gamma<1$. This means that the algorithm will not detect any pattern which contains a sequence of zeros of length larger than $n^\gamma$, since this would be seen as being made up of several identified patterns.

For illustration, if we pick $\gamma \leq 0.41$, the sequence in Figure \ref{fig:intro} d) would result in the following pattern counting: $3$ patterns of the form $(1)$, $2$ patterns of the form $(11)$ and one pattern of the types $(101)$ and $(111)$ (the patterns are colored in red in the picture). Indeed, the sequence has $72$ bits, therefore we would perform a `cut' after 6 consecutive $0$s since $72^\gamma \approx5.7$.

Since $\theta-\theta_{\rm abs}$ is of order $n^{-1/2+ 3\epsilon}$, the expected number of patterns is of order $n^{6\epsilon}$  and the gaps between them are expected to be of size $n^{1-6\epsilon}$. Therefore, we choose the buffer parameter such that $0<\gamma< 1-6\epsilon$, which is always possible as long $\epsilon<1/6$. 

We now introduce the final estimator using an intuitive argument based on extrapolating the results of Theorem \ref{thm:trajs} from fixed to slowly growing local parameters $u_n+\tau_n$. This means that $N_{\alpha,n}(\omega)$ is approximately distributed as ${\rm Poisson}((u_n+\tau_n)^2 |\mu_\alpha|^2)$, for large $n$; since $\tau_n= n^{3\epsilon}$ is larger than $u_n= O(n^\epsilon)$, the intensity of the Poisson distribution diverges with $n$, and the distribution can be approximated further by the normal  
$N((u_n+\tau_n)^2|\mu_\alpha|^2, (u_n+\tau_n)^2|\mu_\alpha|^2 ) $ with the same mean and variance. Using 
\begin{eqnarray*}
\frac{1}{\tau_n}(u_n+\tau_n)^2 &=&
2u_n + \tau_n + o(1) \\  
\frac{1}{\tau^2_n}(u_n+\tau_n)^2 &=& 1+ o(1)
\end{eqnarray*}
we obtain that 
$$
Y_{\alpha,n} := \frac{1}{|\mu_\alpha|}\left(\frac{N_{\alpha,n}}{\tau_n }-\tau_n |\mu_\alpha|^2 \right) 
$$
has approximate distribution $N(2u_n|\mu_\alpha| , 1)$. A simple computation shows that the optimal estimator of $u_n$ based on the (approximately) normal variables $Y_{\alpha,n}$ is the linear combination
\begin{equation}
\label{eq:u.hat.final}
\hat{u}_n:= Y_n: = \frac{2}{ f_{\theta_{\rm abs}}}\sum_\alpha |\mu_\alpha| Y_{\alpha,n} =\frac{2}{f_{\theta_{\rm abs}}\tau_n}\sum_\alpha N_{\alpha,n} - \frac{\tau_n}{2}
\end{equation}
where $f_{\theta_{\rm abs}}= 4\sum_\alpha |\mu_\alpha|^2$ is the quantum Fisher information rate of the output by Corollary \ref{coro:fisher}. Note that $Y_n$ depends only on the \emph{total number of patterns} of the trajectory $\omega$, not to be confused with the total number of $1$s. 

Since $Y_{\alpha,n}$ is approximately normal with distribution $N(2u_n|\mu_\alpha| , 1)$, we obtain that $\hat{u}_n$ has approximate distribution $N(u_n, f_{\theta_{\rm abs}}^{-1})$. 
The final estimator of $\theta$ is 
\begin{equation}
\label{eq:theta.hat.final}
\hat{\theta}_n:=\tilde{\theta}_n+\frac{\hat{u}_n}{\sqrt{n}}.
\end{equation}
and it attains the QCRB in the sense that 
$\sqrt{n}(\hat{\theta}_n-\theta)$ converges in distribution to $N(0,f_  \theta^{-1})$.

At the moment we only have a rigorous proof of this statement assuming a stronger version of Proposition \ref{prop:grweps}, which we were not able to obtain; however, we point out that Theorem 2 in \cite{GiGoGu} establishes a similar optimality result in the case of multi-parameter estimation with independent, identical copies.

Before proceeding, we need to introduce some more notations. Let us consider the following collection of random variables:
$$\overline{N}_{n,\theta,\tilde{\theta}}\sim\text{Poisson}(\lambda_{\rm tot}(\tilde{\theta})(\sqrt{n}(\theta-\tilde{\theta})+\tau_n)^2),$$
where $n \in \mathbb{N}$, and $\tilde{\theta}$, $\theta \in \Theta$. We recall that $\tau_n=n^{3\epsilon}$ is the displacement size. Let us define
$$\overline{Y}_{n,\theta,\tilde{\theta}}:=\frac{1}{2\tau_n \lambda_{\rm tot}(\tilde{\theta})}
\overline{N}_{n,\theta,\tilde{\theta}}
- \frac{\tau_n}{2}.$$
\begin{theorem} \label{thm:opt}
Let fix $\theta \in \Theta$ at which $f_\theta$ is continuous and let $\tilde{\theta}_n$ be a preliminary estimator which uses $\tilde{n}:=n^{1-\epsilon}$ samples with $\epsilon$ small enough, such that it satisfies the concentration bound
\begin{equation}\label{eq:concentration}
\mathbb{P}_\theta(| \tilde{\theta}_n-\theta |>n^{-1/2 +\epsilon})
\leq C e^{-n^{\epsilon} r}
\end{equation}
for some constants $C,r>0$. Let 
\[
\hat{\theta}_n:=\tilde{\theta}_n+\frac{Y_n}{\sqrt{n}}  ,
\]
be the final estimator as defined in \eqref{eq:u.hat.final} and \eqref{eq:theta.hat.final}.

If for every $a \in \mathbb{R}$ one has
\begin{equation} \label{eq:hypo}
\lim_{n \rightarrow +\infty}\sup_{\tilde{\theta}:|\tilde{\theta}-\theta|<n^{-1/2+\epsilon}} |\mathbb{E}_\theta[e^{iaY_n}|\tilde{\theta}_n=\tilde{\theta}]-\mathbb{E}[e^{ia\overline{Y}_{n,\theta, \tilde{\theta}}}]|=0,
\end{equation}
then $\hat{\theta}$ is asymptotically optimal and asymptotically normal, i.e. the following convergence in law holds for large $n$
\[
\sqrt{n}(\hat{\theta}_n-\theta) \xrightarrow[]{{\cal L}_\theta} {\cal N}\left (0, \frac{1}{f_\theta} \right ).
\]
\end{theorem}

The proof of the following theorem and the relationship between the extra hypothesis in Eq. \eqref{eq:hypo} and Proposition \ref{prop:grweps} can be found in Appendix \ref{app:finalestim}. In few words, the difference between the statement of Proposition \ref{prop:grweps} and the hypothesis in Eq. \eqref{eq:hypo} is the following: Eq. \eqref{eq:grweps} concerns the asymptotic behaviour of the probability of observing a finite number of patterns, while, in order to prove Theorem \ref{thm:opt}, one would need the asymptotic behaviour of the probability of observing a collections of patterns with a growing number of patters.

In the following section we illustrate our method with results of numerical simulations on a qubit model.

\section{Numerical experiments}
\label{sec:numerics}
In this section we illustrate the estimation protocol through numerical simulations, using a qubit model inspired by the previous work \cite{GiGoGu}.
The simulations are implemented in Python using the QuTiP package \cite{Qutip}.

The quantum Markov chain model consists of a two-dimensional system coupled to two-dimensional noise units by a unitary $U_\theta$ with unknown parameter $\theta \in \mathbb{R}$, where the noise units are all prepared in the same initial state $| 0\rangle$. Since the system interacts with a fresh noise unit at each step, it suffices to specify the action of  $U_\theta$ on the states $\ket{00}$ and $\ket{10}$, and we define
\begin{eqnarray*}
\label{eqn:sims.U}
    U_\theta : \ket{00} &\to& \cos(\theta)  \sqrt{1-\theta^2} \ket{00} \nonumber \\
    &&+ i \sin(\theta) \sqrt{1-\theta^2} \ket{10} + \theta \ket{11}, \nonumber \\
    U_\theta : \ket{10} &\to& i \sin(\theta)  \sqrt{1-\lambda} \ket{00} \nonumber \\
    &&+ \cos(\theta) \sqrt{1-\lambda} \ket{10} + \sqrt{\lambda} e^{i \phi} \ket{01}, \nonumber
\end{eqnarray*}
where $\lambda, \phi$ are known parameters. In  simulations we  used $\phi=\pi/4, \lambda=0.8$ and $\theta=0.2$ for the true values of the parameter.

In a first simulation study we verify the predictions made in Corollary \ref{coro:lconv}. We run the dynamics for a total of $n=6\times 10^5$ time steps, with absorber parameter 
$\theta_{\rm abs}= \theta_0$ (cf. section \ref{sec:CA}) and system parameter $\theta= \theta_0+ u/\sqrt{n}$ such that $\theta=0.2$ and the local parameter is $u =2$.

For each run we perform repeated measurements in the standard basis of the output units, to produce a measurement record 
$\omega =(\omega_1, \dots, \omega_n)$. For each such trajectory, the pattern counts $N_{\alpha, n}(\omega)$ are obtained by identifying patterns (sequences starting and ending with a 1) which are separated by at least $n^\gamma$ $0$s and no pattern contains more than $n^\gamma$ $0$s, where $\gamma>0$ is a small parameter. This can be done by combing through the sequence and identifying occurrences of a given such pattern padded by $n^\gamma$ $0$s to the left and right (taking care of the special case of the first and last patterns). Note that for any given $n$ this procedure will not count patterns with more than  $n^\gamma$ successive $0$s. Since the mean counts for each pattern $\alpha$ is $|\mu_\alpha|^2 u^2$ and $|\mu_\alpha|^2$ decays exponentially with $|\alpha|$, we find that patterns of such length are unlikely to occur for large $n$. The results of $N=2000$ independent repetitions of the experiment are illustrated in Figure \ref{fig:counts.no.initial} which shows a good match between the counts histograms corresponding to several patterns (in blue) and the theoretical Poisson distributions 
${\rm Poisson}(|\mu_\alpha|^2 u^2)$ (in orange) , as predicted by Corollary \ref{coro:lconv}.

\begin{figure}[h]
    \centering
    \includegraphics[width=\linewidth]{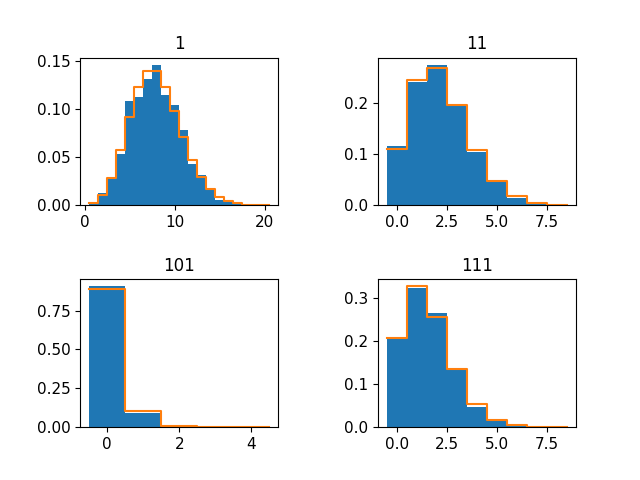}
    \caption{ (Blue) Counts histograms for patterns $\alpha=1, 11,101, 111$ from N=2000 trajectories.  (Orange line) The Poisson distribution with intensity given by $|\mu_\alpha|^2 u^2$ matches well the empirical counts distribution, as prescribed by Corollary \ref{coro:lconv}.}
    \label{fig:counts.no.initial}
\end{figure}

In a second simulation study we use system parameter
$\theta=0.2$ and set the absorber at $\theta_{\rm abs} = \theta-\delta_n$, with displacement $\delta_n= n^{-1/2}\tau_n$ for $n=6\times 10^5$ and 
$\tau_n = 7$. We perform the same measurement as above and extract the pattern counts $N_{\alpha,n}(\omega)$ for each trajectory $\omega$. We then use the pattern counts estimator \eqref{eq:theta.hat.final} to estimate $\theta$, taking 
$\tilde{\theta}_n=\theta$. This amounts to assuming that the first stage of the general estimation procedure outlined in section \ref{sec:estimator} gives a perfect estimator, which is then used in setting the absorber parameter in the second step. While this procedure cannot be used in a practical situation, the study has theoretical value in that it allows us to study the performance of the pattern counts estimator in \emph{its own right}, rather than in conjunction with the first step estimator. Figure \ref{fig:est.no.initial} shows that the final estimator 
$\hat{\theta}_n$ has Gaussian distribution with variance closely matching $1/(n f_\theta)$, thus achieving the QCRB in this idealised setup. This can also be seen by comparing the ``effective'' Fisher information 
$ F_{\rm eff}:= (n \overline{(\hat{\theta}_n-\theta)^2})^{-1}$ where the mean square error is estimated from the data, with the QFI rate $f_\theta$; the former is equal to $F_{\rm eff} = 13.8$ while the latter is $f_\theta = 13.5$.

\begin{figure}[h!]
    \centering
    \includegraphics[width=\linewidth]{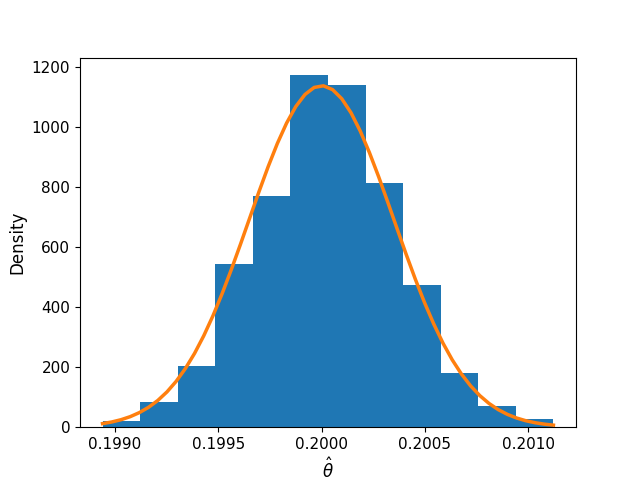}
    \caption{(Blue) Histogram of the final estimator  $\hat{\theta}_n$ from N=1000 trajectories with no first stage estimation ($\theta=\tilde{\theta}_n$). The effective Fisher information (inverse of rescaled estimated variance)   
    $F_{\rm eff}\approx13.8$ matches closely the QFI $f_\theta = 13.5$. (Orange line) For comparison we plot the density of the normal distribution with mean 
    $\bar{\theta}=0.2$ and variance 
    $\sigma^2 =  (n f_\theta)^{-1}$.
    } 
     \label{fig:est.no.initial}
\end{figure}

In the third simulations study we implement the full estimation procedure described in section \ref{sec:estimator} including the first stage estimator. In the first step we fix $\theta =0.2$, and run the Markov chain (without the absorber) for $\tilde{n}= 4\times 10^5$. We then perform sequential measurements in the standard basis on the output noise units to obtain a measurement trajectory 
from which we compute $\tilde{\theta}_n$ by equating the empirical mean (the average number of $1$ counts)
with the stationary expected value 
 $c_\tau: = {\rm Tr} (\rho^{\rm ss}_\tau K_{\tau,1}^* K_{\tau,1})$, where $K_{\tau,i}$ are the system's Kraus operators. The effective Fisher information of this estimator is 
 $F_{\rm eff} = 4.06$, significantly lower than the QFI $f_\theta =13.5$.

We then set the absorber to 
$\theta_{\rm abs} = \tilde{\theta}_n - n^{-1/2}\tau_n$ where $\tau_n = 25.5 $ and $n=6.6\times 10^6$ and we run the system and absorber chain for $n^\prime = n-\tilde{n}$ steps. The system and absorber  is initialised in the pure stationary state corresponding to $\tilde{\theta}_n$, but since the system and absorber dynamics is assumed to be primitive, any other initial state will 
result in equivalent asymptotic results. We then perform the counting measurement as in the previous simulation studies to obtain a trajectory $\omega\in \{0,1\}^{n\prime}$ and extract the pattern counts $N_{\alpha,n^\prime}(\omega)$.

From these average counts we compute the estimator \eqref{eq:theta.hat.final}, where the Fisher information $f$ is computed at $\tilde{\theta}_n$ and multipled by $n^\prime/n$ to account for the smaller number of samples used in the last step. The results are illustrated in Figure \ref{fig:est.full.protocol} which compares the histogram of $\hat{\theta}_n$ (in blue) with the density of a normal distribution with mean $\theta$ and variance $1/(n f_\theta)$. We find that there is a good fit with the normal distribution but less accurate that that of the second simulation study, cf. Figure \ref{fig:est.no.initial}. This is expected, since the final estimator is based on a two stage estimation process and does not use any prior information about $\theta$. A more accurate measure of the protocol's performance is given by the  effective Fisher information which works out as $F_{\rm eff} = 10.8$ compared to the QFI rate $f_\theta=13.5$, while the effective Fisher information of the first stage was only $4.06$. These simulation results are in agreement with the theoretical arguments put forward in section \ref{sec:estimator} which indicate that the two stage estimator attains the QCRB asymptotically.

\begin{figure}
    \centering
    \includegraphics[width=\linewidth]{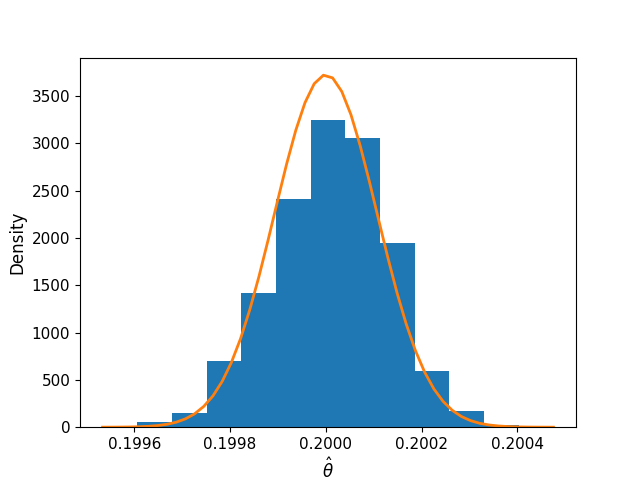}
    \caption{(Blue) Histogram of final estimator $\hat{\theta}_n$ from N=1070 trajectories; the effective Fisher information is $F_{\rm eff}=10.8$ compared to the QFI 
    $f_\theta =13.5$. (Orange line) The density of the normal distribution with mean $\bar{\theta}=0.2$ and variance $\sigma^2 = (nf_\theta)^{-1}$}.
     \label{fig:est.full.protocol}
\end{figure}

      \label{fig:final.estimator}
\section{Conclusions and Outlook}

In this paper we developed a computationally efficient strategy to estimate  dynamical parameters of a quantum Markov chain and provided strong theoretical evidence that the estimator achieves the quantum Cram\'{e}r-Rao bound in the large time limit. In addition, we established asymptotic results pertaining to the mathematical structure of the output state which are of more general interest.

The estimation strategy consisted of two estimation stages. In the first stage a rough  estimator is computed from outcomes of simple output measurements by using a fast but non-optimal procedure, e.g. equating the empirical counts with the expected value at the estimated parameter. In the second stage, we used a coherent quantum absorber \cite{StannigelRAblZoller} to post-process the output \cite{Godley2023,DayouCounting}. When tuned to the true value of the system parameter, the absorber ``reverts'' the evolution such that, in the stationary regime, the post-processed output is identical to the input ``vacuum'' product state. On the other hand, small  mismatches between system and absorber parameters lead to slight rotations away from this product state, which can be detected by simple sequential measurements on the output ``noise units''.
To achieve the perfect ``null'' setup it would seem natural to use the first stage estimator as absorber parameter, but as shown in \cite{GiGoGu}, this leads to non-identifiability issues and sub-optimal final estimators. Instead, we applied the displaced-null technique \cite{GiGoGu} which prescribes an extra parameter shift, calibrated to 
remove the non-identifiability issue while preserving the optimality of the sequential measurement.

The key theoretical contributions of this work are related to the understanding of the output state and the stochastic measurement process. We introduced the concept of translationally invariant modes (TIMs) of the output and showed how they generate a Bosonic algebra in the asymptotic limit. Each mode is labelled by a binary sequence called a ``pattern''
and its creation operator is an average of shifted blocks consisting of tensor products operators. 
We then showed that when the mismatch between system and absorber parameter scales at the estimation rate $n^{-1/2}$, the restriction of the output state to the TIMs becomes a multi-mode coherent state whose displacement depends linearly on the mismatch (quantum Gaussian shift model). Moreover, we showed that the TIMs carry all the quantum Fisher information of the output state and are therefore the relevant quantum statistics of the problem. While homodyne is the standard optimal measurement for such models, 
in the presence of the additional parameter displacement the 
modes amplitudes become large and counting measurements become effectively equivalent to homodyne. Surprisingly, we discovered that the sequential counting measurement acts as a joint measurement of all TIMs number operators. Due to the proximity to the vacuum state, typical trajectories consist of a relatively small number of patterns separated by long sequences of $0$s. We showed that for large times the patterns counts distribution converges to the Poisson distribution of the TIMs coherent state. These insights allowed us to devise a simple ``pattern counting'' estimator for estimating the unknown parameter.

Our discrete-time results open the way for fast and optimal continuous-time estimation strategies based on coherent absorber post-processing \cite{DayouCounting} and simple pattern counting estimation, as opposed to expensive maximum likelihood methods. Interesting and important topics for future investigations concerns the robustness of the pattern counting method with respect to various types of noises, improving the estimation accuracy for short times, extensions to multi-parameter models and the relationship between the general theory of quantum absorbers \cite{Tsang24} and quantum estimation.

\emph{Acknowledgements.} This work was supported by the EPSRC grant EP/T022140/1. MG thanks Juan Garrahan, Mankei Tsang, Dayou Yang, Martin Plenio and Susana Huelga for fruitful discussions. The Python code and  the simulation data used in this paper can be found at \url{https://www.doi.org/10.5281/zenodo.16755954} 

\appendix

\section
{Displaced-null measurements and asymptotic normality}
\label{sec.LAN&DNM}

In this section we review the concept of \emph{quantum local asymptotic normality} (QLAN) in its simplest form involving a one-parameter family of pure states, and use it to better understand the asymptotic theory of SLD and displaced-null measurements. For the general QLAN theory we refer to \cite{LAN1,LAN2,LAN3,LAN4,LAN5,LAN6} and \cite{Yamagata13,Fujiwara20,Fujiwara22}. For our purposes, QLAN expresses the fact that the statistical model $|\psi_\theta\rangle^{\otimes n}$ describing an ensemble of $n$ identically prepared systems can be approximated by a simpler quantum Gaussian model consisting of a \emph{coherent state} of a continuous variables (CV) system, whose mean is related to $\theta$ by a linear transformation. Below, we will show that the SLD measurement and the null measurement in the multi-copy model correspond to measuring canonical variables and respectively displaced number operators in the Gaussian picture.

More precisely, let  
$\theta_0$ be a fixed parameter value and write $\theta= \theta_0 + u/\sqrt{n}$ where $u$ is a \emph{local parameter} describing the deviation from $\theta_0$. The latter can be thought of as being the first stage estimator $\tilde{\theta}_n$, and will be used as such later on, but for the moment is considered to be fixed and known. 
To simplify the presentation we assume that $\langle \psi_{\theta_0} |\dot{\psi}_{\theta_0}\rangle= 0$, which can always be satisfied by appropriately choosing the phase of $|\psi_\theta\rangle$ around $\theta_0$. Let $\{|0\rangle, \dots , |d-1\rangle\} $ be an orthonormal basis such that 
$|0\rangle = |\psi_{\theta_0}\rangle$. 
Then in the first order of approximation we can write the state as 
\begin{eqnarray*}
|\psi_\theta\rangle&=&
|0\rangle  + \frac{u}{\sqrt{n}} |\dot\psi_{\theta_0}\rangle
+O(n^{-1})\\
&=&
|0\rangle  + \frac{u}{\sqrt{n}} 
\sum_{j=1}^{d-1} c_j |j\rangle + O(n^{-1}), \quad
c_j = \langle j |
\dot\psi_{\theta_0}
 \rangle .
\end{eqnarray*}
The SLD is given by 
\begin{equation}
\label{eq:SLD.pure}
\mathcal{L}_{\theta_0}  = 
2(|\psi_{\theta_0} \rangle \langle \dot\psi_{\theta_0} | + 
|\dot \psi_{\theta_0} \rangle \langle \psi_{\theta_0} |) 
\end{equation}
and the QFI is 
$F_{\theta_0}= 4\|\dot \psi_{\theta_0}\|^2 = 4\|{\bf c}\|^2$, where ${\bf c}= 
(c_1,\dots, c_{d-1})\in \mathbb{C}^{d-1}$.

We now construct a Gaussian model which approximates the multiple copies model in the neighbourhood of $\theta_0$. Let 
$\mathcal{F}_{d-1}=\mathcal{F}^{\otimes{(d-1)}}$ be the Fock space of $d-1$ CV modes with annihilation operators 
$(a_1,\dots ,a_{d-1})$, and let $|{\bf n}\rangle = \otimes_{j=1}^{d-1} |n_j\rangle$ be the Fock basis, where ${\bf n}= (n_1, \dots , n_{d-1})\in \mathbb{N}^{d-1}$. We further denote by 
$|{\bf z}\rangle = \otimes_{j=1}^{d-1} |z_i\rangle$ a coherent state with amplitude ${\bf z}= (z_1,\dots , z_{d-1})\in \mathbb{C}^{d-1}$ such that 
$\langle {\bf z} |a_j|{\bf z}\rangle =z_j$. On the Fock space we consider the quantum statistical model with parameter 
$u\in \mathbb{R}$ consisting of coherent states
\begin{equation}
\label{eq:Gaussian.shift}
|{\bf c}u\rangle = \otimes_{j=1}^{d-1} |c_j u\rangle. 
\end{equation}
We now show (in three different ways) that the multiple copies model  
$|\psi^n_u\rangle := |\psi_{\theta_0+u/\sqrt{n}}^{\otimes n}\rangle$, which capture the local properties of the original model around $\theta_0$ in terms of the rescaled parameter $u$, is approximated by the Gaussian model $|{\bf c}u\rangle$. Firstly, one can show that the collective variables 
\begin{eqnarray}
Q^n_j&: =& \frac{1}{\sqrt{2n}}
\sum_{i=1}^n{(|0\rangle\langle j| + |j\rangle\langle 0|)_{j}}
\nonumber\\
P^n_j&: =& \frac{1}{\sqrt{2n}}
\sum_{i=1}^n{( -i|0\rangle\langle j| +i |j\rangle\langle 0|)_{j}}
\label{eq:CLT}
\end{eqnarray}
 converge in distribution to the canonical variables $Q_j, P_j$ of the CV Gaussian model, with respect to the state $|{\bf c}u\rangle$. Secondly, the convergence holds on the geometric level, i.e. it can be expressed in terms of the overlaps 
$$
\lim_{n\to\infty}
\langle \psi^n_u| \psi^n_v\rangle =
\langle {\bf c}u | {\bf c}v\rangle.
$$
where $u,v$ are fixed local parameters. The reader will recognise that for $d=2$ these statements are statistical reformulations of results from coherent spin states \cite{Radcliffe}.

The third type of QLAN convergence is akin to the Le Cam strong convergence of statistical models in classical statistics \cite{Vaart1998}; it states that the multi-copy and Gaussian models can be mapped into each other operationally by means of quantum channels. Let 
$\mathcal{S}_{n}$ be the symmetric subspace of 
$\left(\mathbb{C}^d\right)^{\otimes n}$ and let $V_n: \mathcal{S}_{n}\to \mathcal{F}_{d-1}$ be the isometry 
$$
|{\bf n}; n\rangle \mapsto |{\bf n}\rangle
$$
where $|{\bf n}; n\rangle$ is the normalised projection onto ${\cal S}_n$ of the basis vector $\left (\otimes_{j=1}^{d-1} |j\rangle^{\otimes n_j} \right )\otimes|0\rangle^{n-|{\bf n}|}\in \left(\mathbb{C}^d\right )^{\otimes n}$. Then the following holds for any $0<\gamma<1$ \cite{LAN6}
$$
\lim_{n\to \infty}
\sup_{|u|\leq n^{1-\gamma}} \| V_n |\psi^n_u\rangle  - |{\bf c}u\rangle \| =0
$$
which means that the original model can be mapped isometrically (with no loss of information) into the Gaussian model with vanishing asymptotic error, uniformly over all the relevant local parameters.

From the estimation perspective, QLAN can be used to devise an asymptotically optimal, two steps measurement strategy for the multi-copy model. 
We first compute a preliminary estimator $\tilde{\theta}_n$ by using a small sub-sample and set $\theta_0 =\tilde{\theta}_n$. Subsequently, we map the ensemble state into the (approximately) Gaussian one by applying the isometry $V_n$, and then measure the SLD ${\bf a}{\bf c}^\dagger + {\bf c}{\bf a}^\dagger$ of the Gaussian shift model, which is the limit of $\frac{1}{\sqrt{n}} \sum_{i=1}^n\mathcal{L}^{(i)}_{\theta_0}$ where $\mathcal{L}_{\theta_0}$ is the SLD \eqref{eq:SLD.pure}. Joining the two estimation stages, we compute the final estimator as
$\hat{\theta}_n = \tilde{\theta}_n +\hat{u}_n/\sqrt{n}$ where 
$\hat{u}_n = X/F_{\tilde{\theta}_n}$ and $X$ is the result of the second stage measurement. Similarly to  \eqref{eq:two.step.SLD}, the estimator
$\hat{\theta}_n$ defined above is optimal in the sense that it achieves the QCRB asymptotically, cf. equation \eqref{eq:asymptotic.achievability.QCRB}.

We now move away from SLD measurements and consider the QLAN perspective on displaced-null measurements. An important point here is that the parameter $u$ of the Gaussian shift model $|{\bf  c}u\rangle$ is not completely arbitrary but can be considered to be bounded by $n^{\epsilon}$, where $n$ is the sample size of the qudit ensemble. Indeed, recall that 
the reference parameter $\theta_0$ should be thought of as a preliminary estimator $\tilde{\theta}_n$ obtained by measuring a sub-sample of $\tilde{n}= n^{1-\epsilon}$ qudits. 
Assuming this has standard concentration properties i.e. $|\theta -\tilde{\theta}_n|< n^{-1/2 +\epsilon}$ with high probability, and writing $\theta= \tilde{\theta}_n+u/\sqrt{n}$ we obtain $|u|\leq n^{\epsilon}$ \cite{GiGoGu}.

While the collective SLD observable in the qudit model maps onto a quadrature of the limit model, a measurement in the null basis 
$\{ |e_0 \rangle \equiv|\psi_{\theta_0} \rangle, \dots ,|e_{d-1}\rangle \}$ corresponds to a simultaneous measurement of the number operators $N_j = a_j^*a_j$, for all the CV modes $a_j$. More precisely, in the limit of large $n$ the joint distribution of the counts $\{N_{j,n}\}_{j=1}^{d-1}$ for outcomes corresponding to vectors $\{|e_j\rangle\}_{j=1}^{d-1}$, converges to the joint distribution of the number operators $\{N_j = a_j^* a_j\}_{j=1}^{d-1}$ with respect to the coherent state $|{\bf c}u\rangle$, which is the product of independent Poisson distributions 
$ {\rm Poisson} (|c_j|^2 u^2)$. Since the latter depends on $u^2$, the local parameter is non-identifiable so the measurement does not distinguish the parameters $\theta_{\pm}:= \theta_0\pm u/\sqrt{n}$. Consequently, any final estimator has an error of the same order as the statistical error of the initial estimator, and is therefore far from optimal. Let us now change the reference point to 
$\theta_0 = \tilde{\theta}_n - \delta_n$, with $\delta_n = \tau_n/\sqrt{n}$ and $\tau_n = n^{3\epsilon}$. From this vantage point the limit model is 
 $|{\bf c}(u+ \tau_n)\rangle$, and measuring in the displaced-null basis $\{ |e_0 \rangle \equiv|\psi_{\theta_0} \rangle, \dots |e_{d-1}\rangle \}$ is again asymptotically equivalent to measuring the number operators $\{N_j\}_{j=1}^{d-1}$, with the difference that the limiting distribution is the product of ${\rm Poisson}(|c_j|^2(u+\tau_n)^2)$. Since $|u|<n^{\epsilon}$ (with high probability) due to the preliminary estimation step, and $\tau_n\gg n^{\epsilon}$ this means that $u$ is uniquely determined by the measurement distribution.
 Moreover, for large $n$ the Poisson distribution can be approximated by the Gaussian $N(\lambda_i^n, \lambda_i^n )$ where 
 $\lambda_i^n= |c_i|^2(u+\tau_n)^2$. By expanding $(u+\tau_n)^2/\tau_n$ and neglecting $u^2/\tau_n= O(n^{-\epsilon})$ we obtain that the rescaled variable
 $$
Y_{i,n} = \frac{1}{|c_i|}
\left(\frac{N_{i,n}}{\tau_n} -\tau_n |c_i|^2\right)
 $$
converges in  distribution to $N(2u|c_i|, 1)$. A simple signal to noise analysis shows that the best estimator of $u$ is the linear combination
$$
\hat{u}_n= \frac{1}{2\|{\bf c}\|^2}\sum_{j=1}^{d-1} |c_j| Y_{j,n} = \frac{1}{2\tau_n \|c\|^2}\sum_{j=1}^{d-1}N_{j,n} - \frac{\tau_n}{2}
$$
Its limiting distribution is $N(u, F_\theta^{-1})$ where $F_\theta= 4\|{\bf c}\|^2$ is the quantum Fisher information of the original model. In particular
the final estimator $\hat{\theta}_n =\tilde{\theta}_n+ \hat{u}/\sqrt{n}$ achieves the QCRB in the sense of \eqref{eq:asymptotic.achievability.QCRB}.

\section{Proofs of Proposition \ref{prop:fock.states} 
and Corollary \ref{prop:creation.action}} 
\label{app:TIM}

\emph{Proof of Proposition \ref{prop:fock.states}.} Let $\bm{\alpha} = (\alpha^{(1)},\dots, \alpha^{(p)}), \bm{\beta}=(\beta^{(1)}, \dots ,\beta^{(q)})$ be the patterns with non-zero counts for ${\bf n}$ and respectively ${\bf m}$; let $(n_1, \dots, n_p)$ and   $(m_1,\dots, m_q)$ be their counts, and let $M$ be the maximum length of all patterns in $\bm{\alpha}$ and $\bm{\beta}$. 

Since pattern creation operators $A_{\alpha}^*(n)$ involve sums of $\sigma^\alpha_i$ for different positions $i$, the Fock state 
$|{\bf n} ;n\rangle$ is a superposition of vectors obtained by applying the following type of \emph{ordered} products to the vacuum
$$
\prod_{k=1}^{n_p} \sigma^{\alpha^{(p)}}_{i_{p,k}}\dots 
\prod_{k=1}^{n_1} \sigma^{\alpha^{(1)}}_{i_{1,k}}
$$
where $i_{j,k}$ is the index marking the location of the left end of the $k$-th pattern $\alpha^{(j)}$, with $j\in \{1,\dots, p\}$ and $k\in \{1,\dots , n_j\} $. In general, some of these operators may `overlap' (act on the same spin) and the computation of the superposition becomes cumbersome.
However, since the total length of the patterns is finite, for large n, the main contribution in the superposition comes from arrangements in which there are no overlapping patterns. Even more, we can restrict to arrangements where the patterns are separated by at least $M$ zeros, in which case each pattern  in the sequence of zeros and ones can be identified unambiguously. 

Therefore
\begin{eqnarray}
&&|{\bf n};n\rangle 
\nonumber \\
&&=
\frac{1}{\sqrt{n^{|{\bf n}|}}}
\frac{1}{\sqrt{{\bf n}!}}
\sum_{{\bf i} \in\mathcal{I}({\bf n};n)} \, \prod_{j=1}^p 
\, \prod_{k=1}^{n_j} 
\, \sigma^{\alpha^{(j)}}_{i_{j,k}} 
|\Omega_n  \rangle + o(1) 
\nonumber\\
&&=
\frac{1}{\sqrt{n^{|{\bf n}|}}}
\frac{1}{\sqrt{{\bf n}!}}
\sum_{{\bf i} \in\mathcal{I}({\bf n};n) } 
|\omega({\bf i}, {\bf n})\rangle
+ o(1)
\label{eq:Fock.state1}
\end{eqnarray}
where ${\bf n}!:=n_1!\dots n_p!$ and  $\mathcal{I}({\bf n};n)$ is the subset of locations ${\bf i}=\{ i_{j,k}\}$ leading to non-overlapping patterns with counts set ${\bf n}$, such that all patterns are at distance of at least $M$ from each other, and $\omega({\bf i}, {\bf n})\in \{0,1\}^n$ is the basis vector (trajectory) obtained by placing the all the patterns corresponding to the counts set ${\bf n}$ at locations prescribed by ${\bf i}$. More precisely ${\bf i}\in \mathcal{I}({\bf n};n)$ if for any two pairs $(j,k)$ and $(\tilde{j}, \tilde{k})$ with 
$i_{j,k}\leq i_{\tilde{j},\tilde{k}}$ one has 
    $i_{j,k}+ |\alpha^{(j)}|+M-1 <i_{\tilde{j}, \tilde{k}}$. The basis vector $|\omega({\bf i}, {\bf n})\rangle$ consist of zeros except patterns $\alpha^{(j)}$ written at positions $i_{j,k}$ for $j\in \{1,\dots ,p\}$ and $k\in \{1,\dots , n_j\}$.

    To show that the remainder term is $o(1)$ note that any term in the superposition $|{\bf n } ,n \rangle $ corresponding to a specific product of $\sigma$s is a vector of the standard basis, and any such vector 
    $|\omega\rangle$ has at most a fixed number $M\cdot |{\bf n}|$ of $1$s. Note also that the action of applying a pattern 
    $\sigma^{\alpha}_i$ to the vacuum cannot be reversed by applying subsequent $\sigma$s, since these contain only creation operators. This means that for a given 
    $|\omega\rangle$, the locations ${\bf i} =\{i_{j,k}\}$ of the $\sigma$s  producing this vector are limited to an area of size $2 M$ around each $1$ in the sequence, or in other words, the coefficient 
    of $|\omega\rangle$ is bounded by $(2 M)^{M|{\bf n}|}$. On the other hand, the number of basis vectors 
    $|\omega\rangle$ obtained by applying patterns such that at least two overlap or are at distance smaller than $M$ from each other is $o(n^{|{\bf n}|})$ since the number of possible locations for two such $\sigma$s is $O(n)$. Therefore, the remainder term in \eqref{eq:Fock.state1} is $o(1)$.

Using equation \eqref{eq:Fock.state1} we obtain
\begin{eqnarray*}
&&\langle{\bf m};n |{\bf n};n\rangle=\\
&&\frac{1}{\sqrt{n^{|{\bf n}|} \cdot n^{|{\bf m}|}\cdot {\bf n}! \cdot {\bf m}!}}
\sum_{{\bf i} , \tilde{\bf i} }
\langle \omega({\bf i}, {\bf n})|\omega(\tilde{\bf i}, {\bf m}) \rangle +o(1)
\end{eqnarray*}
where the sum runs over 
${\bf i}\in \mathcal{I}({\bf n};n)$ and 
$\tilde{\bf i}\in \mathcal{I}({\bf m};n)$. Since each ${\bf i}$ and $\tilde{\bf i}$ uniquely determines the patterns it contains, the basis vectors $|\omega({\bf i}, {\bf n})\rangle$  and 
$|\omega(\tilde{\bf i}, {\bf m}) \rangle$ have non-zero overlap (coincide) only if their sets of patterns coincide, 
i.e. ${\bf n}= {\bf m}$. Therefore, if ${\bf n}\neq {\bf m}$  hold then
$$
\lim_{n\to\infty}
\langle{\bf m};n |{\bf n};n\rangle =0.
$$

On the other hand,
if ${\bf m} = {\bf n}$ then
$$
\lim_{n\to\infty}
\langle{\bf n};n |{\bf n};n\rangle = 
\lim_{n\to\infty}
\frac{1}{n^{|{\bf n}|} {\bf n}!} ({\bf n}!)^2
\frac{|\mathcal{I}({\bf n},n)|}{{\bf n}!} =1
$$
where we took into account that each basis 
vector $|\omega ({\bf i}, {\bf n})\rangle$ appears ${\bf n}!$ times in the sum over ${\bf i}\in \mathcal{I}({\bf n};n)$, and that
$$
\lim_{n\to \infty} 
\frac{|\mathcal{I}({\bf n} , n)|}{n^{|\bf n|} } =1 .
$$
\qed

\emph{Proof of Corollary \ref{prop:creation.action}.} As shown in the proof of Proposition \ref{prop:fock.states} the "Fock state" $|{\bf n}; n\rangle$ can be approximated by a superposition of basis states in which the patterns in ${\bf n}$ are non-overlapping and are situated at least at a certain distance from each other. Moreover the remainder term contains $o(n^{|{\bf n}|})$ basis vectors with bounded coefficients. After applying $A^*_\beta(n)$, the multiplicity of each basis vector $|\omega\rangle$ is finite and the number of possible basis vectors is 
$o(n^{|{\bf n}|+1})$. Taking account of the factor $n^{-(|{\bf n}|+1)/2}$ we find that the action of $A^*_\beta(n)$ on the 
$o(1)$ remainder in \eqref{prop:fock.states} is still $o(1)$. On the other hand, the action on the main term in \eqref{prop:fock.states} is to add a pattern $\beta$ separated by the other patterns by ${\rm max}(M, |\beta|)$, with a negligible term coming from locations in which $\beta$ overlaps with or is too close to one of the existing patterns. The factor $\sqrt{n_\beta +1}$ comes from the definition of the "Fock state" $|{\bf n}+ \delta^{(\beta)};n\rangle$.

The action of $A_\beta(n)$ on one of the vectors $|\omega({\bf i}, {\bf n})\rangle$ is to produce a superposition of basis vectors in which the pattern 
$\beta$ has been removed from the set 
${\bf n}$ of patters, at all possible locations. This may include removing part of a existing pattern $\alpha^{(j)}$ which coincides with $\beta$. These two case will produce orthogonal vectors and can be evaluated separately. In the first case, the pattern $\beta$ is removed from one of the locations where such pattern existed. 

In the Fourier decomposition of 
$\sqrt{n_\beta}|{\bf n}-\bm{\delta}^\beta; n\rangle$, 
a basis vector 
$|\omega\rangle$ given by non-overlapping patterns has coefficient 
$$
c(\omega) = 
\sqrt{n_\beta}\frac{\bf n!}{n_\beta} \frac{\sqrt{n\cdot n_\beta}}{\sqrt{n^{|{\bf n}|}{\bf n}!}} =
{\bf n}! \frac{\sqrt{n}}{\sqrt{n^{|{\bf n}|}{\bf n}!}}
$$
Such a vector can be obtained in approximately $n$ ways by removing a pattern $\beta$ from a basis vector which appears in the decomposition of 
$|{\bf n}; n\rangle$. Therefore its coefficient in $A_{\beta} |{\bf n}, n\rangle$ is approximately
$$
\frac{{\bf n}!}{\sqrt{n}}\cdot n\cdot  \frac{1}{\sqrt{n^{|{\bf n}|} {\bf n}!}}
$$
where we took into account a factor $1/\sqrt{n}$ from the definition of $A_\beta$. We therefore obtain that the coefficients of the non-overlapping terms in $A_\beta |{\bf n}; n\rangle $ and 
$\sqrt{n_\beta} |{\bf n} -\bm{\delta}^\beta; n\rangle$ agree asymptotically. The fact that the $o(1)$ terms remain small after applying $A_\beta$ can be shown similarly to the above.

\qed

\section{Proofs of Theorem \ref{thm:limdistribution} 
and Lemma \ref{lem:alternative}} \label{app:thmlimdistr}

\begin{proof}[Proof of Theorem \ref{thm:limdistribution}]
In order to prove the theorem, we will make use of the method of moments; this can be done since the moment problem corresponding to the moments of Poisson and Gaussian random variables admits a unique solution (this can be seen for instance using Cram\'er condition, see \cite{Li17,Sc20}).

\underline{\textbf{Number operators.}} We will show the convergence of the moments of $N_\alpha(n)$ in the state $\rho^{\rm out}_{u,n}$ to those of a Poisson random variable of intensity $\lambda(u):=u^2 \lambda_\alpha$. We recall the the $r$-th moment of a Poisson random variable of intensity $\lambda(u)$ is equal to $\sum_{m=1}^r S(r,m) \lambda(u)^m$ where $S(r,m)$ are the Stirling numbers of the second type: a combinatorial interpretation of $S(r,m)$ is the number of partitions in $m$ non-empty subsets of a set of cardinality $r$. 

Let us focus on the expression of the $r$-th moment of $N_\alpha(n)$. For simplicity we 
denote $\langle X\rangle_u := \tilde{\rho}^{\rm out}_{u,n}(X)$. From \eqref{eq:number.op} we have
\begin{equation}\label{eq:rmom}
\langle N_\alpha(n)^r \rangle_{u}=\frac{1}{n^r} \sum_{\substack{i_1,\dots,i_r=1 \\ j_1,\dots,j_r=1}}^{n-|\alpha|+1}\langle \sigma^\alpha_{i_1}\sigma^{\alpha*}_{j_1}\cdots\sigma^\alpha_{i_r}\sigma^{\alpha*}_{j_r}  \rangle_{u}.
\end{equation}

\textbf{Splitting of the sum based on non-overlapping groups of $\sigma$s.} Let us consider a term $\sigma^\alpha_{i_1}\sigma^{\alpha*}_{j_1}\cdots\sigma^\alpha_{i_r}\sigma^{\alpha*}_{j_r} $ and represent each $\sigma^\alpha_{i}$ or $\sigma^{\alpha*}_{i}$  as a block of length $|\alpha|$ covering positions $\{i, i+1, \dots i+|\alpha|-1\}$ of the string $\{1,2,\dots n\}$.
Depending on the overlapping pattern of the blocks, the indices can be split (uniquely) in a number $s= s(i_1, j_1,\dots ,i_r, j_r)$ of groups ($1 \leq s \leq 2r$) such that blocks in different groups do not overlap, and each group cannot be split into further non-overlapping sub-groups. Among these groups we identify  
$g =g(i_1, j_1,\dots ,i_r, j_r) $ special groups characterised by the fact that they are made up of one or more pairs of blocks of the type $(\sigma_{z}^{\alpha*}\sigma^\alpha_{z})$ for some $z$. 
We call ${\cal P}_0$ the set of such groups. Note that not all groups associated to a product of $\sigma$s can be in $\mathcal{P}_0$ because the order of $\sigma_{z}^{\alpha*}$ and $\sigma_{z}^{\alpha}$ in the latter is opposite to that in which such terms appear in $N_\alpha(n)$.

For example, consider the pattern $\alpha=$"$11$" for $n=10$ and the term
$$
\sigma^{\alpha}_8\sigma_1^{\alpha *} \sigma^{\alpha}_6 \sigma^{\alpha*}_8 \sigma^{\alpha}_1 \sigma_2^{\alpha^*} \sigma_2^{\alpha}\sigma_5^{\alpha*} 
$$
This has $3$ non-overlapping groups of operators (that commute with each others): the first one
$$\sigma^{\alpha*}_1\sigma_1^{\alpha} \sigma^{\alpha*}_2  \sigma_2^{\alpha}
$$
belongs to ${\cal P}_0$, while the second and  third ones 
$$
\sigma^\alpha_6\sigma^{\alpha*}_5
\qquad \text{and}\qquad \sigma^\alpha_8 \sigma_8^{\alpha*}
$$
are not in ${\cal P}_0$. Therefore $s=3$ and $g=1$.

Let us now look at the expected value of a given product of $\sigma$s. We have
\begin{equation}
\label{eq:exp.cs&ts}
\langle \sigma^\alpha_{i_1}\sigma^{\alpha*}_{j_1}\cdots\sigma^\alpha_{i_r}\sigma^{\alpha*}_{j_r}  \rangle_{u}=
\rho^{\rm ss}_{u}\left ({\cal C}_{u,1} \tilde{\mathcal{T}}^{x_1}_{u}\cdots \tilde{\mathcal{T}}^{x_{s-1}}_{u}{\cal C}_{u,s} (\mathbf{1})\right )
\end{equation}
where $x_i$'s are the distances between the non-overlapping groups and ${\cal C}_{u,i}$ is the map corresponding to the $i$-th group, which is computed according to the following rule. In  every group we multiply the operators on a given position (we may have several $\sigma^+, \sigma^-$ or $\sigma^0$ on one position) and the result will be an element of the set $\mathcal{O}:= \{\sigma^+, \sigma^-, \sigma^0, \ket{0}\bra{0}, \ket{1}\bra{1}, {\bf 0}\}$. For example the first group above gives
$$
\sigma^{\alpha*}_1\sigma_1^{\alpha} \sigma^{\alpha*}_2  \sigma_2^{\alpha} 
=(\ket{0}\bra{0})_1(\ket{0}\bra{0})_2(\ket{0}\bra{0})_3
$$
while
$$
\sigma^\alpha_6\sigma^{\alpha*}_5
=\sigma^-_5(\ket{1}\bra{1})_6 \sigma^+_7 
$$
and 
$$
 \sigma^\alpha_8 \sigma_8^{\alpha*}
=(\ket{1}\bra{1})_8(\ket{1}\bra{1})_9.
$$
Suppose the result of this computation is 
$O^1_{i}\cdots O^k_{i+k-1}$, with $O^k\in \mathcal{O}$; then in the expectation, this translates into a  superoperator 
${\cal C}_{u} = {\cal C}_u[O^1]\circ \dots \circ {\cal C}_u[O^k]$ obtained by composing in the same order basic maps ${\cal C}_u[O]$ defined as follows
\begin{align*}
&{\cal C}_u [\sigma^+]:= {\cal A}_{u,1}, \quad {\cal C}_u [\sigma^0]:= {\cal A}_{u,0}, \quad {\cal C}_u [\sigma^-]:= {\cal A}_{u,-1},\\ &{\cal C}_u [\mathbf{0}]:= 0,
{\cal C}_u [\ket{j}\bra{j}]:= {\cal B}_{u,j}, ~  j\in\{0,1\}.
\end{align*}
where ${\cal A}_{u,i}$ are defined as in equation \eqref{eq:aop} by replacing $\tilde{K}_i$ with $\tilde{K}_{u,i}$
and
\[{\cal B}_{u,j}(x)=\begin{cases} \tilde{K}^*_{u,1} x\tilde{K}_{u,1} & j=1 \\
\tilde{K}^*_{u,0}x\tilde{K}_{u,0} & j=0 \\\end{cases}.\]
In what follows we will drop the label $u$ when $u=0$.


We can now compute the expectation of our example product of $\sigma$s as
\[\begin{split}   
&\langle \sigma^{\alpha}_8\sigma_1^{\alpha *} \sigma^{\alpha}_6 \sigma^{\alpha*}_8 \sigma^{\alpha}_1 \sigma_2^{\alpha^*} \sigma_2^{\alpha}\sigma_5^{\alpha*} \rangle_u=\\
&\rho^{\rm ss}_{u} \left ( {\cal B}_{u,0}^{3} \tilde{\mathcal{T}}_{u}{\cal A}_{u,-1}{\cal B}_{u,1}{\cal A}_{u,1}{\cal B}^2_{u,1}(\mathbf{1})\right ).
\end{split}\]

Let us define the space
\[
\mathbf{1}^\perp:=\{x \in B(\mathcal{H}_{sa}): \tilde{\rho}^{\rm ss}(x)=0\},
\]
and note that the following properties hold true:
\begin{itemize}
    \item ${\cal A}_0$ and ${\cal B}_0$ leave $\mathbf{1}^\perp$ invariant;
    \item the image of ${\cal A}_i$ and ${\cal B}_j$  for $i,j \neq 0$ is contained in $\mathbf{1}^\perp$;
    \item for every $m \geq 1$, the norm of ${\cal A}^m_0=\tilde{\mathcal{T}}^m$ restricted to $\mathbf{1}^\perp$ is less or equal than $c\lambda^m$ for some $c>0$ and $1>\lambda>0$. 
    
\end{itemize}

The reason why we singled out groups in ${\cal P}_0$ is given by the following key observation: such groups are \emph{the only ones} for which  the corresponding map ${\cal C}_{u,i}$ is only composed by ${\cal A}_{u,0}$s and ${\cal B}_{u,0}$s. Any other group will contain at least one $\mathcal{A}_{u,\pm 1}$ or $\mathcal{B}_{u,1}$ factor.  From this and the properties above it follows that if $\mathcal{C}$ is not in $\mathcal{P}_0$ then 
\begin{equation}
\label{eq:exponential.decay.TC}
\|\tilde{\mathcal{T}}^m \circ \mathcal{C}\|\leq c\lambda^m.
\end{equation}

We will often identify the maps with the corresponding patterns saying, for instance, that the map is in ${\cal P}_0$ when the corresponding group is.

We will prove the convergence of moments by expanding in Taylor series in $u$ and showing the convergence at each order. 
The moment of order $r$ of the Poisson distribution with intensity 
$u^2\lambda_\alpha$ is 
\begin{equation}\label{eq:Poisson.moment}
m_r = \sum_{k=1}^r S(r, k) u^{2k}\lambda_\alpha^{k}
\end{equation}
where $S(r, k)$ are the Stirling numbers of second kind; the $m$th derivative with respect to $u$ at $u=0$ is 
$S(r, m/2)m!\lambda_\alpha^{m/2}$ for $m$ even, and zero otherwise. 

Recall that each term in \eqref{eq:rmom} gives rise to 
a a certain set of groups of indices made up of overlapping blocks, and each group corresponds to a map  ${\cal C}_{u,i}$ such that the expectation is expressed as in equation in \eqref{eq:exp.cs&ts}. We will show that in limit of large $n$, the only terms which contribute to the $m$-th derivative are those coming from certain configurations with $s-g=m$ where $s$ is the number of groups and $g$ is the number of groups in $\mathcal{P}_0$.

\textbf{Taylor approximation for a given set $({\cal C}_{u,1}, \dots {\cal C}_{u,s})$ up to order $m=s-g$.}
 Let us consider the sum of all the terms in Eq. \eqref{eq:rmom} coming from all the correlations corresponding to a given sequence of maps $ \boldsymbol{\mathcal{C}}:=
 ({\cal C}_{u,1}, \dots ,{\cal C}_{u,s})$. This contribution is given by the product between a combinatorial factor (independent of $n$) counting how many products of $\sigma$s in \eqref{eq:rmom} lead to the same set of maps 
 $\boldsymbol{\mathcal{C}}$ and the following sum
\begin{equation} \label{eq:eq1}
\begin{split} &
\frac{1}{n^r}\sum_{x_0+\dots +x_{s} = n-K}  \tilde{\rho}^{\rm ss}_u
\left ({\cal C}_{u,1} \tilde{\mathcal{T}}^{x_1}_{u}\cdots \tilde{\mathcal{T}}^{x_{s-1}}_{u}{\cal C}_{u,s} (\mathbf{1})\right ),\\
\end{split}
\end{equation}
where $K$ is the total length of all the $s$ blocks in  $\boldsymbol{\mathcal{C}}$, which is smaller than $2|\alpha| r$. 
Note that the factors $\tilde{\mathcal{T}}^{x_0}_u$ and 
$\tilde{\mathcal{T}}^{x_s}_u$ have been suppressed due to stationarity but the indices $x_0, x_s$ are still present in the sum. We denote by $g$ the number of maps in $({\cal C}_{u,1}, \dots, {\cal C}_{u,s})$ that belong to ${\cal P}_0$. 

Let us now consider the Taylor expansion of the correlations in \eqref{eq:eq1}. The $0$-th order term is $0$ because at least one ${\cal C}_i$ is not in ${\cal P}_0$ and hence $\tilde{\rho}^{\rm ss}$ annihilates the result. This is because on one hand, $C_{i*}(\tilde{\rho}^{\rm ss})=\tilde{\rho}^{\rm ss}$ for all $C_{i}\in \mathcal{P}_0$, but on the other hand at least one  ${\cal C}_i$ is not in $\mathcal{P}_0$ and hence it contains a term $\mathcal{A}_{\pm1}$ or $\mathcal{B}_1$ for which $\mathcal{A}_{\pm 1 *}(\tilde{\rho}^{\rm ss})=\mathcal{B}_{1*} (\tilde{\rho}^{\rm ss}) =0$ since $\tilde{K}_1\ket{\chi^{\rm ss}} =0$. 

Before addressing in detail the derivatives of \eqref{eq:eq1} we make some remarks concerning the magnitude of the sum. Note that this contains $O(n^{s})$ terms which are uniformly bounded and the largest possible value of $s$ is $2r$ (all patterns are non-overlapping). On the other hand the sum is preceded by the factor $n^{-r}$ and any derivative further multiplies it by $n^{-1/2}$ since all operators depend on $u$ via $u/\sqrt{n}$. These arguments alone are not sufficient to deduce the convergence of (low order) derivatives, at least for configurations $\boldsymbol{\mathcal{C}}$ with large $s$. 

The key additional ingredient is the fact that for $\mathcal{C}_{i}\notin \mathcal{P}_0$, the factors $\tilde{\mathcal{T}}^{x_i} \mathcal{C}_{i}$ are exponentially decreasing, cf. equation \eqref{eq:exponential.decay.TC}. This will provide more conservative upper bounds for the derivatives of the sum \eqref{eq:eq1}, as detailed below. Before considering that, note that the exponential bound can also be used to get an alternative proof of the convergence to $0$ of
$$\frac{1}{n^r}\sum_{x_0+\dots +x_{s} = n-K}  {\cal C}_{u,1} \tilde{\mathcal{T}}^{x_1}_{u}\cdots \tilde{\mathcal{T}}^{x_{s-1}}_{u}{\cal C}_{u,s} (\mathbf{1}).$$
Indeed, by summing over $x_i$s for groups not in $\mathcal{P}_0$,  we get that the sum \eqref{eq:eq1} is $O(n^g/n^r)$ rather than $O(n^s/n^r)$. Since in any configuration $\boldsymbol{\mathcal{C}}$ , at least one group is not in $\mathcal{P}_0$, and each group in $\mathcal{P}_0$ has at least two of the original $2r$ blocks, we see that $g<r$ so that the whole sum is $O(n^{-1})$. The upshot for derivatives will be that in order to create contributions that do not decay, the derivatives have to be applied in an "efficient" way, namely to factors of the type $\mathcal{T}^{x_i}\circ \mathcal{C}_{i}$ with 
$\mathcal{C}_{i}$ not in $\mathcal{P}_0$. 
This will break the exponential decay and allow for the balancing of the terms in the derivative with the pre-factor $n^{-r}$.


Consider now the case of the first order derivative of the sum \eqref{eq:eq1}. This will split into a sum of sub-sums, one for each of the terms $\mathcal{C}_{u,i}$ or 
$\mathcal{T}^{x_i}_u$ that are differentiated. Each sub-sum will be shown to have a decaying contribution to the derivative.

\begin{itemize}
\item[i)] 
If a term $\mathcal{C}_{u,i}$ in $\mathcal{P}_0$ is differentiated, then the sub-sum is similar to the original sum with the difference that $\mathcal{C}_{u,i}$ is replaced by another bounded term $\dot{\mathcal{C}}_{i}$). The overall contribution is then of the order $n^{-r} \times n^{g} \times n^{-1/2}$ which decays since $g<r$.

\item[ii)] 
If a term $\tilde{\mathcal{T}}^{x_{i-1}}_u$ is differentiated, for which $\mathcal{C}_{u,i}$ is in $\mathcal{P}_0$ then the product $\tilde{\mathcal{T}}^{x_{i-1}}_u \mathcal{C}_{u,i}$ becomes the sum 
\begin{equation}
\label{eq:sum.diff.T}
S_i = \sum_{l=1}^{x_{i-1}} \tilde{\mathcal{T}}
^{x_{i-1}-l}\dot{\tilde{\mathcal{T}}}\tilde{\mathcal{T}}^{l-1}{\cal C}_{i}(\cdot)
\end{equation}
By ergodicity $\tilde{\mathcal{T}}^{l-1}(\mathcal{C}_i(\cdot))$ converges exponentially fast to $\tilde{\rho}^{\rm ss}(\mathcal{C}_i(\cdot))\mathbf{1}$ for large $l$ so 
$$
S_i = 
\sum_{l^\prime=1}^{x_{i-1}} \tilde{\mathcal{T}}^{l^\prime}(\dot{\tilde{\mathcal{T}}}(\mathbf{1})) \cdot \tilde{\rho}^{\rm ss}(\mathcal{C}_i(\cdot)) + O(1)
$$
By the same argument $\tilde{\mathcal{T}}^{l^\prime}\dot{\tilde{\mathcal{T}}}(\mathbf{1})$ converges to $\tilde{\rho}^{\rm ss}(\dot{\tilde{\mathcal{T}}}(\mathbf{1})) = 0$ for large values of $l^\prime$. The latter follows from differentiating $\tilde{\rho}^{\rm ss}_u (\tilde{\mathcal{T}}_u(\cdot)) = \tilde{\rho}^{\rm ss}_u(\cdot)$. Therefore $S_i$ is $O(1)$ which does not change the magnitude of the overall sum before differentiation, so this contribution decays as well, as argued in point i).

\item[iii)] 
If a term $\tilde{\mathcal{T}}^{x_{i-1}}_u$ is differentiated, for which $\mathcal{C}_{u,i}$ is not in $\mathcal{P}_0$, one obtains a sum $S_i$ as in equation \eqref{eq:sum.diff.T}. Since $\mathcal{C}_i$ leaves $\mathbf{1}^\perp$ invariant we have 
$\|\tilde{\mathcal{T}}^l \circ \mathcal{C}_i\| \leq c \lambda^l$ so the terms with large $l$ become negligible while for small $l$ (large $x_{i-1}-l$) 
$\tilde{\mathcal{T}}^{x_{i-1}-l}(\cdot)$ converges to 
$\tilde{\rho}^{\rm ss}(\cdot)\mathbf{1}$. Therefore, using 
$\sum_{l^\prime =0}^{\infty} \tilde{\mathcal{T}}^{l^\prime} \mathcal{C}_i = \tilde{\mathcal{R}} \mathcal{C}_i$ where $\mathcal{R}$ is the Moore-Penrose inverse of ${\rm Id}- \tilde{\mathcal{T}}$, we obtain
$$
S_i=\rho^{\rm ss}(\dot{\tilde{\mathcal{T}}}\tilde{\mathcal{R}}{\cal C}_{i}(\cdot)) \mathbf{1}+o(1).
$$
Hence differentiating $\tilde{\mathcal{T}}^{x_{i-1}}_u$ for which $\mathcal{C}_{u,i}$ is not in $\mathcal{P}_0$ removes the exponential decay associated to un-differentiated  term $\tilde{\mathcal{T}}^{x_{i-1}}_u\mathcal{C}_{u,i}$ so that the upper bound for this contribution to the derivative is 
$n^{g+1}n^{-1/2}n^{-r}$. Since $g<r$ we have 
$g+1 -1/2-r<0$ so the contribution to the derivative 
decays (even though it decays slower than the previous sub-sums, which will be important when considering higher derivatives).


\item[iv)]
If a term $\mathcal{C}_{u,i}$ not in $\mathcal{P}_0$ is differentiated then for large $x_{i-1}$
$$
 \tilde{\mathcal{T}}^{x_{i-1}}\dot{{\cal C}}_{i}(\cdot)=\tilde{\rho}^{\rm ss}(\dot{{\cal C}}_{i}(\cdot))\mathbf{1} + o(1)
$$
This means that such a derivative breaks the exponential decay of the product $ \tilde{\mathcal{T}}^{x_{i-1}}{\cal C}_{i}$ but the contribution to the sum is still bounded by  
$n^{g+1}n^{-1/2}n^{-r}$.
\end{itemize}

The upshot of the argument for the first order derivative is that differentiating a $\mathcal{C}_{u,i}$ in $\mathcal{P}_0$ or the factor $\mathcal{T}^{x_{i-1}}$ in front of it, does not increase the magnitude of the sum and the contribution to the overall sum decays. On the other hand, differentiating a $\mathcal{C}_{u,i}$ which is not in $\mathcal{P}_0$ or the factor $\mathcal{T}^{x_{i-1}}$ in front of it, breaks the overall exponential decay of the product and contributes with an additional factor $n^{1/2}$ compared to the un-differentiated term.

We now consider higher order derivatives. 
One can again consider the possible positions where derivatives are applied and evaluate their separate contributions to the derivative. Following the same argument as in point i) above, one can see that differentiating a term $\mathcal{C}_{u,i}$ in $\mathcal{P}_0$ once or multiple times does not bring any change compared to the term before differentiation. Similarly, as in point ii) one finds that differentiating $\mathcal{T}^{x_{i-1}}_{u}$ in front of $\mathcal{C}_{u,i}$ which is in $\mathcal{P}_0$, once or multiple times does not bring any change (this is due to the fact that $\frac{d^k}{du^k}\mathcal{T}_u^x(\mathbf{1})=0$ for every $k \geq 1$). Now we move our attention to terms $\mathcal{T}^{x_{i-1}}_{u}\mathcal{C}_{u,i}$ for which $\mathcal{C}_{u,i}$ is not in $\mathcal{P}_0$. As in cases iii) and iv) above, differentiating either one of the two factors will break the exponential decay and bring an extra overall multiplicative factor $n^{1/2}$ compared to the un-differentiated terms. However if any of the two factors is differentiated further (more than one derivative in the product $\mathcal{T}^{x_{i-1}}_{u}\mathcal{C}_{u,i}$) then the additional derivative does not change the bound except for the multiplicative factor $n^{-1/2}$ due to differentiation. This shows that in order to obtain contributions that do not vanish asymptotically, the derivatives have to be placed 
on the products $\mathcal{T}^{x_{i-1}}_{u}\mathcal{C}_{u,i}$ for which $\mathcal{C}_{u,i}$ is not in $\mathcal{P}_0$, with at most one derivative for each product.

Recall that we are considering derivatives up to order 
$m=s - g$ which is equal to the number of $\mathcal{C}_{u,i}$ not in $\mathcal{P}_0$. According to the argument above, the most favourable positions for the derivatives is on different terms $\mathcal{T}^{x_{i-1}}_{u}\mathcal{C}_{u,i}$. Therefore, for a derivative of order $k\leq m$ the entire derivative can be upper bounded as 
$$
n^{-r}\times n^{-k/2}\times n^{g} \times n^{k} =
n^{-r+g+k/2}
$$
Since $k\leq m = s-g$ the exponent is smaller than 
$t:= -r+g+ (s-g)/2$. We will show that $t\leq 0$ with equality if and only if $m= s-g$ and the configuration $\boldsymbol{\mathcal{C}}$ is such that all groups in $\mathcal{P}_0$ consists of 2 perfectly overlapping $\sigma$s and the groups not in $\mathcal{P}_0$ consists of single blocks. Indeed for given total number $2r$ of $\sigma$ 
blocks, $s+g$ is maximum if all groups in $\mathcal{P}_0$ have just two $\sigma$ blocks and all others have a single block, so $2g+ s-g = s+g \leq 2r$, which implies $t \leq 0$. In conclusion, for any configuration $\boldsymbol{\mathcal{C}}$, the derivatives below order $m=s-g$ decay asymptotically, and that of order $m$ does as well unless $\boldsymbol{\mathcal{C}}$ is of the special type described above. Derivatives of order above $m$ will be treated below. Note that for the special configurations, $m$ is even so all odd derivatives decay, as is expected from the fact that the intensity of the limit Poisson distribution is proportional to $u^2$. We will now compute the 
limit of the $m$th derivative for the special configuration. 

If $\mathcal{C}_{i}\in \mathcal{P}_0$ consists of two blocks of $\sigma$s then the corresponding product is
$$
\sigma^{\alpha *}\sigma^\alpha= O_{\alpha_1}\dots O_{\alpha_{|\alpha|}}
$$
where 
$\alpha=(\alpha_1, \dots, 
\alpha_{|\alpha|})$ and $O_0 := \mathbf{1}$,  
$O_1= |0\rangle\langle 0|$. Therefore 
$$
\mathcal{C}_i = C^{(\alpha_1)}\dots 
C^{(\alpha_{|\alpha|})}
$$
with $C^{(0)} = \mathcal{A}_0 = \mathcal{T}$ and $C^{(1)}= \mathcal{B}_0$.
Therefore when $x_1,\dots, x_s$ are large the following holds
$$
\mathcal{T}^{x_{i-1}}{\cal C}_{i}(\cdot) = 
\rho^{\rm ss}({\cal C}_i (\cdot)) \mathbf{1}
 + o(1) =  \rho^{\rm ss}(\cdot ) \mathbf{1} + o(1)
$$
since 
$\rho^{\rm ss}\circ \mathcal{A}_0 =\rho^{\rm ss}$ and $\rho^{\rm ss}\circ \mathcal{B}_0 =\rho^{\rm ss}$
This means that for large $x_1,\dots, x_s$ the derivative of order 
$m=2(m/2)$ factorises as 
\[
|\rho^{\rm ss}(\dot{\mathcal{T}}\mathcal{R}{\cal A}^{\alpha}(\mathbf{1})+\rho^{\rm ss}(\tilde{{\cal A}}^{\alpha}(\mathbf{1}))|^{2(m/2)}
\]
where $\tilde{\cal A}^{\alpha}$ is obtained from ${\cal A}^{\alpha}$ differentiating the first ${\cal A}_1$. This formula follows from the contributions obtained at points iii) and iv) above and the fact that the groups that are differentiated contain a single $\sigma$ block.

We remark that the term for which we computed the limit appears in the $m$-th derivative with a factor $m!$ in front.

Finally, we need to compute the combinatorial factor counting the number of terms in the expectation which produce the desired set $\boldsymbol{\mathcal{C}}$ consisting of $g= r-m/2$ groups of two $\sigma$s in $\mathcal{P}_0$ plus 
$m$ groups of single $\sigma$s (with $s=r+m/2$ total number of groups). 

We will show that this is exactly $S(r,m/2)$ (Stirling number of second type) i.e. the number of partitions of a set of $r$ elements into $m/2$ non-empty subsets. For this it is enough to show that the numbers of ways in which we can pair the $\sigma$s is in a bijection with the partitions in $m/2$ classes of $r$ elements. Consider the collection of $\sigma$s
\[
\sigma_{i_1}^\alpha \sigma_{j_1}^{\alpha*} \cdots \sigma_{i_r}^\alpha \sigma_{j_r}^{\alpha*},
\]
which, for our purposes, we can identify with the set of "dipoles"
\[
(\alpha, \alpha*)_1, \cdots, (\alpha, \alpha*)_r
\]
In order to create a ${\cal C}$ in ${\cal P}_0$, we need to pair a $\sigma_{j_z}^{\alpha*}$ with a $\sigma_{i_w}^{\alpha}$ such that $z<w$. This induces an equivalence relation on the collection of dipoles where we identify two dipoles $(\alpha, \alpha*)_z$, $(\alpha, \alpha*)_w$ if $j_z=i_w$ and we impose transitivity of this relation. By construction,  in each equivalence class 
$\{ (\alpha, \alpha*)_{z_1},(\alpha, \alpha*)_{z_k}\}$ we have $j_{z_l} = i_{z_{l+1}}$ so all positions of $\alpha*$s are equal to the position of $\alpha$ in the next dipole. In other words, an equivalence class uniquely determines a set of pairs of equal indices, and altogether the set of equivalence classes uniquely determine the splitting into $\mathcal{P}_0$ groups. Since $i_{z_1}$ and $j_{z_k}$ are the only indices that are not paired in the chosen equivalence class, the number of classes is $m/2$, i.e. half of the number of groups not in $\mathcal{P}_0$. The number of ways to group the $r$ dipoles in $m/2$ equivalence classes is the Stirling number of second type $S(r,m/2)$ which provided the combinatorial factor for our $m$th order derivative.

With this we conclude that the $m$th order derivative of \eqref{eq:rmom} converges to the corresponding derivative of \eqref{eq:Poisson.moment}.

\textbf{Reminder of the Taylor approximation (order $s-g+1$).} The last thing we need to check is that the remainder corresponding to the sum of all the terms in Eq. \eqref{eq:rmom} with $s$ non-overlapping blocks is negligible; the reminder is given by
\[
\small\begin{split} \frac{u^{m^\prime}}{m^\prime!n^{r+\frac{m^\prime}{2}}}\sum_{x_1+\dots +x_{s} = n-K} \left .\frac{d^{m^\prime}}{du^{m^\prime}}f_{x_1,\dots,x_s}(u)\right |_{u=\eta_{x_1,\dots,x_{s-1}}}
\end{split}
\]
for some $|\eta_{x_1,\dots,x_{s-1}}|\leq |u|/\sqrt{n}$, where $m^\prime=s-g+1$ and
$$f_{x_1,\dots,x_s}(u)=\rho^{\rm ss}_u\left ({\cal C}_{1,u} \mathcal{T}^{x_1}_{u}\cdots \mathcal{T}^{x_{s-1}}_{u}{\cal C}_{s,u} (\mathbf{1})\right ).$$
Notice that $r+m^\prime/2>s$: indeed,
$$
r+\frac{s-g+1}{2}>s \Leftrightarrow 2r+1>g+s.$$
Therefore, it is enough to show that
$$\left .\frac{d^{l+1}}{du^{l+1}}f_{x_1,\dots,x_s}(u)\right |_{u=\eta_{x_1,\dots,x_{s-1}}}, \quad 2 \leq l \leq 2r$$
are uniformly bounded: the only terms that requires some care are the derivatives of $\mathcal{T}^x_{u}$, which are of the type
\[
\sum_{0\leq l_1 \leq \cdots \leq l_k \leq x} \mathcal{T}_\eta^{x-l_k} \mathcal{T}_\eta^{(m_k)} \cdots \mathcal{T}_\eta^{(m_1)}\mathcal{T}_\eta^{l_1-1},
\]
where $\mathcal{T}^{(m)}_\eta$ stays for the $m$-th derivative of $\mathcal{T}_u$ evaluated at $\eta$. Using
\begin{itemize}
    \item the spectral decomposition of $\mathcal{T}_\eta(\cdot)=\rho^{\rm ss}_\eta(\cdot)\mathbf{1}+\mathcal{R}_\eta$, with $\rho^{\rm ss}_\eta(\mathcal{R}_\eta(\cdot))=0$ and $\|\mathcal{R}_\eta\| \leq \lambda<1$ (for $n$ big enough) and
    \item the fact that $\mathcal{T}^{(m)}(\mathbf{1})=0$ ($m \geq 1$),
\end{itemize}
we have that
\[
\begin{split}
    &\sum_{0\leq l_1 \leq \cdots \leq l_k \leq x} \|\mathcal{T}_\eta^{x-l_k} \mathcal{T}_\eta^{(m_k)} \cdots \mathcal{T}_\eta^{(m_1)}\mathcal{T}_\eta^{l_1-1}\|= \\
    &\sum_{0\leq l_1 \leq \cdots \leq l_k \leq x} \|\mathcal{T}_\eta^{x-l_k}\mathcal{T}_\eta^{(m_k)}\mathcal{R}_\eta^{l_k-l_{k-1}-1}\cdots \mathcal{T}_\eta^{(m_1)}\mathcal{R}_\eta^{l_1-1}\| \leq\\
    &C\sum_{0\leq l_1 \leq \cdots \leq l_k \leq x} \lambda^{l_k},
    \end{split}
\]
which is bounded.

\underline{\textbf{Quadratures.}} For the sake of keeping notation simple, we will show the proof in the case of $z=1$, but the same reasoning applies to the general case inserting $z$ and $\overline{z}$ where needed. We recall that $k=|\alpha|$.

First of all, notice that the mean of $Q_\alpha(n)$ converges to $u\mu_{\alpha,1}$: indeed, its first moment is given by
\[
\frac{1}{\sqrt{n}} \sum_{i=1}^{n-k+1}\tilde{\rho}^{\rm ss}(u/\sqrt{n})\left (   \sqrt{2}\Re({\cal A}^\alpha) (\mathbf{1})\right )  \rightarrow u \mu_{\alpha,1}.
\]

In order to simplify the proof, we consider the standardised random variable
\[\begin{split}
    &\sqrt{2}\left (Q_\alpha(n) - \frac{n-k+1}{n}u \mu_{\alpha, 1} \right )=\\
    &A^*_\alpha(n)+A_\alpha(n)-2\Re(\mu_{\alpha})u=\\
    &\frac{1}{\sqrt{n}} \sum_{i=1}^{n-k+1} \tilde{\sigma}^{\alpha}_i(u) + \tilde{\sigma}^{\alpha*}_i(u),
\end{split}\]
where
$$\tilde{\sigma}^\alpha(u)=\sigma^\alpha-\frac{u}{\sqrt{n}}\rho((\dot{\mathcal{T}}\mathcal{R}{\cal A}^{\alpha}+\tilde{{\cal A}}^{\alpha})(\mathbf{1}))\mathbf{1}.$$ In order to prove the statement, we need to show that the sequence of standardized quadratures converges in law to a standard Gaussian random variable.

For $r \geq 2$, the $r$-th moment of the standardised random process at time $n$ has the following expression:
\[
\sum_{j_1,\dots,j_r \in \{\alpha,\alpha*\}} \frac{1}{n^{r/2}} \sum_{i_1,\dots,i_r =1}^{n-k+1} \langle \tilde{\sigma}^{j_1}_{i_1}\cdots\tilde{\sigma}^{j_r}_{i_r}  \rangle_{u/\sqrt{n}}.\]

Consider the correlation corresponding to a choice of indices $i_1,\dots,i_r$ such that there are exactly $s$ non-overlapping groups and let $g$ be the number of groups with overlapping terms of the form $\tilde{\sigma}_{z}^{\alpha*}(u)\tilde{\sigma}^\alpha_{z}(u)$ (as before, we denote the set of such groups as ${\cal P}_0$). As before, these are the only overlapping groups that at $u=0$ will produce an operator which is not in $\mathbf{1}^\perp$.

First suppose that $s-g>0$. We can prevent the maps which do not belong to ${\cal P}_0$ to be annihilated by $\rho$ and to cause an exponential decay differentiating and repeating the same computations as in the case of number operators. If we consider the $m$-th term in the Taylor expansion up to $m=s-g$, we can see that it grows at most as $n$ to the power $g+m/2-r/2$; if we want the exponent to be bigger or equal than $0$, we need that $m\geq r-2g$; since $m\leq s-g$, one realises that $g+m/2-r/2$ can at most be equal to $0$ and this is true when $m=s-g$ and $s+g=r$, which means that every group in ${\cal P}_0$ is of the form $\tilde{\sigma}_{z}^{\alpha*}(u)\tilde{\sigma}^\alpha_{z}(u)$ and all the other groups are singletons. However, if a group is composed by a single element, since we centered the random process, the first derivative will not be enough to cancel the exponential decay. Therefore, if $s-g>0$, the corresponding terms will decrease to $0$.

On the other hand, if $r$ is even and $s=g=r/2$, the $0$-th order term is equal to $1$, hence the leading term comes from the case when $r$ is even and $g=r/2$. In this case one can see that the limit quantity is equal to $1$ times the way we can pair the $\tilde{\sigma}(u)$'s in groups of the type $\tilde{\sigma}_{z}^{\alpha*}(u)\tilde{\sigma}^\alpha_{z}(u)$ and this is given by $(r-1)!!$ (which is the number of partition into pairs of a set of $r$ elements).

The reminder can be controlled as in the case of number operators.
\end{proof}

\begin{proof}[Proof of Lemma \ref{lem:alternative}]
\textbf{Alternative expression for $\mu_\alpha$.} 
Using that $\tilde{K}_0\ket{\chi^{\rm ss}}=\ket{\chi^{\rm ss}}$ and $\tilde{K}_1\ket{\chi^{\rm ss}}=0$, one can write
\[\begin{split}
\dot{\tilde{\mathcal{T}}}_*(\tilde{\rho}^{\rm ss})=&\sum_{i=0}^{1}\dot{\tilde{K}}_i\ket{\chi^{\rm ss}}\bra{\chi^{\rm ss}}\tilde{K}_i^* + \tilde{K}_i\ket{\chi^{\rm ss}}\bra{\chi^{\rm ss}}\dot{\tilde{K}}^*_i\\
=&\ket{\dot{\tilde{K}}_0\chi^{\rm ss}}\bra{\chi^{\rm ss}} + \ket{\chi^{\rm ss}}\bra{\dot{\tilde{K}}_0\chi^{\rm ss}}\end{split}.
\]
Under the "gauge condition \eqref{eq:gauge} 
and using the explicit expression of $\tilde{K}_{\theta,i}$'s in Eq. \eqref{eq:tilde.K.0} we obtain $\langle \chi^{\rm ss}|\dot{\tilde{K}}_0 \chi^{\rm ss} \rangle=0$.


From 
$\tilde{K}_0^*\tilde{K}_0 + \tilde{K}_1^*\tilde{K}_1 = \mathbf{1}$ we obtain $\tilde{K}_0^*|\chi^{\rm ss} \rangle = |\chi^{\rm ss} \rangle$.
Therefore 
$$
\tilde{K}_0 = |\chi^{\rm ss}\rangle \langle \chi^{\rm ss}| + P^{\rm ss}_\perp \tilde{K}_0 P^{\rm ss}_\perp
$$ 
where $P^{\rm ss}_\perp = \mathbf{1}- |\chi^{\rm ss}\rangle \langle \chi^{\rm ss}|$, which implies
$$\sum_{k \geq 0}\tilde{K}_0^{k} \ket{\dot{\tilde{K}}_0 \chi^{\rm ss}}= \ket{(\mathbf{1}-\tilde{K}_0)^{-1}\dot{\tilde{K}}_0 \chi^{\rm ss}}$$
and
\begin{eqnarray}
\tilde{\mathcal{R}}_*
\dot{\tilde{\mathcal{T}}}_*
(\tilde{\rho}^{\rm ss})
&=&\ket{(\mathbf{1}-\tilde{K}_0)^{-1}\dot{\tilde{K}}_0 \chi^{\rm ss}}\bra{\chi^{\rm ss}}
\nonumber\\
&+&\ket{\chi^{\rm ss}}\bra{(\mathbf{1}-\tilde{K}_0)^{-1}\dot{\tilde{K}}_0 \chi^{\rm ss}}.
\label{eq:R.T.dot.rho.ss}
\end{eqnarray}
When we evaluate it against
${\cal A}^\alpha(\mathbf{1})$, the first term in the previous equation gets killed, while the second one produces the term
$$\langle \tilde{K}_{\alpha_{|\alpha|}} \cdots \tilde{K}_{\alpha_1}(\mathbf{1}-\tilde{K}_0)^{-1}\dot{\tilde{K}}_0 \chi^{\rm ss} | \chi^{\rm ss}\rangle.$$
The rest of the proof is just a trivial check.

\textbf{Expression for the total intensity.} Since $\lambda_\alpha \geq 0$, one has that $\sum_{\alpha}\lambda_\alpha =C \in [0,+\infty]$ and the limit is always the same irrespectively of the choice of partial sums. Notice that
\begin{eqnarray*}\lambda_{(1)}&=&|\langle (\tilde{K}_1(\mathbf{1}-\tilde{K}_0)^{-1} \dot{\tilde{K}}_0 + \dot{\tilde{K}}_1)\chi^{\rm ss}|\chi^{\rm ss} \rangle|^2\\
&=&\Tr(\ket{\chi^{\rm ss}}\bra{\chi^{\rm ss}}Y),
\end{eqnarray*}
where
\begin{eqnarray} 
Y&:=&\ket{(\tilde{K}_{1}(\mathbf{1}-\tilde{K}_0)^{-1}\dot{\tilde{K}}_0 +\dot{\tilde{K}}_1)\chi^{\rm ss}}
\nonumber\\
&&\bra{(\tilde{K}_{1}(\mathbf{1}-\tilde{K}_0)^{-1}\dot{\tilde{K}}_0 +\dot{\tilde{K}}_1)\chi^{\rm ss}}.
\label{eq:Y}
\end{eqnarray}

For any $\alpha$ such that $|\alpha|\geq 2$, $\lambda_\alpha$ is equal to:
\[\begin{split}
&| \langle \tilde{K}_{\alpha_{|\alpha|}} \cdots \tilde{K}_{\alpha_2} (\tilde{K}_{\alpha_1}(\mathbf{1}-\tilde{K}_0)^{-1}\dot{\tilde{K}}_0 +\dot{\tilde{K}}_1)\chi^{\rm ss} | \chi^{\rm ss}\rangle|^2=\\
&| \langle \tilde{K}_{1} \tilde{K}_{\alpha_{|\alpha|-1}}\cdots \tilde{K}_{\alpha_2} (\tilde{K}_{1}(\mathbf{1}-\tilde{K}_0)^{-1}\dot{\tilde{K}}_0 +\dot{\tilde{K}}_1)\chi^{\rm ss} | \chi^{\rm ss}\rangle|^2=\\
&\Tr(X \tilde{K}_{\alpha_{|\alpha|-1}}\cdots \tilde{K}_{\alpha_2}Y\tilde{K}^*_{\alpha_2}\cdots \tilde{K}^*_{\alpha_{|\alpha|-1}} ),
\end{split}
\]
where $X:=\ket{\tilde{K}_1^*\chi^{\rm ss}}\bra{\tilde{K}_1^*\chi^{\rm ss}}$ and $Y$ is the same as in Eq. \eqref{eq:Y}.

Therefore
$$
\sum_{ 2 \leq |\alpha| \leq N}\lambda_\alpha=\Tr\left (\sum_{k=0}^{N-1}\tilde{\mathcal{T}}^k(X)Y \right )$$
and
$$
\lim_{N\rightarrow +\infty}\sum_{ 2 \leq |\alpha| \leq N}\lambda_\alpha=\sum_{ 2 \leq |\alpha|}\lambda_\alpha=\Tr(\tilde{\mathcal{R}}(X)Y),$$
therefore
$$
\sum_{\alpha}\lambda_\alpha =\Tr(({\rm Id}+\tilde{\mathcal{R}}{\cal B}_1)(\ket{\chi^{\rm ss}}\bra{\chi^{\rm ss}})Y).$$
Let us massage a little the expression we obtained:
\[
\begin{split}
  &\Tr(({\rm Id}+\tilde{\mathcal{R}}{\cal B}_1)(\ket{\chi^{\rm ss}}\bra{\chi^{\rm ss}})Y)=\\
  &\langle \chi^{\rm ss}|Y\chi^{\rm ss} \rangle + \langle \chi^{\rm ss}|{\cal B}_{1*}\tilde{\mathcal{R}}_*(Y-\Tr(Y)\ket{\chi^{\rm ss}}\bra{\chi^{\rm ss}}) \chi^{\rm ss}\rangle=\\
  &\langle \chi^{\rm ss}|Y\chi^{\rm ss} \rangle + \langle \chi^{\rm ss}|\tilde{\mathcal{T}}_*\tilde{\mathcal{R}}_*(Y-\Tr(Y)\ket{\chi^{\rm ss}}\bra{\chi^{\rm ss}})\chi^{\rm ss} \rangle-\\
  &\langle \chi^{\rm ss}|\mathcal{R}_*(Y-\Tr(Y)\ket{\chi^{\rm ss}}\bra{\chi^{\rm ss}})\chi^{\rm ss} \rangle=\\
  &\langle \chi^{\rm ss}|Y\chi^{\rm ss} \rangle -\langle \chi^{\rm ss}|(Y-\Tr(Y)\ket{\chi^{\rm ss}}\bra{\chi^{\rm ss}}) \chi^{\rm ss}\rangle=\Tr(Y).\\ 
\end{split}
\]
We used the fact that ${\cal B}_0(\ket{\chi^{\rm ss}}\bra{\chi^{\rm ss}})=\ket{\chi^{\rm ss}}\bra{\chi^{\rm ss}}$ and that $\tilde{\mathcal{T}}_*\tilde{\mathcal{R}}_*=\tilde{\mathcal{R}}_*-{\rm Id}$. Finally,
$$\Tr(Y)=\|(\tilde{K}_1(\mathbf{1}-\tilde{K}_0)^{-1}\dot{\tilde{K}}_0+\dot{\tilde{K}}_1)\chi^{\rm ss}\|^2$$
\end{proof}

\section{Proof of Corollary \ref{coro:fisher}} \label{app:fisher}

\begin{proof}[Proof of Corollary \ref{coro:fisher}]

We will use the expression of $\sum_{\alpha}\lambda_{\alpha}$ given by
$$-\Re(\langle \chi^{\rm ss}, 2\dot{\tilde{K}}_0(\mathbf{1}-\tilde{K}_0)^{-1}\dot{\tilde{K}}_0+\Ddot{\tilde{K}}_0 \chi^{\rm ss} \rangle )$$
as states in equation \eqref{eq:lambdatot} in Theorem \ref{thm:trajs}. The proof of this identity can be found in Appendix \ref{app:thmtrjs}; the expression we use here has the advantage that it immediately shows that $\sum_{\alpha}\lambda_\alpha$ does not change for different choices of the postprocessing. Before proceeding, we recall the expression of all the terms appearing in the previous equation using quantities of the dynamics of the system alone.

\begin{itemize}
\item $\ket{\chi^{\rm ss}}=\sum_{i=1}^{d} \sqrt{\lambda_i} \ket{\varphi_i}_S \otimes \ket{\varphi_i}_A$, where $\rho^{\rm ss}=\sum_{i=1}^{d}\lambda_i \ket{\varphi_i}\bra{\varphi_i}$ is the spectral resolution of the stationary state of the system dynamics;
\item 
The Kraus operator $\tilde{K}_0$ is uniquely determined as follows (left tensor is the system)
\begin{equation}
\label{eq:tilde.K.0}
\tilde{K}_0=\sum_{k=0}^{1} \sum_{i,j=1}^{d} \sqrt{\frac{\lambda_j}{\lambda_i}} \langle \varphi_j,K_k^*\varphi_i \rangle K_k \otimes \ket{\varphi_i} \bra{\varphi_j}. 
\end{equation}
\item 
The derivative of $\dot{\tilde{K}}_0$ is (where we keep in mind than only the system unitary depends on $\theta$)
$$\dot{\tilde{K}}_0=\sum_{k=0}^{1} \sum_{i,j=1}^{d} \sqrt{\frac{\lambda_j}{\lambda_i}} \langle \varphi_j,K_k^*\varphi_i \rangle \dot{K}_k \otimes \ket{\varphi_i} \bra{\varphi_j}, $$
\item 
The second derivative is
$$\Ddot{\tilde{K}}_0=\sum_{k=0}^{1} \sum_{i,j=1}^{d} \sqrt{\frac{\lambda_j}{\lambda_i}} \langle \varphi_j,K_k^*\varphi_i \rangle \Ddot{K}_k \otimes \ket{\varphi_i} \bra{\varphi_j}. $$
\end{itemize}
Note that
\[
\begin{split}
    \langle \chi^{\rm ss}, \Ddot{\tilde{K}}_0 \chi^{\rm ss} \rangle&= \sum_{k=0}^{1}\sum_{i,j=1}^{d}\lambda_j \langle \varphi_i, \Ddot{K}_k \varphi_j \rangle_S \langle \varphi_j,K_k^* \varphi_i \rangle_S\\
    &=\sum_{k=0}^{1} \Tr(\Ddot{K}_k \rho^{\rm ss}K^*_k).
\end{split}
\]
Hence, using that $\sum_{k=0}^{1}\Ddot{K}_k^*K_k +K_k^*\Ddot{K}_k +2\dot{K}^*_k\dot{K}_k=0$, one has
\[
-\Re(\langle \chi^{\rm ss}, \Ddot{\tilde{K}}_0 \chi^{\rm ss} \rangle)=\sum_{k=0}\Tr(\rho^{\rm ss}\dot{K}^*_k\dot{K}_k).
\]
Let us consider the rest of the total intensity: using that $(\mathbf{1}-\tilde{K}_0)^{-1}\dot{\tilde{K}}_0\ket{\chi^{\rm ss}}=\sum_{l=0}^{+\infty}\tilde{K}_0^{l}\dot{\tilde{K}}_0\ket{\chi^{\rm ss}}$, one gets

\begin{eqnarray*}
    &&\langle \chi^{\rm ss}, \dot{\tilde{K}}_0(\mathbf{1}-\tilde{K}_0)^{-1}\dot{\tilde{K}}_0 \chi^{\rm ss} \rangle\\
    &&=\sum_{i,j=1}^{d}\sum_{a,b=0}^{1}\sum_{l=0}^{+\infty} \sum_{k_1,\dots k_l=0}\lambda_j \langle \varphi_i,\dot{K}_a K_{k_l} \cdots K_{k_1}\dot{K}_b \varphi_j \rangle\cdot\\
    &&\cdot \langle \varphi_j,K_{b}^* K^*_{k_1} \cdots K_{k_l}^* K_a^* \varphi_i \rangle\\
    &&
    =\sum_{a,b=0}^{1}\sum_{l=0}^{+\infty} 
    \sum_{k_1,\dots k_l=0}^{1}\Tr(\dot{K}_a K_{k_l} \cdots K_{k_1}\dot{K}_b \rho^{\rm ss}K_{b}^* K^*_{k_1} \cdots K_{k_l}^* K_a^* )\\
    &&=\Tr \left ( \sum_{a=0}^{1}K^*_a \dot{K}_a \mathcal{R}\left ( \sum_{b=0}^{1} \dot{K}_b \rho^{\rm ss} K_b^*\right )\right ).
\end{eqnarray*}
Therefore,
\begin{eqnarray*}
 &&-2\Re(\langle \chi^{\rm ss}, \dot{\tilde{K}}_0(\mathbf{1}-\tilde{K}_0)^{-1}\dot{\tilde{K}}_0 \chi^{\rm ss} \rangle)\\
 &&=2 \Tr \left (\Im \left (\sum_{a=0}^{1}K^*_a\dot{K}_a \right ) \mathcal{R} \left (\Im \left ( \sum_{b=0}^{1}\dot{K}_b \rho^{\rm ss} K_b^*\right ) \right ) \right ).
\end{eqnarray*}
\end{proof}
\section{Proof of Theorem \ref{thm:trajs} and Proposition \ref{prop:grweps}} \label{app:thmtrjs}

\begin{proof}[Proof of Theorem \ref{thm:trajs} and Proposition \ref{prop:grweps}]
First of all, let us show that $\lambda_{\rm tot}$ given by Eq. \eqref{eq:lambdatot} is equal to $\sum_{\alpha} \lambda_\alpha$, whose expression is given by equation \eqref{eq:sum.lambda.alpha}. We have
\[\begin{split}
&\|(\tilde{K}_1(\mathbf{1}-\tilde{K}_0)^{-1}\dot{\tilde{K}}_0+\dot{\tilde{K}}_1)\chi^{\rm ss}\|^2=\\
&\langle (\mathbf{1}-\tilde{K}_0)^{-1}\dot{\tilde{K}}_0\chi^{\rm ss},\tilde{K}_1^*\tilde{K}_1 (\mathbf{1}-\tilde{K}_0)^{-1}\dot{\tilde{K}}_0\chi^{\rm ss} \rangle\\
&+\langle (\mathbf{1}-\tilde{K}_0)^{-1}\dot{\tilde{K}}_0\chi^{\rm ss},\tilde{K}_1^*\dot{\tilde{K}}_1 \chi^{\rm ss} \rangle\\
&+\langle \chi^{\rm ss},\dot{\tilde{K}}^*_1\tilde{K}_1 (\mathbf{1}-\tilde{K}_0)^{-1}\dot{\tilde{K}}_0 \chi^{\rm ss}\rangle\\
&+\langle \chi^{\rm ss},\dot{\tilde{K}}_1^*\dot{\tilde{K}}_1 \chi^{\rm ss} \rangle.
\end{split}
\]
Note that
\[\begin{split}
&\langle \chi^{\rm ss},\dot{\tilde{K}}_1^*\dot{\tilde{K}}_1 \chi^{\rm ss} \rangle=\frac{1}{2}\langle \chi^{\rm ss}, \Ddot{\cal B}_1(\mathbf{1}) \chi^{\rm ss} \rangle=-\frac{1}{2}\langle \chi^{\rm ss}, \Ddot{\cal B}_0(\mathbf{1}) \chi^{\rm ss} \rangle\\
&=-\Re(\langle \chi^{\rm ss},\Ddot{\tilde{K}}_0 \chi^{\rm ss} \rangle )-\|\dot{\tilde{K}}_0 \chi^{\rm ss}\|^2.
\end{split}\]
Moreover, since $\tilde{K}_0^*\tilde{K}_0 +\tilde{K}_1^*\tilde{K}_1 = \mathbf{1} $,
\[
\begin{split}
  &\langle (\mathbf{1}-\tilde{K}_0)^{-1}\dot{\tilde{K}}_0\chi^{\rm ss},\tilde{K}_1^*\tilde{K}_1 (\mathbf{1}-\tilde{K}_0)^{-1}\dot{\tilde{K}}_0\chi^{\rm ss} \rangle\\
  &=\|(\mathbf{1}-\tilde{K}_0)^{-1}\dot{\tilde{K}}_0\chi^{\rm ss}\|^2-\|\tilde{K}_0(\mathbf{1}-\tilde{K}_0)^{-1}\dot{\tilde{K}}_0\chi^{\rm ss}\|^2.
\end{split}
\]
As for the other terms, one has
\[
\begin{split}
  &\langle (\mathbf{1}-\tilde{K}_0)^{-1}\dot{\tilde{K}}_0\chi^{\rm ss},\tilde{K}_1^*\dot{\tilde{K}}_1 \chi^{\rm ss} \rangle\\
  &=\langle (\mathbf{1}-\tilde{K}_0)^{-1}\dot{\tilde{K}}_0\chi^{\rm ss},\dot{\cal B}_1(\mathbf{1}) \chi^{\rm ss} \rangle\\
  &=-\langle (\mathbf{1}-\tilde{K}_0)^{-1}\dot{\tilde{K}}_0\chi^{\rm ss},\dot{\cal B}_0(\mathbf{1}) \chi^{\rm ss} \rangle\\
  &=-\langle \dot{\tilde{K}}_0(\mathbf{1}-\tilde{K}_0)^{-1}\dot{\tilde{K}}_0\chi^{\rm ss},\chi^{\rm ss} \rangle \\
  &-\langle \tilde{K}_0(\mathbf{1}-\tilde{K}_0)^{-1}\dot{\tilde{K}}_0\chi^{\rm ss} , \dot{\tilde{K}}_0 \chi^{\rm ss}\rangle
\end{split}
\]
and analogously
\[
\begin{split}
&\langle \chi^{\rm ss},\dot{\tilde{K}}^*_1\tilde{K}_1 (\mathbf{1}-\tilde{K}_0)^{-1}\dot{\tilde{K}}_0\chi^{\rm ss}\rangle\\
&=-\langle \chi^{\rm ss},\dot{\tilde{K}}_0(\mathbf{1}-\tilde{K}_0)^{-1}\dot{\tilde{K}}_0\chi^{\rm ss} \rangle \\
&-\langle \dot{\tilde{K}}_0 \chi^{\rm ss},\tilde{K}_0(\mathbf{1}-\tilde{K}_0)^{-1}\dot{\tilde{K}}_0\chi^{\rm ss}  \rangle.
\end{split}
\]
Putting everything together, one gets
\[
\begin{split}
 &\sum_{\alpha}\lambda_{\alpha}=\|(\tilde{K}_1(\mathbf{1}-\tilde{K}_0)^{-1}\dot{\tilde{K}}_0+\dot{\tilde{K}}_1)\chi^{\rm ss}\|^2\\
&=\|(\mathbf{1}-\tilde{K}_0)^{-1}\dot{\tilde{K}}_0\chi^{\rm ss}\|^2-\|\tilde{K}_0(\mathbf{1}-\tilde{K}_0)^{-1}\dot{\tilde{K}}_0\chi^{\rm ss}\|^2\\
&-2\Re(\langle \dot{\tilde{K}}_0(\mathbf{1}-\tilde{K}_0)^{-1}\dot{\tilde{K}}_0\chi^{\rm ss},\chi^{\rm ss} \rangle)\\
&-2\Re(\langle \tilde{K}_0(\mathbf{1}-\tilde{K}_0)^{-1}\dot{\tilde{K}}_0\chi^{\rm ss} , \dot{\tilde{K}}_0 \chi^{\rm ss} \rangle)\\
&-\Re(\langle \chi^{\rm ss},\Ddot{\tilde{K}}_0 \chi^{\rm ss} \rangle )-\|\dot{\tilde{K}}_0 \chi^{\rm ss} \|^2\\
&=\lambda_{\rm tot}+\|(\mathbf{1}-\tilde{K}_0)^{-1}\dot{\tilde{K}}_0\chi^{\rm ss}\|^2\\
&-\|(\tilde{K}_0(\mathbf{1}-\tilde{K}_0)^{-1}+{\mathbf{1}})\dot{\tilde{K}}_0\chi^{\rm ss}\|^2=\lambda_{\rm tot}.
\end{split}
\]

\bigskip Let us now prove the first part of the theorem. What we are actually going to show is that for $n \rightarrow +\infty$ and $|u| \leq n^{\epsilon^\prime}$, one has
$$
e^{\lambda_{tot}u^2}\nu_{u,n}\left (B_{\mathbf{m}}(n)\right )\asymp\prod_{i=1}^{k} \frac{(\lambda_{\alpha^{(i)}}u^2)^{m_{\alpha^{(i)}}}}{m_{\alpha^{(i)}}!}.
$$
Let us consider a fixed \textit{ordered} sequence of excitation patterns $\alpha^{(1)},\dots,\alpha^{(k)}$ (here we do not require them to be distinct). For an observation time $n$ big enough, the probability of observing such a sequence of patterns separated one from the other by more than $n^\gamma$ consecutive $0$s is given by
\begin{eqnarray}
\sum_{\substack{x_1 +\dots +x_{k+1} = n-K\\x_2,\dots,x_{k} \geq n^{\gamma}}}&&\tilde{\rho}^{\rm ss}_u \left ({\cal B}^{x_1}_{u,0}{\cal B}_{u,\alpha^{(1)}} {\cal B}^{x_2}_{u,0}\cdots \right.
\nonumber\\
&&\left .\cdots{\cal B}^{x_k}_{u,0}{\cal B}_{u,\alpha^{(k)}}{\cal B}^{x_{k+1}}_{u,0}(\mathbf{1}))\right ).
\label{eq.proba.b}
\end{eqnarray}

where $K=\sum_{i=1}^{k}|\alpha_i|$ and ${\cal B}_{u,\alpha}(x):={\cal B}_{u,\alpha_{1}}\cdots {\cal B}_{u,\alpha_{|\alpha|}}(x).$

The rest of the proof follows a similar line as the proof of Theorem \ref{thm:limdistribution}: we study the Taylor expansion of the series and identify the leading terms in the limit $n \rightarrow +\infty$. We will often use the spectral decomposition of ${\cal B}_0$, i.e.
$$
{\cal B}_0(x)=\tilde{\rho}^{\rm ss}(x) \ket{\chi^{\rm ss}}\bra{\chi^{\rm ss}} + 
{\cal E}_0(x)
$$
such that
\begin{itemize}
    \item for any $k \geq 1$, for  some constants $C>0$, $0 <\lambda<1$, one has $\|{\cal E}_0^k\| \leq C\lambda^k $,
    \item $\tilde{\rho}^{\rm ss}({\cal E}_0(\cdot))=0$ and 
    ${\cal E}_0(\ket{\chi^{\rm ss}}\bra{\chi^{\rm ss}})=0$.
\end{itemize}
Notice that, in general, $\tilde{\rho}^{\rm ss}(u)$ and the eigenvector of ${\cal B}_{u,0*}$ corresponding to the spectral radius only coincide for $u=0$.

We will also use the following identities
\begin{itemize}
\item 
$\tilde{\rho}^{\rm ss} (\dot{\tilde{\mathcal{T}}}\tilde{\mathcal{R}}{\cal B}_{0}^x{\cal B}_{\alpha}(\cdot )) =0$ follows from \eqref{eq:R.T.dot.rho.ss} and the fact that $\tilde{K}_1|\chi^{\rm ss}\rangle=0$.
\item 
$\tilde{\rho}^{\rm ss}(\dot{\cal B}_\alpha(\cdot))=0$ follows from $\tilde{K}_1|\chi^{\rm ss}\rangle=0$
\item 
$ \tilde{\rho}^{\rm ss}(\dot{\mathcal{B}}_0(\cdot)) =0$ follows from the "gauge condition" \eqref{eq:gauge} 
and using the explicit expression \eqref{eq:tilde.K.0}.
\end{itemize}

\textbf{Significant terms in the Taylor approximation (up to the $2k$-th term).}
We will show that up to derivatives of order $2k$ the only contribution in the Taylor expansion of \eqref{eq.proba.b} which does not decay with $n$ is that coming from the part of the order $2k$ derivative in which one takes the second order derivative to each of the blocks ${\cal B}^{x_i}_{u,0}{\cal B}_{u,\alpha^{(i)}}$, for $i=1,\dots k$. This follows from the observations below:


1) the first derivative of \eqref{eq.proba.b} at $u=0$ is zero since 
$\tilde{\rho}^{\rm ss}({\cal B}_{\alpha}(\cdot ))=0$ and the first derivative of $\tilde{\rho}^{\rm ss}_u({\cal B}_{u,0}^x {\cal B}_{u,\alpha}(\cdot))$ at $u=0$ is equal to
\[\tilde{\rho}^{\rm ss} (\dot{\tilde{\mathcal{T}}}\tilde{\mathcal{R}}{\cal B}_{0}^x{\cal B}_{\alpha}(\cdot ))+\sum_{l=0}^{x-1}\tilde{\rho}^{\rm ss}(\dot{\cal B}_0{\cal B}_0^l {\cal B}_{\alpha}(\cdot))+\tilde{\rho}^{\rm ss}(\dot{\cal B}_\alpha(\cdot))
\]
and the three terms have been shown to be equal to zero above.


2) Each block of the type 
${\cal B}^x_{u,0} {\cal B}_{u,\alpha}$ needs to be differentiated twice. Indeed if it is not differentiated at all, then it will bring an exponential decaying contribution since 
$$\sum_{n^\gamma \leq x \leq n}\|{\cal B}^x_{0} {\cal B}_{\alpha}\| \leq C \sum_{n^\gamma \leq x \leq n}\lambda^x=\frac{C(\lambda^{n^\gamma}-\lambda^{n+1})}{1-\lambda} \rightarrow 0 ;$$
This follows from the fact that since $\alpha$ contains at least one $1$, we have ${\tilde{\rho}}^{\rm ss}({\cal B}_{\alpha} (\cdot)) =0$ which means that ${\cal B}_\alpha(\cdot)$ belongs to the subspace on which 
$\mathcal{B}_0$ acts as the strict contraction $\mathcal{R}_0$. 

Alternatively, if we take a first order derivative we get a vanishing contribution since 
$$
\sum_{n^\gamma \leq x \leq n}\|{\cal B}_{0}^x \dot{{\cal B}}_{\alpha}\| \leq C\sum_{n^\gamma \leq x \leq n}\lambda^x
$$ 
which decays as the term in the previous equation and
\[
\begin{split}&\sum_{n^\gamma \leq k < x \leq n} \|{\cal B}_0^{x-k} \dot{\cal B}_0 {\cal B}_0^{k-1}{\cal B}_{\alpha}\|  \\
&\leq C\sum_{n^\gamma \leq k < x \leq n} \lambda^{x} \asymp Cn^{\gamma}\lambda^{n^\gamma} \rightarrow_{n \rightarrow +\infty} 0.
\end{split}
\]
On the other hand, by differentiating a block of the form ${\cal B}^x_{u,0}{\cal B}_{u,\alpha}$ twice, one stops the exponential decay and obtains a linear growth in $n$; however, we need to remember that every time we use a derivative, everything gets multiplied by $u/\sqrt{n}$, which is of the order $n^{-1/2+\epsilon^\prime}$. Therefore, if we consider the terms in the Taylor expansion up to order $2k$, the only one that does not decay to $0$ is the one with $2k$ derivatives where we use two of them on the block $\tilde{\rho}^{\rm ss}({\cal B}^{x_1}_{u,0}{\cal B}_{u,\alpha^{(1)}}(\cdot))$ at the beginning and other two on each following block of the form ${\cal B}_{u,0}^{x_i}{\cal B}_{u,\alpha^{(i)}}(\cdot)$. Any other term where we have less derivatives involved or where we spend them in a different way is either $0$ or decays at least as (for $k \geq 2$)
$$(n^{\frac{1}{2}+\gamma}\lambda^{n^\gamma})n^{2k\epsilon^\prime}$$
uniformly in 
$u$ for $|u|\leq n^{\epsilon^\prime} $; indeed, for any derivative used in a different way (they cannot be more than $2(k-1)$ because two of them need to be used for the first block), one gains a growth of $n^{1/2}$, but suffers a decay of at least $n^{\gamma}\lambda^{n^\gamma}$.


We now focus on the leading (order $2k$) term of the Taylor expansion. The second derivative of a block of the type 
$\mathcal{B}^x_{u,0}\mathcal{B}_{u,\alpha}$ looks like
\begin{eqnarray*}
&&\mathcal{B}^x_{0} 
\ddot{\mathcal{B}}_\alpha + 
\sum_{k=1}^x 
\mathcal{B}^{k-1}_{0}
\ddot{\mathcal{B}}_{0} 
\mathcal{B}^{x-k}
\mathcal{B}_\alpha\\
&&+
2\sum_{k=1}^{x} 
\mathcal{B}^{k-1}_{0} \dot{\mathcal{B}}_{0} 
\mathcal{B}^{x-k}_{0}\dot{\mathcal{B}}_\alpha \\
&&+
2\sum_{1\leq k< s \leq x} \mathcal{B}^{k-1}_{0}\dot{\mathcal{B}}_{0} 
 \mathcal{B}^{s-k-1}_{0}\dot{\mathcal{B}}_{0} \mathcal{B}^{x-s}_0 \mathcal{B}_\alpha
\end{eqnarray*}
For large $n$ this becomes
\begin{eqnarray*}
&&
\tilde{\rho}^{\rm ss}\left[
(
\ddot{\mathcal{B}}_\alpha   + 
2 \dot{\mathcal{B}}_{0} 
\mathcal{R}_{0}\dot{\mathcal{B}}_\alpha )(\cdot)\right]
 |\chi^{\rm ss}\rangle \langle \chi^{\rm ss}|
 \\
 &&
+\tilde{\rho}^{\rm ss}\left[(\ddot{\mathcal{B}}_0 + 2 \dot{\mathcal{B}}_{0} 
 \mathcal{R}_{0}\dot{\mathcal{B}}_{0} )
 \mathcal{R}_0 \mathcal{B}_\alpha )(\cdot)\right]
 |\chi^{\rm ss}\rangle \langle \chi^{\rm ss}|
\end{eqnarray*}
where 
$
\mathcal{R}_0
$ 
is the Moore-Penrose inverse of ${\rm Id} - \mathcal{B}_0$, and we have used the spectral decomposition of $\mathcal{B}_0$. Moreover, the rightmost term 
$B_0^{x_{k+1}}(\mathbf{1})$ (which is not differentiated) converges to 
$|\chi^{\rm ss}\rangle \langle \chi^{\rm ss}|$. This means that the full $2k$ derivative of \eqref{eq.proba.b} becomes
$$
\frac{u^{2k}}{k!}\prod_{i=1}^{k}\frac{1}{2} 
\tilde{\rho}^{\rm ss}
\left[{\cal B}^{(2)}_{\alpha^{(i)}} (\ket{\chi^{\rm ss}}\bra{\chi^{\rm ss}}) \right]
$$
where 
\begin{eqnarray*}
{\cal B}^{(2)}_{\alpha} :=
\Ddot{\cal B}_{\alpha}+2 \dot{\cal B}_0 {\cal R}_0\dot{\cal B}_\alpha 
+(\Ddot{{\cal B}}_{0}+
    2\dot{\cal B}_0 {\cal R}_0 \dot{\cal B}_0){\cal R}_0 {\cal B}_{\alpha^{}}.
\end{eqnarray*}


The number of terms with two derivatives in each block is equal to $(2k)!/2^k$ (it is the same as the number of terms of the form $2^k=\left (\frac{d^2}{dx}^2(x^2) \right )^k$ in the $2k$-th derivative of $(x^{2})^k$): $(2k)!$ simplifies with the one coming from the Taylor expansion, while the factor $2^{-k}$ can be distributed to each factor. The expression of the factors and the upper bound on the error can be obtained differentiating and using the spectral decomposition of
 ${\cal B}_0$; the $1/k!$ in front comes from the sum: indeed one can see
 that
 $$\sum_{\substack{x_1 +\dots +x_{k+1} = n-K\\x_2,\dots,x_{k} \geq n^{\gamma}}} \asymp \sum_{x_1 +\dots +x_{k+1} = n} \asymp \frac{n^k}{k!},$$
 which is the number of ways one can choose $k$ numbers out of $n$.


The first term of the form
$$
\frac{1}{2}\tilde{\rho}^{\rm ss}(\Ddot{\cal B}_{\alpha}(\ket{\chi^{\rm ss}}\bra{\chi^{\rm ss}}))= |\langle \chi^{\rm ss}|\tilde{K}_{\alpha_{|\alpha|}}\cdots \tilde{K}_{\alpha_2}\dot{\tilde{K}}_1 \chi^{\rm ss} \rangle|^2.
$$
This follows from the fact that $K_1|\chi^{\rm ss} \rangle =0$, so the derivatives need to be applied the the first $K_1$ terms of $\mathcal{B}_{\alpha}$.
For the second term, one gets
\[\begin{split}
&\tilde{\rho}^{\rm ss}(\dot{\cal B}_0 {\cal R}_0\dot{\cal B}_{\alpha}(\ket{\chi^{\rm ss}}\bra{\chi^{\rm ss}}))=\\
&2\Re\left (\langle \tilde{K}_{\alpha_{|\alpha|}}\cdots \tilde{K}_{\alpha_2}\dot{\tilde{K}}_1 \chi^{\rm ss}| \chi^{\rm ss} \rangle\times\right.\\
&\left.
 \langle \chi^{\rm ss}|\tilde{K}_{\alpha_{|\alpha|}} \cdots \tilde{K}_{\alpha_1}(\mathbf{1}-\tilde{K}_0)^{-1}\dot{\tilde{K}}_0 \chi^{\rm ss} \rangle\right).\\
\end{split}\]
Finally, for the last term, below we  will show that
\begin{equation}
\begin{split}
& \frac{1}{2}\tilde{\rho}^{\rm ss}((\Ddot{{\cal B}}_{0}+2\dot{\cal B}_0 {\cal R}_0 \dot{\cal B}_0){\cal R}_0 {\cal B}_{\alpha}(\ket{\chi^{\rm ss}}\bra{\chi^{\rm ss}}))=\\
&|\langle\chi^{\rm ss}| \tilde{K}_{\alpha_{|\alpha|}} \cdots \tilde{K}_{\alpha_1}(\mathbf{1}-\tilde{K}_0)^{-1}\dot{\tilde{K}}_0 \chi^{\rm ss} \rangle|^2.
\label{eq:third.term}
\end{split}
\end{equation}
Note that the sum of the three terms is equal to $|\mu_\alpha|^2$ 
where $\mu_\alpha$ is given in equation \eqref{eq:mu_alpha2} in Lemma \ref{lem:alternative}.

In conclusion, we obtained that for large $n$, the probability of any sequence of $n$ outcomes showing the ordered sequence of excitation patterns given by $\alpha^{(1)}, \dots, \alpha^{(k)}$ is asymptotically equivalent to
\begin{equation}\label{eq:leading}
\frac{1}{k!}\prod_{i=1}^{k} (\lambda_{\alpha^{(i)}}u^2)
\end{equation}
plus a reminder coming from neglecting the terms of order bigger than $2k$ in the Taylor expansion. If we are able to show that the reminder is negligible compared to the term in Eq. \eqref{eq:leading}, then we can prove the statement in Eq. \eqref{eq:countlim}. Indeed, suppose that the sequence we are analysing belongs to $B_{\mathbf{m}}$; then we can partition $B_{\mathbf{m}}$ into
$$\frac{k!}{\prod_{i=1}^k m_{\alpha^{(i)}}!}$$
disjoint subsets containing the excitation patters in $(\alpha^{(1)},m_{\alpha^{(1)}}), \dots, (\alpha^{(k)},m_{\alpha^{(k)}})$ in a fixed order and whose probability asymptotically behaves as (as we just showed)
$$
\frac{1}{k!}\prod_{i=1}^{k} (u^2\lambda_{\alpha^{(i)}})^{m_{\alpha^{(i)}}}.$$

We now prove \eqref{eq:third.term}.
Notice that, since
\[\begin{split}
&|\langle\chi^{\rm ss}| \tilde{K}_{\alpha_{|\alpha|}} \cdots \tilde{K}_{\alpha_1}(\mathbf{1}-\tilde{K}_0)^{-1}\dot{\tilde{K}}_0 \chi^{\rm ss} \rangle|^2=\\
&\langle (\mathbf{1}-\tilde{K}_0)^{-1}\dot{\tilde{K}}_0 \chi^{\rm ss} |{\cal B}_\alpha(\ket{\chi^{\rm ss}}\bra{\chi^{\rm ss}} )|(\mathbf{1}-\tilde{K}_0)^{-1}\dot{\tilde{K}}_0 \chi^{\rm ss}\rangle,
\end{split}\]
we need to prove that
\[\begin{split}
&{\cal R}_{0*}(\Ddot{{\cal B}}_{0*}+2\dot{\cal B}_{0*} {\cal R}_{0*} \dot{\cal B}_{0*})(\tilde{\rho}^{\rm ss}))=\\
&2\ket{(\mathbf{1}-\tilde{K}_0)^{-1}\dot{\tilde{K}}_0 \chi^{\rm ss}}\bra{(\mathbf{1}-\tilde{K}_0)^{-1}\dot{\tilde{K}}_0 \chi^{\rm ss}}+{\rm rem},
\end{split}
\]
where by ${\cal R}_{0*}$ we mean the Moore-Penrose inverse and "${\rm rem }$" is a term which is gives $0$ when evaluated against ${\cal B}_\alpha(\ket{\chi^{\rm ss}}\bra{\chi^{\rm ss}}).$ 
Equivalently,
\[\begin{split}
&(\Ddot{{\cal B}}_{0*}+2\dot{\cal B}_{0*} {\cal R}_{0*} \dot{\cal B}_{0*})(\tilde{\rho}^{\rm ss})=\\
&2({\rm Id}-{\cal B}_{0*})(\ket{(\mathbf{1}-\tilde{K}_0)^{-1}\dot{\tilde{K}}_0 \chi^{\rm ss}}\bra{(\mathbf{1}-\tilde{K}_0)^{-1}\dot{\tilde{K}}_0 \chi^{\rm ss}})+{\rm rem}^\prime,
\end{split}
\]
where ${\cal R}_{0*}({\rm rem}^\prime)={\rm rem}$. 

By explicit computations, one can see that
\[
\begin{split}
  &({\rm Id}-{\cal B}_{0*})(\ket{(\mathbf{1}-\tilde{K}_0)^{-1}\dot{\tilde{K}}_0 \chi^{\rm ss}}\bra{(\mathbf{1}-\tilde{K}_0)^{-1}\dot{\tilde{K}}_0 \chi^{\rm ss}})=\\
  &\ket{(\mathbf{1}-\tilde{K}_0)^{-1}\dot{\tilde{K}}_0 \chi^{\rm ss}}\bra{(\mathbf{1}-\tilde{K}_0)^{-1}\dot{\tilde{K}}_0 \chi^{\rm ss}}-\\
  &\ket{\tilde{K}_0(\mathbf{1}-\tilde{K}_0)^{-1}\dot{\tilde{K}}_0 \chi^{\rm ss}}\bra{\tilde{K}_0(\mathbf{1}-\tilde{K}_0)^{-1}\dot{\tilde{K}}_0 \chi^{\rm ss}}.\\
\end{split}
\]
Using that $\tilde{K}_0(\mathbf{1}-\tilde{K}_0)^{-1}=(\mathbf{1}-\tilde{K}_0)^{-1}-\mathbf{1}$, one gets
\[
\begin{split}
  &({\rm Id}-{\cal B}_{0*})(\ket{(\mathbf{1}-\tilde{K}_0)^{-1}\dot{\tilde{K}}_0 \chi^{\rm ss}}\bra{(\mathbf{1}-\tilde{K}_0)^{-1}\dot{\tilde{K}}_0 \chi^{\rm ss}})=\\
  &\ket{(\mathbf{1}-\tilde{K}_0)^{-1}\dot{\tilde{K}}_0 \chi^{\rm ss}}\bra{\dot{\tilde{K}}_0 \chi^{\rm ss}}+
\ket{\dot{\tilde{K}}_0 \chi^{\rm ss}}\bra{(\mathbf{1}-\tilde{K}_0)^{-1}\dot{\tilde{K}}_0 \chi^{\rm ss}}-\\
  &\ket{\dot{\tilde{K}}_0 \chi^{\rm ss}}\bra{\dot{\tilde{K}}_0 \chi^{\rm ss}}.\\
\end{split}
\]
On the other hand,
\begin{eqnarray*}
   &&(\Ddot{{\cal B}}_{0*}+2\dot{\cal B}_{0*} {\cal R}_{0*} \dot{\cal B}_{0*})(\rho)\\
   &&=
   \ket{\Ddot{\tilde{K}}_0 \chi^{\rm ss}}\bra{\chi^{\rm ss}}+2\ket{\dot{\tilde{K}}_0\chi^{\rm ss}}\bra{\dot{\tilde{K}}_0\chi^{\rm ss}}+\ket{\chi^{\rm ss}}\bra{\Ddot{\tilde{K}}_0 \chi^{\rm ss}}\\ &&+2\ket{\dot{\tilde{K}}_0(\mathbf{1}-\tilde{K}_0)^{-1}\dot{\tilde{K}}_0 \chi^{\rm ss}}\bra{\chi^{\rm ss}}\\&&+2\ket{\chi^{\rm ss}}\bra{\dot{\tilde{K}}_0(\mathbf{1}-\tilde{K}_0)^{-1}\dot{\tilde{K}}_0 \chi^{\rm ss}}\\
   &&+2\ket{\tilde{K}_0(\mathbf{1}-\tilde{K}_0)^{-1}\dot{\tilde{K}}_0 \chi^{\rm ss}}\bra{\dot{\tilde{K}}_0\chi^{\rm ss}}\\
   &&+2\ket{\dot{\tilde{K}}_0\chi^{\rm ss}}\bra{\dot{\tilde{K}}_0(\mathbf{1}-\tilde{K}_0)^{-1}\dot{\tilde{K}}_0 \chi^{\rm ss}}\\
   &&=\ket{(\Ddot{\tilde{K}}_0 +2\dot{\tilde{K}}_0(\mathbf{1}-\tilde{K}_0)^{-1}\dot{\tilde{K}}_0)\chi^{\rm ss}}\bra{\chi^{\rm ss}}\\
   &&+\ket{\chi^{\rm ss}}\bra{(\Ddot{\tilde{K}}_0 +2\dot{\tilde{K}}_0(\mathbf{1}-\tilde{K}_0)^{-1}\dot{\tilde{K}}_0)\chi^{\rm ss}}\\
   &&+2\ket{(\mathbf{1}-\tilde{K}_0)^{-1}\dot{\tilde{K}}_0 \chi^{\rm ss}}\bra{\dot{\tilde{K}}_0. \chi^{\rm ss}}\\
   &&+
2\ket{\dot{\tilde{K}}_0 \chi^{\rm ss}}\bra{(\mathbf{1}-\tilde{K}_0)^{-1}\dot{\tilde{K}}_0 \chi^{\rm ss}}-2\ket{\dot{\tilde{K}}_0 \chi^{\rm ss}}\bra{\dot{\tilde{K}}_0 \chi^{\rm ss}}.\\
\end{eqnarray*}
Note that the last two lines are exactly equal to $2({\rm Id}-{\cal B}_{0*})(\ket{(\mathbf{1}-\tilde{K}_0)^{-1}\dot{\tilde{K}}_0 \chi^{\rm ss}}\bra{(\mathbf{1}-\tilde{K}_0)^{-1}\dot{\tilde{K}}_0 \chi^{\rm ss}})$. Let us look at the remaining part:
\[
\begin{split}
{\rm rem}^\prime=&\ket{(\Ddot{\tilde{K}}_0 +2\dot{\tilde{K}}_0(\mathbf{1}-\tilde{K}_0)^{-1}\dot{\tilde{K}}_0)\chi^{\rm ss}}\bra{\chi^{\rm ss}}+\\
&\ket{\chi^{\rm ss}}\bra{(\Ddot{\tilde{K}}_0 +2\dot{\tilde{K}}_0(\mathbf{1}-\tilde{K}_0)^{-1}\dot{\tilde{K}}_0)\chi^{\rm ss}}.
\end{split}
\]
one can easily see that
\[
\begin{split}
&{\cal R}_{0*}({\rm rem}^\prime)=\\
&\ket{({\mathbf{1}-\tilde{K}_0})^{-1}(\Ddot{\tilde{K}}_0 +2\dot{\tilde{K}}_0(\mathbf{1}-\tilde{K}_0)^{-1}\dot{\tilde{K}}_0)\chi^{\rm ss}}\bra{\chi^{\rm ss}}\\
&+\ket{\chi^{\rm ss}}\bra{{(\mathbf{1}-\tilde{K}_0})^{-1}(\Ddot{\tilde{K}}_0 +2\dot{\tilde{K}}_0(\mathbf{1}-\tilde{K}_0)^{-1}\dot{\tilde{K}}_0)\chi^{\rm ss}}.
\end{split}
\]
In the previous Eq. we used $(\mathbf{1}-K_0)^{-1}$ for the Moore-Penrose inverse: in general, $\ket{(\Ddot{\tilde{K}}_0 +2\dot{\tilde{K}}_0(\mathbf{1}-\tilde{K}_0)^{-1}\dot{\tilde{K}}_0)\chi^{\rm ss}}$ is not orthogonal to $\ket{\chi}^{\rm ss}$.
It is now clear that 
$
\Tr({\cal R}_{0*}({\rm rem}^\prime){\cal B}_{\alpha}(\ket{\chi^{\rm ss}}\bra{\chi^{\rm ss}})=0
$
which proves \eqref{eq:third.term}.

\textbf{Remainder.} Now, we need to take care of the reminder: it is enough to show that the following expression is $o(n^{2k\epsilon})$:
{\small 
\begin{equation} \label{eq:reminder}\begin{split}
\frac{u^{2k+1}}{n^{k+1/2}}\sum_{\substack{x_1 +\dots +x_{k+1} = n-K\\x_2,\dots,x_{k} \geq n^{\gamma}}}&\frac{d^{2k+1}}{u^{2k+1}}\tilde{\rho}^{\rm ss}\left (\frac{u}{\sqrt{n}} \right )\left (\tilde{{\cal B}}^{x_1}_{0,u}\tilde{{\cal B}}_{\alpha_1,u} \tilde{{\cal B}}^{x_2}_{0,u}\cdots \right.\\
&\left .\left .\cdots\tilde{{\cal B}}^{x_k}_{0,u}\tilde{{\cal B}}_{\alpha_k,u}\tilde{{\cal B}}^{x_{k+1}}_{0,u}(\mathbf{1}))\right ) \right |_{u=\eta}
\end{split}
\end{equation}
}
for any $|\eta|\leq u/\sqrt{n}$, where
$$
\tilde{{\cal B}}_{\alpha,u}=e^{\lambda_{tot} |\alpha|u^2/n}{\cal B}_{\alpha,u}
$$
for any string $\alpha$.

Let us first point out some properties of the maps 
$\tilde{{\cal B}}_{\alpha,u}$:
\begin{itemize}
\item for $u$ small enough,
$$\tilde{{\cal B}}_{0,u}(\cdot)=a(u) l(u)(\cdot)r(u) + \tilde{{\cal E}}_{0,u}(\cdot) $$
where
\begin{itemize}
    \item $a(u)=1+O(u^{3}/n^{3/2})$, 
    \item$l(u)(r(u))\equiv 1$, 
    \item $l(u)(\tilde{{\cal E}}_{0,u}(\cdot))=0$,
    \item $\tilde{{\cal E}}_{0,u}(r(u))=0$ and
    \item $\|\tilde{{\cal E}}^k_{0,u}\|\leq C\lambda^k$ for some $C \geq 0$ and $0<\lambda <1$.
\end{itemize}
The order of the reminder in the expression $a(u)$ is due to the fact that $\dot{\lambda}$ is equal to the first derivative at $0$ of the spectral radius of ${\cal B}_{0,u}$ which is equal to $0$ because it attains a maximum there; the fact that $\Ddot{\lambda}=0$ as well is due to the multiplicative factor in front of ${\cal B}_{0,u}$ in the definition of $\tilde{{\cal B}}_{0,u}$.

\item $\tilde{\rho}^{\rm ss}(u/\sqrt{n})(\tilde{{\cal E}}_{0,u}(\cdot))=O(u/\sqrt{n})$.
\item For any excitation pattern $\alpha$ one has
\[
\begin{split}
    &l(u)(\tilde{{\cal B}}_{\alpha,u}(\cdot))=O(u^2/n),\\
    &l(u)(\dot{\tilde{{\cal B}}}_{\alpha,u}(\cdot))=O(u/\sqrt{n}).
\end{split}
\]
Indeed, the first derivative of $l(u)(\tilde{{\cal B}}_{\alpha,u}(\cdot))$ at $0$ is given by
$$
\tilde{\rho}^{\rm ss}(\dot{{\cal B}}_{\alpha}(\cdot))+ \dot{l}({\cal B}_\alpha(\cdot)).$$
That the first addend is $0$ has been shown earlier, while for the second one, it is clear using $\dot{l}=\tilde{\rho}^{\rm ss}(\dot{{\cal B}}_{\alpha} ({\rm Id}-{\cal B}_0)^{-1}(\cdot))$.
\item Moreover,
\begin{equation}
\label{eq:decayprop}\begin{split}
    &l(u)(\dot{\tilde{{\cal B}}}_{0,u}(r(u)))=O(u^2/n) \text{ and}\\
    &l(u)((\Ddot{\tilde{{\cal B}}}_{0,u}+2\dot{\tilde{{\cal B}}}_{0,u}\tilde{{\cal E}}_{0,u}\dot{\tilde{{\cal B}}}_{0,u}(r(u)))=O(u/\sqrt{n}).
\end{split}
\end{equation}
This can be seen differentiating
$$l(u)\tilde{\cal B}_{0,u}(x(u))=a(u)$$
and evaluating at $0$.
\item Finally,
$$ l(u)(\dot{\tilde{{\cal B}}}_{0,u} \tilde{{\cal E}}_{0,u}\tilde{{\cal B}}_{\alpha,u}(\cdot))=O(u/\sqrt{n}), $$
since 
$\tilde{\rho}^{\rm ss}(\dot{{\cal B}}_{0})(\cdot)=0$.
\end{itemize}

Let us now study the growth of the derivatives of $\tilde{\cal B}_{0,u}^{x}\tilde{\cal B}_{\alpha,u}(\cdot)$ for $0 \leq x  \leq n\rightarrow +\infty$:
\begin{enumerate}
\item $0^{\rm th}$ order:
\[\begin{split}
    \tilde{\cal B}_{0,u}^{x}\tilde{\cal B}_{\alpha,u}(\cdot)&=a(u)l(u)(\tilde{\cal B}_{\alpha,u}(\cdot))x(u)+O(\lambda^x) \\
    &=O \left (\frac{u^2}{n}+\lambda^x \right );
\end{split}\]
\item $1^{\rm st}$ order:
\[\begin{split}
    &\sum_{1 \leq l \leq x}\tilde{{\cal B}}_{0,u}^{x-l}\dot{\tilde{{\cal B}}}_{0,u}\tilde{{\cal B}}_{0,u}^{l-1}\tilde{{\cal B}}_{\alpha,u}(\cdot)+\tilde{\cal B}_{0,u}^{x}\dot{\tilde{{\cal B}}}_{\alpha,u}(\cdot)=\\
    &xa(u)^2l(u)(\dot{\tilde{{\cal B}}}_{0,u}(x(u)))l(u)(\tilde{{\cal B}}_{\alpha,u}(\cdot))+\\
    &a(u)\tilde{{\cal R}}_{0,u}\dot{\tilde{{\cal B}}}_{0,u}(x(u)))l(u)(\tilde{{\cal B}}_{\alpha,u}(\cdot))+\\
    &a(u)l(u)(\dot{\tilde{{\cal B}}}_{0,u}\tilde{{\cal R}}_{0,u}\tilde{{\cal B}}_{\alpha,u}(\cdot))x(u) + O\left (\frac{u}{\sqrt{n}} +x\lambda^x\right )
    \\&=O\left ( \frac{u}{\sqrt{n}}+x\lambda^x\right ),
    \end{split}\]
    where we used that $\epsilon <1/6$;
    \end{enumerate}
    We remark that in the case where the block $\tilde{\cal B}_{0,u}^{x}\tilde{\cal B}_{\alpha,u}(\cdot)$ is the first one, due to the action of $\tilde{\rho}^{\rm ss}(u/\sqrt{n})$, the $0^{\rm th}$-order term becomes $O(u^2/n+u\lambda^x/\sqrt{n})$ and the $1^{\rm st}$-order one becomes $O(u/\sqrt{n})$, while if the block is not the first one, $\lambda^x$ decays exponentially fast in $n$ since $x \geq n^\gamma$. 
    \begin{enumerate}
    \setcounter{enumi}{2}
    \item $2^{\rm nd}$ order:
    \[
    \begin{split}
       &\sum_{1 \leq l \leq x}\tilde{{\cal B}}_{0,u}^{x-l}\Ddot{\tilde{{\cal B}}}_{0,u}\tilde{{\cal B}}_{0,u}^{l-1}\tilde{{\cal B}}_{\alpha,u}(\cdot)+\\
       &2\sum_{1 \leq l<k  \leq x}\tilde{{\cal B}}_{0,u}^{x-k}\dot{\tilde{{\cal B}}}_{0,u}\tilde{{\cal B}}_{0,u}^{k-l-1}\dot{\tilde{{\cal B}}}_{0,u}\tilde{{\cal B}}_{0,u}^{l-1}\tilde{{\cal B}}_{\alpha,u}(\cdot)+\\
       &2\sum_{1 \leq l \leq x}\tilde{{\cal B}}_{0,u}^{x-l}\dot{\tilde{{\cal B}}}_{0,u}\tilde{{\cal B}}_{0,u}^{l-1}\dot{\tilde{{\cal B}}}_{\alpha,u}(\cdot)+\\
       &\tilde{\cal B}_{0,u}^{x}\Ddot{\tilde{{\cal B}}}_{\alpha,u}(\cdot)=\\
       &2x^2a(u)^3(l(u)(\dot{\tilde{{\cal B}}}_{0,u}(x(u)))^2l(u)(\tilde{{\cal B}}_{\alpha,u}(\cdot))+\\
       &xa(u)^2l(u)((\Ddot{\tilde{{\cal B}}}_{0,u}+2\dot{\tilde{{\cal B}}}_{0,u}\tilde{{\cal R}}_{0,u}\dot{\tilde{{\cal B}}}_{0,u})(x(u)))l(u)(\tilde{{\cal B}}_{\alpha,u}(\cdot))+\\
       &2xa(u)^2l(u)(\dot{\tilde{{\cal B}}}_{0,u}(x(u))l(u)(\dot{\tilde{{\cal B}}}_{\alpha,u}(\cdot))+O(1)=O(1);
    \end{split}
    \]
    \item $3^{\rm rd}$ order:
\[\begin{split}
    &\sum_{1 \leq l \leq x}\tilde{{\cal B}}_{0,u}^{x-l}\dddot{\tilde{{\cal B}}}_{0,u}\tilde{{\cal B}}_{0,u}^{l-1}\tilde{{\cal B}}_{\alpha,u}(\cdot)+\\
    &3\sum_{1 \leq l<k  \leq x}\tilde{{\cal B}}_{0,u}^{x-k}\Ddot{\tilde{{\cal B}}}_{0,u}\tilde{{\cal B}}_{0,u}^{k-l-1}\dot{\tilde{{\cal B}}}_{0,u}\tilde{{\cal B}}_{0,u}^{l-1}\tilde{{\cal B}}_{\alpha,u}(\cdot)+\\
    &3\sum_{1 \leq l<k  \leq x}\tilde{{\cal B}}_{0,u}^{x-k}\dot{\tilde{{\cal B}}}_{0,u}\tilde{{\cal B}}_{0,u}^{k-l-1}\Ddot{\tilde{{\cal B}}}_{0,u}\tilde{{\cal B}}_{0,u}^{l-1}\tilde{{\cal B}}_{\alpha,u}(\cdot)+\\
    &6\sum_{1 \leq l<k <m \leq x}\tilde{{\cal B}}_{0,u}^{x-m}\dot{\tilde{{\cal B}}}_{0,u}\tilde{{\cal B}}_{0,u}^{m-k-1}\dot{\tilde{{\cal B}}}_{0,u}\tilde{{\cal B}}_{0,u}^{k-l-1}\dot{\tilde{{\cal B}}}_{0,u}\tilde{{\cal B}}_{0,u}^{l-1}\tilde{{\cal B}}_{\alpha,u}(\cdot)+\\
    &3\sum_{1 \leq l \leq x}\tilde{{\cal B}}_{0,u}^{x-l}\Ddot{\tilde{{\cal B}}}_{0,u}\tilde{{\cal B}}_{0,u}^{l-1}\dot{\tilde{{\cal B}}}_{\alpha,u}(\cdot)+\\
    &6\sum_{1 \leq l<k  \leq x}\tilde{{\cal B}}_{0,u}^{x-k}\dot{\tilde{{\cal B}}}_{0,u}\tilde{{\cal B}}_{0,u}^{k-l-1}\dot{\tilde{{\cal B}}}_{0,u}\tilde{{\cal B}}_{0,u}^{l-1}\dot{\tilde{{\cal B}}}_{\alpha,u}(\cdot)+\\
       &3\sum_{1 \leq l \leq x}\tilde{{\cal B}}_{0,u}^{x-l}\dot{\tilde{{\cal B}}}_{0,u}\tilde{{\cal B}}_{0,u}^{l-1}\Ddot{\tilde{{\cal B}}}_{\alpha,u}(\cdot)+\\
       &\tilde{\cal B}_{0,u}^{x}\dddot{\tilde{{\cal B}}}_{\alpha,u}(\cdot)=O(n^{2\epsilon});
       \end{split}\]
       \item $m^{\rm th}$ order for $m \geq 4$: first, notice that
       \[
       \begin{split}
           &\frac{d^m}{du^m}\left (\tilde{\cal B}_{0,u}^{x}\tilde{\cal B}_{\alpha,u}(\cdot)\right )=\\
           &\sum_{l=0}^{m}\binom{m}{l}\frac{d^l}{du^l}(\tilde{\cal B}_{0,u}^{x}) \frac{d^{(m-l)}}{du^{(m-l)}}(\tilde{\cal B}_{\alpha,u}(\cdot)).
       \end{split}
       \]
The term $d^l(\tilde{\cal B}_{0,u}^{x})/du^l$ is a sum over all possible ways of distributing the derivatives among the factors $\tilde{\cal B}_{0,u}$.

Then one can use the spectral decomposition of the $\tilde{\cal B}_{0,u}$ which have not been differentiated and glue together differentiated terms using ${\cal R}_{0,u}$ or turning them into products using the projection $l(u)(\cdot)x(u)$ as we did in the previous items. From Eq. \eqref{eq:decayprop}, one can see that the blocks with a single derivative of a single term $\tilde{\cal B}_{0,u}$ bring a growth of the order $n^{2\epsilon^\prime}$, the terms with two derivatives ($\Ddot{\tilde{{\cal B}}}_{0,u}+2\dot{\tilde{{\cal B}}}_{0,u}\tilde{{\cal R}}_{0,u}\dot{\tilde{{\cal B}}}_{0,u}$) bring a growth of the order $n^{1/2+\epsilon^\prime}$ while all other blocks (with at least $3$ elements), cause a growth of the order at most $n$. Notice that the \textit{highest growth is attained by making as many groups of three derivatives as possible} and this will be used in obtaining the following estimate.

The term $d^{(m-l)}(\tilde{\cal B}_{\alpha,u})/du^{(m-l)}$ does not cause any reduction in the growth if $m-l \geq 2$, while causes a decay equal to $n^{-1/2+\epsilon^\prime}$ if $m=l+1$ and $n^{-1+2\epsilon^\prime}$ if $m=l$.

Therefore,
\begin{equation} \label{eq:growth}
    \begin{split}
&\left \| \frac{d^m}{du^m}\left (\tilde{\cal B}_{0,u}^{x}\tilde{\cal B}_{\alpha,u}(\cdot)\right )\right \| \lesssim \\
&n^{\max\{f(m-2), f(m-1)-\frac{1}{2}+\epsilon^\prime, f(m)-1+2\epsilon^\prime\}},
\end{split}\end{equation}
\end{enumerate}
where
\[
f(m)=a(m) +\left (\frac{1}{2}+\epsilon^\prime \right )b(m)+2\epsilon^\prime (m-(3a(m)+2b(m)))
\]
\[a(m):=\left \lfloor \frac{m}{3}\right \rfloor, \quad b(m):=\left \lfloor \frac{1}{2}\left (m-3\left \lfloor \frac{m}{3}\right \rfloor \right )\right \rfloor.
\]
$a(m)$ is the maximum number of groups with three derivatives and $b(m)$ is the maximum number of groups with two derivatives that we can make with the derivatives left. Notice that if $\epsilon^\prime <1/6$, then $f(m) \leq m/3$: indeed,
\[\begin{split}
    &a(m) +\left (\frac{1}{2}+\epsilon^\prime \right )b(m)+2\epsilon^\prime (m-(3a(m)+2b(m)))=\\
    &2\epsilon^\prime m + (1-6\epsilon^\prime)a(m)+ \left ( \frac{1}{2}-3\epsilon^\prime \right ) b(m) \leq \\
    &2\epsilon^\prime m + (1-6\epsilon^\prime)a(m)+ \frac{1}{2}\left ( \frac{1}{2}-3\epsilon^\prime \right )(m-3a(m)) \leq \\
    &\frac{1}{2}\left (\left ( \frac{1}{2}+\epsilon^\prime \right )m +\left ( \frac{1}{2}-3\epsilon^\prime \right )a(m) \right )\leq \\
    &\frac{1}{2}\left (\left ( \frac{1}{2}+\epsilon^\prime \right )m +\left ( \frac{1}{6}-\epsilon^\prime \right )m \right )\leq \frac{m}{3}.
\end{split}\]
Therefore, the growth of the term in Eq. \eqref{eq:growth} is upper bounded by
\begin{equation}\label{eq:exder}
n^{\max\{\frac{m-2}{3}, \frac{m-1}{3}-\frac{1}{2}+\epsilon^\prime, \frac{m}{3}-1+2\epsilon^\prime\}}.\end{equation}
This implies that the first two derivatives spent on a block of the form $\tilde{\cal B}_{0,u}^x\tilde{\cal B}_{\alpha,u}$ cause a growth equal to $n^{1/2-\epsilon^\prime}$, while the other $m-2$ only brings a growth at most equal to the term in Eq. \eqref{eq:exder}. Notice that it is more convenient to spend them in a way that every block of the form $\tilde{\cal B}_{0,u}^x\tilde{\cal B}_{\alpha,u}$ has two derivatives, since one has that
\[
\left (\frac{1}{2}-\epsilon^\prime \right )(m-2) \geq \frac{m-2}{3} \Leftrightarrow m \geq 2, 
\]
moreover
\[
\left (\frac{1}{2}-\epsilon^\prime \right )(m-2) \geq \frac{m-1}{3}-\frac{1}{2}+\epsilon^\prime \Leftrightarrow m \geq 1
\]
and
\[
\left (\frac{1}{2}-\epsilon^\prime \right )(m-2) \geq \frac{m}{3}-1+2\epsilon^\prime \Leftrightarrow m \geq 0.
\]
Therefore the growth of
\[\begin{split}
\sum_{\substack{x_1 +\dots +x_{k+1} = n-K\\x_2,\dots,x_{k} \geq n^{\gamma}}}&\frac{d^{2k+1}}{u^{2k+1}}\tilde{\rho}^{\rm ss}\left (\frac{u}{\sqrt{n}} \right )\left (\tilde{{\cal B}}^{x_1}_{0,u}\tilde{{\cal B}}_{\alpha_1,u} \tilde{{\cal B}}^{x_2}_{0,u}\cdots \right.\\
&\left .\left .\cdots\tilde{{\cal B}}^{x_k}_{0,u}\tilde{{\cal B}}_{\alpha_k,u}\tilde{{\cal B}}^{x_{k+1}}_{0,u}(\mathbf{1}))\right ) \right |_{u=\eta},
\end{split}\]
is of the order of $O(n^{2\epsilon^\prime})$ uniformly in $u$ and it is attained by the term where we spend at least two derivatives in each block of the form $\tilde{\cal B}_{0,u}^x\tilde{\cal B}_{\alpha,u}$ and the last derivative in any of such blocks. One can see that it is not convenient to use any derivative in the final term $\tilde{\cal B}_{0,u}^x(\mathbf{1})$, since it is better to contrast the decay induced by the other blocks. To conclude, the term in Eq. \eqref{eq:reminder} grows at most as $n^{(2k+3)\epsilon^\prime-1/2}=o(n^{2k\epsilon^\prime})$. We proved Proposition \ref{prop:grweps} and notice that Theorem \ref{thm:trajs} follows considering $\epsilon^\prime=0$.
\end{proof}

\section{Proof of Corollary \ref{coro:lconv}} \label{sec:lconvapp}
\begin{proof}
    Let us define the measurable space given by the set $\Omega=\mathbb{N}^{\mathbb{N}}$ together with the $\sigma$-field ${\cal F}$ generated by cylindrical sets; we can consider on $(\Omega, {\cal F})$ the law $\overline{\nu}_u$ of $\{N_\alpha:\alpha \in {\cal P}\}$ and the law of $\{N_\alpha(n): \alpha \in {\cal P}\}$, which, with a slight abuse of notation, we still denote by $\nu_{u,n}$ as well. We know that for every finite set $A$ of patterns, one has
    \begin{equation} \label{eq:convergence}
        \lim_{n \rightarrow +\infty} \sup_{|u|<C} |\nu_{u,n}(A)-\overline{\nu}_u(A)|=0.
    \end{equation}
    Notice that, for every $\epsilon >0$, there exists a set $A_\epsilon$ of finitely many patter such that
    $$\inf_{|u|<C} \overline{\nu}_u(A_\epsilon) > 1-\epsilon.$$
    Therefore, using Eq. \eqref{eq:convergence}, one has that there exists $N_\epsilon$ such that $\forall n \geq N_\epsilon$,
    $$\inf_{|u|<C} \nu_{u,n}(A_\epsilon) > 1-\epsilon, \text{ and } \sup_{|u|<C} |\nu_{u,n}(A_\epsilon)-\overline{\nu}_u(A_\epsilon)|<\epsilon.$$ 
    Therefore, given a bounded function $f:\mathbb{N}^{\cal P} \rightarrow \mathbb{R}$, for every $\epsilon>0$, $\forall n \geq N_\epsilon$, one has
    \[\begin{split}
    &\sup_{|u| < C}|\mathbb{E}_{\nu_{u,n}}[f]-\mathbb{E}_{\overline{\nu}_{u}}[f]| \leq\\
    &\|f\|_\infty\left (\sup_{|u|<C} |\nu_{u,n}(A_\epsilon)-\overline{\nu}_u(A_\epsilon)| +\right .\\
    &\left .\sup_{|u|<C} \nu_{u,n}(A_\epsilon^C)+\overline{\nu}_{u}(A_\epsilon^C) \right ) \leq \\
    &3\|f\|_\infty \epsilon.
    \end{split}\]
    For the arbitrariness of $\epsilon$, we proved the first statement about the weak convergence.

    Consider the random process $N_{\rm tot}(n)$ that counts all the occurrences of $1$'s in the output up to time $n$ and notice that for every $\alpha \in {\cal P}$, $N_\alpha(n) \leq N_{\rm tot}(n)$. Therefore for every $p \geq 1$ and $m \in \mathbb{N}$, one has
    $$\mathbb{E}_{\nu_{u,n}}[N_\alpha(n)^p \mathbf{1}_{\{N_\alpha(n) >m\}}] \leq \mathbb{E}_{\nu_{u,n}}[N_{\rm tot}(n)^p \mathbf{1}_{\{N_{\rm tot}(n) >m\}}].$$
    If we show that the moments of every order of $N_{\rm tot}(n)$ converge to some finite limit and that it converges in law to some limit random variable $X_u$, we obtain the second statement as well. Indeed, let us call $C(p)$ the limit of the $p$-moment of $N_{\rm tot}(n)$; notice that for every $p \geq 1$, $m \in \mathbb{N}$
    \[\begin{split}
        &\limsup_{n \rightarrow +\infty}\mathbb{E}_{\nu_{u,n}}[N_\alpha(n)^p \mathbf{1}_{\{N_\alpha(n) >m\}}] \leq\\
        &\limsup_{n \rightarrow +\infty}\mathbb{E}_{\nu_{u,n}}[N_{\rm tot}(n)^p \mathbf{1}_{\{N_{\rm tot}(n) >m\}}] \leq \\
        &\lim_{n \rightarrow +\infty}\mathbb{E}_{\nu_{u,n}}[N_{\rm tot}(n)^{pq}]^{1/q} \nu_{u,n}(N_{\rm tot}(n) >m)^{1/q^\prime} =\\
        &C(pq)^{1/q} \mathbb{P}(X_u >m)^{1/q^\prime}.
    \end{split}\]
    Notice that in the last inequality we made use of H\"older inequality for some pair of conjugate indices $(q, q^\prime)$. Therefore, if we fix $p\geq 1$, for every $\epsilon>0$, one can choose $m_\epsilon$ such that
    \begin{align*}
    &\limsup_{n \rightarrow +\infty}\mathbb{E}_{\nu_{u,n}}[N_\alpha(n)^p \mathbf{1}_{\{N_\alpha(n) >m_\epsilon\}}] \leq \epsilon/2,\\
   & \mathbb{E}_{\overline{\nu}_{u}}[N_\alpha^p \mathbf{1}_{\{N_\alpha >m_\epsilon\}}] \leq \epsilon/2
    \end{align*}
    and one gets
    \[\begin{split}
    &\limsup_{n \rightarrow +\infty}|\mathbb{E}_{\nu_{u,n}}[N_\alpha(n)^p]-\mathbb{E}_{\overline{\nu}_{u}}[N_\alpha^p]| \leq \\
    & \lim_{n \rightarrow +\infty}|\mathbb{E}_{\nu_{u,n}}[N_\alpha(n)^p\mathbf{1}_{\{N_\alpha(n) \leq m_\epsilon\}}]-\mathbb{E}_{\overline{\nu}_{u}}[N_\alpha^p \mathbf{1}_{\{N_\alpha \leq m_\epsilon\}}]|+\\
    & \limsup_{n \rightarrow +\infty}\mathbb{E}_{\nu_{u,n}}[N_\alpha(n)^p \mathbf{1}_{\{N_\alpha(n) >m_\epsilon\}}] +\\
    &\mathbb{E}_{\overline{\nu}_{u}}[N_\alpha^p \mathbf{1}_{\{N_\alpha >m_\epsilon\}}] \leq \epsilon.
    \end{split}\]
Since this holds for every $\epsilon>0$, we proved the statement.

We need to show that, under $\nu_{u,n}$, $N_{\rm tot}(n)$ converges in law to a random variable $X_u$ with finite moments of every order and that we have convergence of the moments as well. One can see that the Laplace transform of $N_{\rm tot}(n)$ can be expressed as
$$\mathbb{E}_{\nu_{u,n}}[e^{zN_{\rm tot}(n)}]=\tilde{\rho}^{\rm ss}(u/\sqrt{n})(\tilde{\mathcal{T}}_{u,z,n}^{n}(\mathbf{1})) \quad z \in \CC,$$
where
\[\begin{split} \tilde{\mathcal{T}}_{u,z,n}(\cdot)=&\tilde{K}^*_0(u/\sqrt{n})\cdot \tilde{K}_0(u/\sqrt{n})+\\
&e^{z}\tilde{K}^*_1(u/\sqrt{n}) \cdot \tilde{K}^*_1(u/\sqrt{n}).\end{split}\]
Notice that $\tilde{\mathcal{T}}_{0,z}:=\tilde{\mathcal{T}}_{0,z,n}$ is independent from $n$ and is an analytic perturbation of $\tilde{\mathcal{T}}$. If we pick $z$ small enough in modulus, perturbation theory ensures that $\tilde{\mathcal{T}}_{0,z}$ has $1$ as eigenvalue with maximum modulus with $\ket{\chi^{\rm ss}}\bra{\chi^{\rm ss}}$ as left eigenvector. Let $x_z$ be the corresponding right eigenvector such that 
${\rm Tr}(\rho^{\rm ss} x_z) = \bra{\chi^{\rm ss}}x_z \ket{\chi^{\rm ss}}=1$, and let
$$
x_z=\begin{pmatrix} 1 & a \\ b & c \end{pmatrix}
$$
be its block matrix form with respect to 
the decomposition of $\mathcal{H}_{sa}$ into 
$\mathbb{C}|\chi^{\rm ss}\rangle$ and its orthogonal complement. Then one can prove that 
$a=b=0$ by using the fact that the Kraus operators are of the form
$$K_0=\begin{pmatrix} 1 & 0 \\ 0 & \beta\end{pmatrix}, \quad K_1=\begin{pmatrix} 0 & \gamma \\ 0 & \delta\end{pmatrix}$$ 
for some blocks $\beta, \gamma, \delta $ 
such that $|\beta|^2+|\gamma|^2+|\delta|^2=\mathbf{1}$. Let us first fix $z$ small enough, then for $n$ big enough, $\tilde{\mathcal{T}}_{u,z,n}$ has a unique eigenvalue $\lambda_z(u,n)$ of maximum modulus with corresponding left and right eigenvectors $l_z(u,n)$, $x_z(u,n)$ and one has that
\[\begin{split}&\mathbb{E}_{\nu_{u,n}}[e^{zN_{\rm tot}(n)}] =\\
&\lambda_z(u,n)^n l_z(u,n)(\mathbf{1}) \tilde{\rho}^{\rm ss}(u/\sqrt{n})(x_z(u))+o(n)\end{split}\]
where $o(n)$ is uniform in $u$ (and $z$ in compact sets). Notice that both $l_z(u,n)(\mathbf{1})$ and $\tilde{\rho}^{\rm ss}(u/\sqrt{n}) (x_z(u))$ converge to $1$ for $n \rightarrow +\infty$; regarding the behaviour of $\lambda_z(u,n)$ consider the Taylor expansion up to second order in $u$ around $0$:
$$\lambda_z(u,n)=1+\frac{u}{\sqrt{n}}\lambda^{(1)}_z+\frac{u^2}{n}\lambda^{(2)}_z + o(n^{-1}).$$
We can choose $l_z(u,n)$ and $x_z(u,n)$ such that $\Tr(l_z(u,n)x_z(u,n))\equiv 1$, therefore differentiating $\Tr(l_z(u,n)\tilde{\mathcal{T}}_{z,u,n}(x_z(u,n)))=\lambda_z(u,n)$ at $0$ one gets
$$\lambda^{(1)}_z=\Tr(l_z \partial_u\tilde{\mathcal{T}}_{z|u=0} x_z)$$
which can be easily seen to be $0$. Summing up, we proved that
$$\lim_{n\rightarrow +\infty}\mathbb{E}_{\nu_{u,n}}[e^{zN_{\rm tot}(n)}]=e^{\lambda_z^{(2)}}$$
uniformly in $u$ and in $z$ in compact small neighborhoods of $0$. Since for every $u$, $f_n(z):=\mathbb{E}_{\nu_{u,n}}[e^{zN_{\rm tot}(n)}]$ are analytic functions around $0$, then we can deduce that $e^{\lambda_z^{(2)}}$ is analytic as well and that we have uniform convergence (on compact small neighborhoods of $0$) of all the derivatives, which consists exactly in the convergence of moments of all orders and we are done.
\end{proof}

\section{Achievability of the QCRB under additional assumptions 
} \label{app:finalestim}

In this section we present a proof of Theorem \ref{thm:opt} and we comment on the gap between the hypothesis that we need to assume and what we proved in Proposition \ref{prop:grweps}.




\begin{proof}[Proof of Theorem \ref{thm:opt}]
By hypothesis, we know that $\theta$ belongs to the (random) confidence interval $$I_n = (\tilde{\theta}_n-n^{-1/2 +\epsilon}, ~ \tilde{\theta}_n+n^{-1/2 +\epsilon})$$
with high probability. In order to prove the statement, it suffices to show that 
$$
|(\mathbb{E}_\theta[e^{ia \sqrt{n} (\hat{\theta}_n-\theta)}|\tilde{\theta}_n]-e^{-\frac{a^2}{2f_\theta}})|\chi_{\{\theta \in I_n\}}(\tilde{\theta}_n)
$$ 
can be upper bounded uniformly in $\tilde{\theta}_n$ by a sequence converging to $0$. Indeed, notice that for every $a \in \mathbb{R}$
\[\begin{split}
\mathbb{E}_\theta[e^{ia \sqrt{n} (\hat{\theta}_n-\theta)}]&=e^{-\frac{a^2}{2f_\theta}}\mathbb{P}_\theta(\theta \in I_n) \\
&+\int_{\theta \in I_n}p_\theta(d\tilde{\theta}_n)(\mathbb{E}_\theta[e^{ia \sqrt{n} (\hat{\theta}_n-\theta)}|\tilde{\theta}_n]-e^{-\frac{a^2}{2f_\theta}})\\
&+\int_{\theta \notin I_n}p_\theta(d\tilde{\theta}_n)\mathbb{E}_\theta[e^{ia \sqrt{n} (\hat{\theta}_n-\theta)}|\tilde{\theta}_n].\\
\end{split}
\]
Since $\mathbb{P}_\theta(\theta \notin I_n)$ goes to zero, the first term goes to $e^{-\frac{a^2}{2f_\theta}}$ and the third one vanishes. If we show that the second term vanishes then we obtain the convergence in distribution of $\sqrt{n} (\hat{\theta}_n-\theta)$ as in the statement.

First of all notice that for every $a \in \mathbb{R}$, for every $\tilde{\theta} \in \Theta$
\[\begin{split}
&\mathbb{E}[e^{ia(\overline{Y}_{n,\theta, \tilde{\theta}}-\sqrt{n}(\theta-\tilde{\theta}))}]=\\
&e^{\lambda_{\rm tot}(\tilde{\theta})(\sqrt{n}(\theta-\tilde{\theta})+\tau_n)^2 (e^{ia/(2\lambda_{\rm tot}(\tilde{\theta}) \tau_n)} -1)-ia \left (\frac{\tau_n}{2}+\sqrt{n}(\tilde{\theta}-\theta) \right )}.
\end{split}\]
Therefore for every $a \in \RR$ one has
$$\lim_{n \rightarrow +\infty} \sup_{\tilde{\theta}:|\tilde{\theta}-\theta|<n^{-1/2+\epsilon}}\left |\mathbb{E}[e^{ia(\overline{Y}_{n,\theta, \tilde{\theta}}-\sqrt{n}(\theta-\tilde{\theta}))}]-e^{-\frac{a^2}{f_\theta}}\right |=0$$

We will denote by $u=\sqrt{n}(\theta-\tilde{\theta}_n)$. The definition of $\hat{\theta}_n$ implies that
\[\sqrt{n}(\hat{\theta}_n-\theta)=Y_n-u.\]
Therefore we obtain that
\[\begin{split}
    &\sup_{\tilde{\theta}:|\tilde{\theta}-\theta|<n^{-1/2+\epsilon}}\left |\mathbb{E}_{\theta}[e^{ia(Y_n-u)}|\tilde{\theta}_n=\tilde{\theta}]-e^{-\frac{a^2}{2f_\theta}}\right | \leq \\
    &\sup_{\tilde{\theta}: |\tilde{\theta}-\theta|<n^{-1/2+\epsilon}}|\mathbb{E}_\theta[e^{iaY_n}|\tilde{\theta}_n=\tilde{\theta}]-\mathbb{E}[e^{ia\overline{Y}_{n,\theta, \tilde{\theta}}}]|+\\
    &\sup_{\tilde{\theta}:|\tilde{\theta}-\theta|<n^{-1/2+\epsilon}}\left |\mathbb{E}[e^{ia(\overline{Y}_{n,\theta, \tilde{\theta}}-\sqrt{n}(\theta-\tilde{\theta}))}]-e^{-\frac{a^2}{f_\theta}}\right |.
\end{split}\]
Hence,
$$\lim_{n \rightarrow +\infty} \sup_{\tilde{\theta}:|\tilde{\theta}-\theta|<n^{-1/2+\epsilon}}\left |\mathbb{E}_{\theta}[e^{ia(Y_n-u)}|\tilde{\theta}_n=\tilde{\theta}]-e^{-\frac{a^2}{2f_\theta}}\right |=0$$
and we are done.
\end{proof}
Let us briefly comment on the relationship between the additional hypothesis we introduced (Eq. \eqref{eq:hypo}) and the result in Proposition \ref{prop:grweps}; let us consider the following family of stochastic processes: for every $n \in \mathbb{N}$, $\tilde{\theta}$, $\theta \in \Theta$, consider the collection of independent random variables
$$\overline{N}_{n,\theta, \tilde{\theta},\alpha} \sim \text{Poisson}(\lambda_\alpha(\tilde{\theta}) (\sqrt{n}(\theta-\tilde{\theta})+\tau_n)^2), \quad \alpha \in \mathcal{P}$$
and their law $\overline{\nu}_{n, \theta, \tilde{\theta}}$ on $\mathbb{N}^{\cal P}.$ (together with the $\sigma$-field of cylindrical sets). Notice that
$\sum_{\alpha \in {\cal P}}\overline{N}_{n, \theta,\tilde{\theta}, \alpha}$ converges in mean square and has the same law as $\overline{N}_{n, \theta,\tilde{\theta}}.$

Inspecting the proof, one can notice that the convergence in the statement of Proposition \ref{prop:grweps} holds uniformly in a small neighborhood of the reference parameter $\theta_0$, therefore we can restate it in the following way: for $\epsilon$ small enough and for every finite collections of excitation patterns counts $\mathbf{m}$ one has
$$\lim_{n \rightarrow +\infty}\sup_{\tilde{\theta}:|\tilde{\theta}-\theta|<n^{-1/2+\epsilon}} \left | \frac{\nu_{\sqrt{n}(\theta-\tilde{\theta}),n}(B_{\mathbf{m}}(n))}{\overline{\nu}_{\theta, \tilde{\theta},n}(\mathbf{m})}-1\right |=0.$$
If we were able to show that the previous result still holds integrating with respect to $\overline{\nu}_{\theta, \tilde{\theta},n}$, i.e.
\begin{equation}\label{eq:L1conv}
\begin{split}
    &\lim_{n \rightarrow +\infty}\sup_{\tilde{\theta}:|\tilde{\theta}-\theta|<n^{-1/2+\epsilon}} \sum_{\mathbf{m}}\overline{\nu}_{\theta, \tilde{\theta},n}(\mathbf{m})\left | \frac{\nu_{\sqrt{n}(\theta-\tilde{\theta}),n}(B_{\mathbf{m}}(n))}{\overline{\nu}_{\theta, \tilde{\theta},n}(\mathbf{m})}-1\right |=\\
&\lim_{n \rightarrow +\infty}\sup_{\tilde{\theta}:|\tilde{\theta}-\theta|<n^{-1/2+\epsilon}} \sum_{\mathbf{m}}\left | \nu_{\sqrt{n}(\theta-\tilde{\theta}),n}(B_{\mathbf{m}}(n))-\overline{\nu}_{\theta, \tilde{\theta},n}(\mathbf{m})\right |=0,
\end{split}\end{equation}
this would imply the condition in Eq. \eqref{eq:hypo} (it can be seen using the fact that $\sum_{\alpha \in {\cal P}}\overline{N}_{n, \theta,\tilde{\theta}, \alpha}$ and $\overline{N}_{n, \theta,\tilde{\theta}}$ have the same law). Unfortunately, we are not able to prove this.

The last remark we make is that Eq. \eqref{eq:L1conv} cannot be true unless
$$\lim_{n \rightarrow +\infty} \sup_{\tilde{\theta}:|\tilde{\theta}-\theta|<n^{-1/2+\epsilon}} \overline{\nu}_{\theta, \tilde{\theta},n}(G(n)^C)=0,$$
where $G(n)$ is the set of all $\mathbf{m}$'s such that $\nu_{\sqrt{n}(\theta-\tilde{\theta}),n}(B_{\mathbf{m}}(n))>0$. We can prove that this is indeed the case.
\begin{lemma}
If $\epsilon$ is small enough, then
$$\lim_{n \rightarrow +\infty} \sup_{\tilde{\theta}:|\tilde{\theta}-\theta|<n^{-1/2+\epsilon}} \overline{\nu}_{\theta, \tilde{\theta},n}(G(n)^C)=0.$$
\end{lemma}
\begin{proof}Notice that $G(n)$ is the set of all those patterns counts ${\bf m}=(m_{\alpha^{(1)}}, \dots, m_{\alpha^{(k)}})$ such that the following conditions are satisfied
\begin{enumerate}
\item $\alpha^{(i)}$ does not contain more than $n^\gamma$ consecutive $0$s for every $i=1,\dots, k$;
\item $\sum_{i=1}^{k}m_{\alpha^{(i)}}|\alpha^{(i)}| + (k-1)n^{\gamma} \leq n$.
\end{enumerate}
We recall that $|\alpha|$ is the length of the pattern $\alpha.$

\bigskip Let us consider a positive number $\eta<1-\gamma$ and notice that one has $\tilde{G}(n) \subseteq G(n)$, where $\tilde{G}(n)$ is the set of all those patterns counts ${\bf m}=(m_{\alpha^{(1)}}, \dots, m_{\alpha^{(k)}})$ such that
\begin{enumerate}
\item $|\alpha^{(i)}| \leq \eta \log_2(n)$ for every $i=1,\dots, k$ and
\item $\sum_{i=1}^{k}m_{\alpha^{(i)}}|\alpha^{(i)}|+(k-1)n^\gamma \leq n$.
\end{enumerate}
We denote by $A(n)$ the set of patterns satisfying $1.$ and $B(n)$ the set of patterns satisfying $2.$, hence $\tilde{G}(n)=A(n) \cap B(n)$. In order to prove the statement, it suffices to show that
$$
\lim_{n \rightarrow +\infty} \sup_{\tilde{\theta}:|\tilde{\theta}-\theta|<n^{-1/2+\epsilon}} \overline{\nu}_{\theta, \tilde{\theta},n}(\tilde{G}(n)^C)=0.$$
Notice that $\tilde{G}(n)^C=A(n)^C \sqcup (B(n)^C\cap A(n))$; let us first show that
$$\lim_{n \rightarrow +\infty} \inf_{\tilde{\theta}:|\tilde{\theta}-\theta|<n^{-1/2+\epsilon}} \overline{\nu}_{\theta, \tilde{\theta},n}(A(n)) \rightarrow 1.$$
Notice that, 
\[\begin{split}
    &\lim_{n \rightarrow +\infty} \inf_{\tilde{\theta}:|\tilde{\theta}-\theta|<n^{-1/2+\epsilon}}\overline{\nu}_{\theta, \tilde{\theta},n}(A(n))=\\
    &\lim_{n \rightarrow +\infty} \inf_{\tilde{\theta}:|\tilde{\theta}-\theta|<n^{-1/2+\epsilon}}e^{-\left (\sum_{|\alpha| > \eta \log_2(n)} \lambda_\alpha \right ) (\sqrt{n}(\theta-\tilde{\theta})-\tau_n)^2}=1,
    \end{split}\]
since, for every $\tilde{\theta}$ such that $|\theta-\tilde{\theta}|<n^{-1/2+\epsilon}$, one has
$$\left (\sum_{|\alpha| > \eta \log_2(n)} \lambda_\alpha \right ) (\sqrt{n}(\theta-\tilde{\theta})-\tau_n)^2 \lesssim n^{\eta \log_2(\lambda) + 3\epsilon}\rightarrow 0$$
if $\epsilon <-\eta \log_2(\lambda)/3$.
Let us now study the probability of $B(n)^C \cap A(n)$: first notice that
\begin{eqnarray*}
&&B(n)^C \cap A(n) \subseteq \\
&&\{ {\bf m} :|\alpha^{(i)}| \leq \eta\log_2(n), \, \sum_{i=1}^{k}m_{\alpha^{(i)}}|\alpha^{(i)}|+2^{\eta \log_2(n)}n^\gamma > n\},
\end{eqnarray*}
because $2^{\eta \log_2(n)}$ upper bounds the cardinality of all the patterns of length smaller or equal than $\eta \log_2(n)$. Therefore
\[\begin{split}
&\overline{\nu}_{\theta, \tilde{\theta},n}(B(n)^C \cap A(n)) \leq \overline{\nu}_{\theta, \tilde{\theta},n}(C(n))
\\
&\leq \overline{\nu}_{\theta, \tilde{\theta},n}\left ( \left\{ {\bf m}:\sum_{|\alpha| \leq \eta\log_2(n)}m_{\alpha} > \frac{n-n^{\eta+ \gamma}}{\eta \log(n)} \right\}\right ).
\end{split}\]
where $C(n)$ is the set of pattern counts such that $|\alpha^{(i)}| \leq \eta\log_2(n)$ and $\sum_{i=1}^{k}m_{\alpha^{(i)}}\eta \log_2(n)+n^{\eta +\gamma}n > n$.
The last term amounts to the probability that a Poisson random variable of parameter $\sum_{|\alpha| \leq \eta \log_2(n)}\lambda_\alpha (\sqrt{n}(\theta-\tilde{\theta})-\tau_n)^2 \lesssim n^{3\epsilon}$ is bigger than something that grows as
$
n/(\eta \log(n)).
$
Such a probability goes to $0$ uniformly in $\tilde{\theta}$ if $3\epsilon<1$ and we are done.

\end{proof}

\vfill
\vspace{20mm}
\mbox{}


\begin{thebibliography}{141}%
\makeatletter
\providecommand \@ifxundefined [1]{%
 \@ifx{#1\undefined}
}%
\providecommand \@ifnum [1]{%
 \ifnum #1\expandafter \@firstoftwo
 \else \expandafter \@secondoftwo
 \fi
}%
\providecommand \@ifx [1]{%
 \ifx #1\expandafter \@firstoftwo
 \else \expandafter \@secondoftwo
 \fi
}%
\providecommand \natexlab [1]{#1}%
\providecommand \enquote  [1]{``#1''}%
\providecommand \bibnamefont  [1]{#1}%
\providecommand \bibfnamefont [1]{#1}%
\providecommand \citenamefont [1]{#1}%
\providecommand \href@noop [0]{\@secondoftwo}%
\providecommand \href [0]{\begingroup \@sanitize@url \@href}%
\providecommand \@href[1]{\@@startlink{#1}\@@href}%
\providecommand \@@href[1]{\endgroup#1\@@endlink}%
\providecommand \@sanitize@url [0]{\catcode `\\12\catcode `\$12\catcode
  `\&12\catcode `\#12\catcode `\^12\catcode `\_12\catcode `\%12\relax}%
\providecommand \@@startlink[1]{}%
\providecommand \@@endlink[0]{}%
\providecommand \url  [0]{\begingroup\@sanitize@url \@url }%
\providecommand \@url [1]{\endgroup\@href {#1}{\urlprefix }}%
\providecommand \urlprefix  [0]{URL }%
\providecommand \Eprint [0]{\href }%
\providecommand \doibase [0]{http://dx.doi.org/}%
\providecommand \selectlanguage [0]{\@gobble}%
\providecommand \bibinfo  [0]{\@secondoftwo}%
\providecommand \bibfield  [0]{\@secondoftwo}%
\providecommand \translation [1]{[#1]}%
\providecommand \BibitemOpen [0]{}%
\providecommand \bibitemStop [0]{}%
\providecommand \bibitemNoStop [0]{.\EOS\space}%
\providecommand \EOS [0]{\spacefactor3000\relax}%
\providecommand \BibitemShut  [1]{\csname bibitem#1\endcsname}%
\let\auto@bib@innerbib\@empty
\bibitem [{\citenamefont {Girotti}\ \emph {et~al.}(2024)\citenamefont
  {Girotti}, \citenamefont {Godley},\ and\ \citenamefont
  {Gu\c{t}\u{a}}}]{GiGoGu}%
  \BibitemOpen
  \bibfield  {author} {\bibinfo {author} {\bibfnamefont {F.}~\bibnamefont
  {Girotti}}, \bibinfo {author} {\bibfnamefont {A.}~\bibnamefont {Godley}}, \
  and\ \bibinfo {author} {\bibfnamefont {M.}~\bibnamefont {Gu\c{t}\u{a}}},\
  }\href {\doibase 10.1088/1751-8121/ad4c2b} {\bibfield  {journal} {\bibinfo
  {journal} {J. Phys. A: Mathematical and Theoretical}\ }\textbf {\bibinfo
  {volume} {57}},\ \bibinfo {pages} {245304} (\bibinfo {year}
  {2024})}\BibitemShut {NoStop}%
\bibitem [{\citenamefont {Yang}\ \emph {et~al.}(2023)\citenamefont {Yang},
  \citenamefont {Huelga},\ and\ \citenamefont {Plenio}}]{DayouCounting}%
  \BibitemOpen
  \bibfield  {author} {\bibinfo {author} {\bibfnamefont {D.}~\bibnamefont
  {Yang}}, \bibinfo {author} {\bibfnamefont {S.~F.}\ \bibnamefont {Huelga}}, \
  and\ \bibinfo {author} {\bibfnamefont {M.~B.}\ \bibnamefont {Plenio}},\
  }\href {\doibase 10.1103/PhysRevX.13.031012} {\bibfield  {journal} {\bibinfo
  {journal} {Phys. Rev. X}\ }\textbf {\bibinfo {volume} {13}},\ \bibinfo
  {pages} {031012} (\bibinfo {year} {2023})}\BibitemShut {NoStop}%
\bibitem [{\citenamefont {Helstrom}(1976)}]{Helstrom1976}%
  \BibitemOpen
  \bibfield  {author} {\bibinfo {author} {\bibfnamefont {C.~W.}\ \bibnamefont
  {Helstrom}},\ }\href@noop {} {\emph {\bibinfo {title} {Quantum detection and
  estimation theory}}}\ (\bibinfo  {publisher} {Academic Press},\ \bibinfo
  {year} {1976})\BibitemShut {NoStop}%
\bibitem [{\citenamefont {Holevo}(2011)}]{Holevo2011}%
  \BibitemOpen
  \bibfield  {author} {\bibinfo {author} {\bibfnamefont {A.}~\bibnamefont
  {Holevo}},\ }\href {\doibase 10.1007/978-88-7642-378-9} {\emph {\bibinfo
  {title} {Probabilistic and Statistical Aspects of Quantum Theory}}}\
  (\bibinfo  {publisher} {Edizioni della Normale},\ \bibinfo {year}
  {2011})\BibitemShut {NoStop}%
\bibitem [{\citenamefont {Hayashi}(2005)}]{Hayashi2005}%
  \BibitemOpen
  \bibfield  {author} {\bibinfo {author} {\bibfnamefont {M.}~\bibnamefont
  {Hayashi}},\ }\href {\doibase 10.1142/5630} {\emph {\bibinfo {title}
  {Asymptotic theory of quantum statistical inference: Selected papers}}}\
  (\bibinfo  {publisher} {World Scientific Publishing Co.},\ \bibinfo {year}
  {2005})\BibitemShut {NoStop}%
\bibitem [{\citenamefont {Paris}(2008)}]{Paris2008}%
  \BibitemOpen
  \bibfield  {author} {\bibinfo {author} {\bibfnamefont {M.~G.}\ \bibnamefont
  {Paris}},\ }\href {\doibase 10.1142/S0219749909004839} {\bibfield  {journal}
  {\bibinfo  {journal} {International Journal of Quantum Information}\ }\textbf
  {\bibinfo {volume} {7}},\ \bibinfo {pages} {125} (\bibinfo {year}
  {2008})}\BibitemShut {NoStop}%
\bibitem [{\citenamefont {Tóth}\ and\ \citenamefont
  {Apellaniz}(2014)}]{TothReview}%
  \BibitemOpen
  \bibfield  {author} {\bibinfo {author} {\bibfnamefont {G.}~\bibnamefont
  {Tóth}}\ and\ \bibinfo {author} {\bibfnamefont {I.}~\bibnamefont
  {Apellaniz}},\ }\href {\doibase 10.1088/1751-8113/47/42/424006} {\bibfield
  {journal} {\bibinfo  {journal} {Journal of Physics A: Mathematical and
  Theoretical}\ }\textbf {\bibinfo {volume} {47}},\ \bibinfo {pages} {424006}
  (\bibinfo {year} {2014})}\BibitemShut {NoStop}%
\bibitem [{\citenamefont {Demkowicz-Dobrzański}\ \emph
  {et~al.}(2020)\citenamefont {Demkowicz-Dobrzański}, \citenamefont
  {Gorecki},\ and\ \citenamefont {Gu\c{t}\u{a}}}]{RafalReview}%
  \BibitemOpen
  \bibfield  {author} {\bibinfo {author} {\bibfnamefont {R.}~\bibnamefont
  {Demkowicz-Dobrzański}}, \bibinfo {author} {\bibfnamefont {W.}~\bibnamefont
  {Gorecki}}, \ and\ \bibinfo {author} {\bibfnamefont {M.}~\bibnamefont
  {Gu\c{t}\u{a}}},\ }\href {\doibase 10.1088/1751-8121/ab8ef3} {\bibfield
  {journal} {\bibinfo  {journal} {Journal of Physics A: Mathematical and
  Theoretical}\ }\textbf {\bibinfo {volume} {53}},\ \bibinfo {pages} {363001}
  (\bibinfo {year} {2020})}\BibitemShut {NoStop}%
\bibitem [{\citenamefont {Banaszek}\ \emph {et~al.}(2013)\citenamefont
  {Banaszek}, \citenamefont {Cramer},\ and\ \citenamefont {Gross}}]{Tomo2}%
  \BibitemOpen
  \bibfield  {author} {\bibinfo {author} {\bibfnamefont {K.}~\bibnamefont
  {Banaszek}}, \bibinfo {author} {\bibfnamefont {M.}~\bibnamefont {Cramer}}, \
  and\ \bibinfo {author} {\bibfnamefont {D.}~\bibnamefont {Gross}},\ }\href
  {\doibase 10.1088/1367-2630/15/12/125020} {\bibfield  {journal} {\bibinfo
  {journal} {New Journal of Physics}\ }\textbf {\bibinfo {volume} {15}},\
  \bibinfo {pages} {125020} (\bibinfo {year} {2013})}\BibitemShut {NoStop}%
\bibitem [{\citenamefont {Albarelli}\ \emph {et~al.}(2020)\citenamefont
  {Albarelli}, \citenamefont {Barbieri}, \citenamefont {Genoni},\ and\
  \citenamefont {Gianani}}]{Albarelli2020}%
  \BibitemOpen
  \bibfield  {author} {\bibinfo {author} {\bibfnamefont {F.}~\bibnamefont
  {Albarelli}}, \bibinfo {author} {\bibfnamefont {M.}~\bibnamefont {Barbieri}},
  \bibinfo {author} {\bibfnamefont {M.}~\bibnamefont {Genoni}}, \ and\ \bibinfo
  {author} {\bibfnamefont {I.}~\bibnamefont {Gianani}},\ }\href {\doibase 10.1016/j.physleta.2020.126311} {\bibfield  {journal} {\bibinfo  {journal}
  {Physics Letters A}\ }\textbf {\bibinfo {volume} {384}},\ \bibinfo {pages}
  {126311} (\bibinfo {year} {2020})}\BibitemShut {NoStop}%
\bibitem [{\citenamefont {Sidhu}\ and\ \citenamefont {Kok}(2020)}]{Sidhu_2020}%
  \BibitemOpen
  \bibfield  {author} {\bibinfo {author} {\bibfnamefont {J.~S.}\ \bibnamefont
  {Sidhu}}\ and\ \bibinfo {author} {\bibfnamefont {P.}~\bibnamefont {Kok}},\
  }\href {\doibase 10.1116/1.5119961} {\bibfield  {journal} {\bibinfo
  {journal} {{AVS} Quantum Science}\ }\textbf {\bibinfo {volume} {2}},\
  \bibinfo {pages} {014701} (\bibinfo {year} {2020})}\BibitemShut {NoStop}%
\bibitem [{\citenamefont {Giovannetti}\ \emph {et~al.}(2005)\citenamefont
  {Giovannetti}, \citenamefont {Lloyd},\ and\ \citenamefont
  {Maccone}}]{Metrology1}%
  \BibitemOpen
  \bibfield  {author} {\bibinfo {author} {\bibfnamefont {V.}~\bibnamefont
  {Giovannetti}}, \bibinfo {author} {\bibfnamefont {S.}~\bibnamefont {Lloyd}},
  \ and\ \bibinfo {author} {\bibfnamefont {L.}~\bibnamefont {Maccone}},\ }\href
  {\doibase 10.1103/PhysRevLett.96.010401} {\bibfield  {journal} {\bibinfo
  {journal} {Physical Review Letters}\ }\textbf {\bibinfo {volume} {96}},\
  \bibinfo {pages} {010401} (\bibinfo {year} {2005})}\BibitemShut {NoStop}%
\bibitem [{\citenamefont {Fujiwara}\ and\ \citenamefont
  {Imai}(2008)}]{Fujiwara2008}%
  \BibitemOpen
  \bibfield  {author} {\bibinfo {author} {\bibfnamefont {A.}~\bibnamefont
  {Fujiwara}}\ and\ \bibinfo {author} {\bibfnamefont {H.}~\bibnamefont
  {Imai}},\ }\href {\doibase 10.1088/1751-8113/41/25/255304} {\bibfield
  {journal} {\bibinfo  {journal} {Journal of Physics A: Mathematical and
  Theoretical}\ }\textbf {\bibinfo {volume} {41}},\ \bibinfo {pages} {255304}
  (\bibinfo {year} {2008})}\BibitemShut {NoStop}%
\bibitem [{\citenamefont {Giovannetti}\ \emph {et~al.}(2011)\citenamefont
  {Giovannetti}, \citenamefont {Lloyd},\ and\ \citenamefont
  {Maccone}}]{Giovannetti2011}%
  \BibitemOpen
  \bibfield  {author} {\bibinfo {author} {\bibfnamefont {V.}~\bibnamefont
  {Giovannetti}}, \bibinfo {author} {\bibfnamefont {S.}~\bibnamefont {Lloyd}},
  \ and\ \bibinfo {author} {\bibfnamefont {L.}~\bibnamefont {Maccone}},\ }\href
  {\doibase 10.1038/nphoton.2011.35} {\bibfield  {journal} {\bibinfo  {journal}
  {Nature Photonics}\ }\textbf {\bibinfo {volume} {5}},\ \bibinfo {pages} {222}
  (\bibinfo {year} {2011})}\BibitemShut {NoStop}%
\bibitem [{\citenamefont {Escher}\ \emph {et~al.}(2011)\citenamefont {Escher},
  \citenamefont {de~Matos~Filho},\ and\ \citenamefont {Davidovich}}]{Escher11}%
  \BibitemOpen
  \bibfield  {author} {\bibinfo {author} {\bibfnamefont {B.}~\bibnamefont
  {Escher}}, \bibinfo {author} {\bibfnamefont {R.}~\bibnamefont
  {de~Matos~Filho}}, \ and\ \bibinfo {author} {\bibfnamefont {L.}~\bibnamefont
  {Davidovich}},\ }\href {\doibase 10.1038/nphys1958}
  {\bibfield  {journal} {\bibinfo  {journal} {Nature Phys.}\ }\textbf {\bibinfo
  {volume} {7}},\ \bibinfo {pages} {406–411} (\bibinfo {year}
  {2011})}\BibitemShut {NoStop}%
\bibitem [{\citenamefont {Demkowicz-Dobrzański}\ \emph
  {et~al.}(2012)\citenamefont {Demkowicz-Dobrzański}, \citenamefont
  {Kołodyński},\ and\ \citenamefont {Gu\c{t}\u{a}}}]{Metrology2}%
  \BibitemOpen
  \bibfield  {author} {\bibinfo {author} {\bibfnamefont {R.}~\bibnamefont
  {Demkowicz-Dobrzański}}, \bibinfo {author} {\bibfnamefont {J.}~\bibnamefont
  {Kołodyński}}, \ and\ \bibinfo {author} {\bibfnamefont {M.}~\bibnamefont
  {Gu\c{t}\u{a}}},\ }\href {\doibase 10.1038/ncomms2067} {\bibfield  {journal}
  {\bibinfo  {journal} {Nature Communications}\ }\textbf {\bibinfo {volume}
  {3}},\ \bibinfo {pages} {1063} (\bibinfo {year} {2012})}\BibitemShut
  {NoStop}%
\bibitem [{\citenamefont {Girolami}\ \emph {et~al.}(2014)\citenamefont
  {Girolami}, \citenamefont {Souza}, \citenamefont {Giovannetti}, \citenamefont
  {Tufarelli}, \citenamefont {Filgueiras}, \citenamefont {Sarthour},
  \citenamefont {Soares-Pinto}, \citenamefont {Oliveira},\ and\ \citenamefont
  {Adesso}}]{Girolami14}%
  \BibitemOpen
  \bibfield  {author} {\bibinfo {author} {\bibfnamefont {D.}~\bibnamefont
  {Girolami}}, \bibinfo {author} {\bibfnamefont {A.~M.}\ \bibnamefont {Souza}},
  \bibinfo {author} {\bibfnamefont {V.}~\bibnamefont {Giovannetti}}, \bibinfo
  {author} {\bibfnamefont {T.}~\bibnamefont {Tufarelli}}, \bibinfo {author}
  {\bibfnamefont {J.~G.}\ \bibnamefont {Filgueiras}}, \bibinfo {author}
  {\bibfnamefont {R.~S.}\ \bibnamefont {Sarthour}}, \bibinfo {author}
  {\bibfnamefont {D.~O.}\ \bibnamefont {Soares-Pinto}}, \bibinfo {author}
  {\bibfnamefont {I.~S.}\ \bibnamefont {Oliveira}}, \ and\ \bibinfo {author}
  {\bibfnamefont {G.}~\bibnamefont {Adesso}},\ }
  \href {\doibase 10.1103/PhysRevLett.112.210401} {\bibfield  {journal} {\bibinfo  {journal}
  {Phys. Rev. Lett.}\ }\textbf {\bibinfo {volume} {112}},\ \bibinfo {pages}
  {210401} (\bibinfo {year} {2014})}\BibitemShut {NoStop}%
\bibitem [{\citenamefont {Smirne}\ \emph {et~al.}(2016)\citenamefont {Smirne},
  \citenamefont {Ko\l{}ody\ifmmode~\acute{n}\else \'{n}\fi{}ski}, \citenamefont
  {Huelga},\ and\ \citenamefont {Demkowicz-Dobrzański}}]{Smirne16}%
  \BibitemOpen
  \bibfield  {author} {\bibinfo {author} {\bibfnamefont {A.}~\bibnamefont
  {Smirne}}, \bibinfo {author} {\bibfnamefont {J.}~\bibnamefont
  {Ko\l{}ody\ifmmode~\acute{n}\else \'{n}\fi{}ski}}, \bibinfo {author}
  {\bibfnamefont {S.~F.}\ \bibnamefont {Huelga}}, \ and\ \bibinfo {author}
  {\bibfnamefont {R.}~\bibnamefont {Demkowicz-Dobrzański}},\ }\href {\doibase 10.1103/PhysRevLett.116.120801} {\bibfield  {journal} {\bibinfo  {journal}
  {Phys. Rev. Lett.}\ }\textbf {\bibinfo {volume} {116}},\ \bibinfo {pages}
  {120801} (\bibinfo {year} {2016})}\BibitemShut {NoStop}%
\bibitem [{\citenamefont {Seveso}\ \emph {et~al.}(2017)\citenamefont {Seveso},
  \citenamefont {Rossi},\ and\ \citenamefont {Paris}}]{Seveso2017}%
  \BibitemOpen
  \bibfield  {author} {\bibinfo {author} {\bibfnamefont {L.}~\bibnamefont
  {Seveso}}, \bibinfo {author} {\bibfnamefont {M.~A.~C.}\ \bibnamefont
  {Rossi}}, \ and\ \bibinfo {author} {\bibfnamefont {M.~G.~A.}\ \bibnamefont
  {Paris}},\ }\href {\doibase 10.1103/PhysRevA.95.012111} {\bibfield  {journal}
  {\bibinfo  {journal} {Phys. Rev. A}\ }\textbf {\bibinfo {volume} {95}},\
  \bibinfo {pages} {012111} (\bibinfo {year} {2017})}\BibitemShut {NoStop}%
\bibitem [{\citenamefont {Haase}\ \emph {et~al.}(2018)\citenamefont {Haase},
  \citenamefont {Smirne}, \citenamefont {Kolodynski}, \citenamefont
  {Demkowicz-Dobrzanski},\ and\ \citenamefont {Huelga}}]{Haase18}%
  \BibitemOpen
  \bibfield  {author} {\bibinfo {author} {\bibfnamefont {J.~F.}\ \bibnamefont
  {Haase}}, \bibinfo {author} {\bibfnamefont {A.}~\bibnamefont {Smirne}},
  \bibinfo {author} {\bibfnamefont {J.}~\bibnamefont {Kolodynski}}, \bibinfo
  {author} {\bibfnamefont {R.}~\bibnamefont {Demkowicz-Dobrzanski}}, \ and\
  \bibinfo {author} {\bibfnamefont {S.~F.}\ \bibnamefont {Huelga}},\ }\href
  {\doibase 10.1088/1367-2630/aab67f} {\bibfield  {journal} {\bibinfo
  {journal} {New Journal of Physics}\ }\textbf {\bibinfo {volume} {20}},\
  \bibinfo {pages} {053009} (\bibinfo {year} {2018})}\BibitemShut {NoStop}%
\bibitem [{\citenamefont {Pang}\ and\ \citenamefont {Brun}(2014)}]{Metrology4}%
  \BibitemOpen
  \bibfield  {author} {\bibinfo {author} {\bibfnamefont {S.}~\bibnamefont
  {Pang}}\ and\ \bibinfo {author} {\bibfnamefont {T.~A.}\ \bibnamefont
  {Brun}},\ }\href {\doibase 10.1103/PHYSREVA.90.022117} {\bibfield  {journal}
  {\bibinfo  {journal} {Physical Review A - Atomic, Molecular, and Optical
  Physics}\ }\textbf {\bibinfo {volume} {90}},\ \bibinfo {pages} {022117}
  (\bibinfo {year} {2014})}\BibitemShut {NoStop}%
\bibitem [{\citenamefont {Rossi}\ \emph {et~al.}(2020)\citenamefont {Rossi},
  \citenamefont {Albarelli}, \citenamefont {Tamascelli},\ and\ \citenamefont
  {Genoni}}]{Rossi20}%
  \BibitemOpen
  \bibfield  {author} {\bibinfo {author} {\bibfnamefont {M.~A.~C.}\
  \bibnamefont {Rossi}}, \bibinfo {author} {\bibfnamefont {F.}~\bibnamefont
  {Albarelli}}, \bibinfo {author} {\bibfnamefont {D.}~\bibnamefont
  {Tamascelli}}, \ and\ \bibinfo {author} {\bibfnamefont {M.~G.}\ \bibnamefont
  {Genoni}},\ }\href {\doibase 10.1103/PhysRevLett.125.200505} {\bibfield
  {journal} {\bibinfo  {journal} {Phys. Rev. Lett.}\ }\textbf {\bibinfo
  {volume} {125}},\ \bibinfo {pages} {200505} (\bibinfo {year}
  {2020})}\BibitemShut {NoStop}%
\bibitem [{\citenamefont {Zhou}\ and\ \citenamefont {Jiang}(2021)}]{Sisi21}%
  \BibitemOpen
  \bibfield  {author} {\bibinfo {author} {\bibfnamefont {S.}~\bibnamefont
  {Zhou}}\ and\ \bibinfo {author} {\bibfnamefont {L.}~\bibnamefont {Jiang}},\
  }\href {\doibase 10.1103/PRXQuantum.2.010343} {\bibfield  {journal} {\bibinfo
   {journal} {PRX Quantum}\ }\textbf {\bibinfo {volume} {2}},\ \bibinfo {pages}
  {010343} (\bibinfo {year} {2021})}\BibitemShut {NoStop}%
\bibitem [{\citenamefont {Tsang}\ \emph {et~al.}(2016)\citenamefont {Tsang},
  \citenamefont {Nair},\ and\ \citenamefont {Lu}}]{Tsang16}%
  \BibitemOpen
  \bibfield  {author} {\bibinfo {author} {\bibfnamefont {M.}~\bibnamefont
  {Tsang}}, \bibinfo {author} {\bibfnamefont {R.}~\bibnamefont {Nair}}, \ and\
  \bibinfo {author} {\bibfnamefont {X.-M.}\ \bibnamefont {Lu}},\ }\href {\doibase 10.1103/PhysRevX.6.031033} {\bibfield  {journal} {\bibinfo
  {journal} {Phys. Rev. X}\ }\textbf {\bibinfo {volume} {6}},\ \bibinfo {pages}
  {031033} (\bibinfo {year} {2016})}\BibitemShut {NoStop}%
\bibitem [{\citenamefont {Tsang}(2021)}]{Tsang21}%
  \BibitemOpen
  \bibfield  {author} {\bibinfo {author} {\bibfnamefont {M.}~\bibnamefont
  {Tsang}},\ }\href {\doibase 10.1103/PhysRevA.104.052411} {\bibfield
  {journal} {\bibinfo  {journal} {Phys. Rev. A}\ }\textbf {\bibinfo {volume}
  {104}},\ \bibinfo {pages} {052411} (\bibinfo {year} {2021})}\BibitemShut
  {NoStop}%
\bibitem [{\citenamefont {Lupo}\ \emph {et~al.}(2020)\citenamefont {Lupo},
  \citenamefont {Huang},\ and\ \citenamefont {Kok}}]{Lupo20}%
  \BibitemOpen
  \bibfield  {author} {\bibinfo {author} {\bibfnamefont {C.}~\bibnamefont
  {Lupo}}, \bibinfo {author} {\bibfnamefont {Z.}~\bibnamefont {Huang}}, \ and\
  \bibinfo {author} {\bibfnamefont {P.}~\bibnamefont {Kok}},\ }\href {\doibase 10.1103/PhysRevLett.124.080503} {\bibfield  {journal} {\bibinfo  {journal}
  {Phys. Rev. Lett.}\ }\textbf {\bibinfo {volume} {124}},\ \bibinfo {pages}
  {080503} (\bibinfo {year} {2020})}\BibitemShut {NoStop}%
\bibitem [{\citenamefont {Fiderer}\ \emph {et~al.}(2021)\citenamefont
  {Fiderer}, \citenamefont {Tufarelli}, \citenamefont {Piano},\ and\
  \citenamefont {Adesso}}]{Fiderer21}%
  \BibitemOpen
  \bibfield  {author} {\bibinfo {author} {\bibfnamefont {L.~J.}\ \bibnamefont
  {Fiderer}}, \bibinfo {author} {\bibfnamefont {T.}~\bibnamefont {Tufarelli}},
  \bibinfo {author} {\bibfnamefont {S.}~\bibnamefont {Piano}}, \ and\ \bibinfo
  {author} {\bibfnamefont {G.}~\bibnamefont {Adesso}},\ }
  \href {\doibase 10.1103/PRXQuantum.2.020308} {\bibfield  {journal} {\bibinfo  {journal} {PRX
  Quantum}\ }\textbf {\bibinfo {volume} {2}},\ \bibinfo {pages} {020308}
  (\bibinfo {year} {2021})}\BibitemShut {NoStop}%
\bibitem [{\citenamefont {Oh}\ \emph {et~al.}(2021)\citenamefont {Oh},
  \citenamefont {Zhou}, \citenamefont {Wong},\ and\ \citenamefont
  {Jiang}}]{Oh21}%
  \BibitemOpen
  \bibfield  {author} {\bibinfo {author} {\bibfnamefont {C.}~\bibnamefont
  {Oh}}, \bibinfo {author} {\bibfnamefont {S.}~\bibnamefont {Zhou}}, \bibinfo
  {author} {\bibfnamefont {Y.}~\bibnamefont {Wong}}, \ and\ \bibinfo {author}
  {\bibfnamefont {L.}~\bibnamefont {Jiang}},\ }
  \href {\doibase 10.1103/PhysRevLett.126.120502} {\bibfield  {journal} {\bibinfo  {journal}
  {Phys. Rev. Lett.}\ }\textbf {\bibinfo {volume} {126}},\ \bibinfo {pages}
  {120502} (\bibinfo {year} {2021})}\BibitemShut {NoStop}%
\bibitem [{\citenamefont {Tsang}\ \emph {et~al.}(2011)\citenamefont {Tsang},
  \citenamefont {Wiseman},\ and\ \citenamefont {Caves}}]{TWC11}%
  \BibitemOpen
  \bibfield  {author} {\bibinfo {author} {\bibfnamefont {M.}~\bibnamefont
  {Tsang}}, \bibinfo {author} {\bibfnamefont {H.~M.}\ \bibnamefont {Wiseman}},
  \ and\ \bibinfo {author} {\bibfnamefont {C.}~\bibnamefont {Caves}},\ }\href{\doibase 10.1103/PhysRevLett.106.090401} {\bibfield
  {journal} {\bibinfo  {journal} {Phys. Rev. Lett.}\ }\textbf {\bibinfo
  {volume} {106}},\ \bibinfo {pages} {090401} (\bibinfo {year}
  {2011})}\BibitemShut {NoStop}%
\bibitem [{\citenamefont {Berry}\ \emph {et~al.}(2015)\citenamefont {Berry},
  \citenamefont {Tsang}, \citenamefont {Hall},\ and\ \citenamefont
  {Wiseman}}]{Berry2015}%
  \BibitemOpen
  \bibfield  {author} {\bibinfo {author} {\bibfnamefont {D.~W.}\ \bibnamefont
  {Berry}}, \bibinfo {author} {\bibfnamefont {M.}~\bibnamefont {Tsang}},
  \bibinfo {author} {\bibfnamefont {M.~J.}\ \bibnamefont {Hall}}, \ and\
  \bibinfo {author} {\bibfnamefont {H.~M.}\ \bibnamefont {Wiseman}},\ }\href
  {\doibase 10.1103/PhysRevX.5.031018} {\bibfield  {journal} {\bibinfo
  {journal} {Physical Review X}\ }\textbf {\bibinfo {volume} {5}} (\bibinfo
  {year} {2015}),\ 10.1103/PhysRevX.5.031018}\BibitemShut {NoStop}%
\bibitem [{\citenamefont {Ng}\ \emph {et~al.}(2016)\citenamefont {Ng},
  \citenamefont {Ang}, \citenamefont {Wheatley}, \citenamefont {Yonezawa},
  \citenamefont {Furusawa}, \citenamefont {Huntington},\ and\ \citenamefont
  {Tsang}}]{Ng16}%
  \BibitemOpen
  \bibfield  {author} {\bibinfo {author} {\bibfnamefont {S.}~\bibnamefont
  {Ng}}, \bibinfo {author} {\bibfnamefont {S.~Z.}\ \bibnamefont {Ang}},
  \bibinfo {author} {\bibfnamefont {T.~A.}\ \bibnamefont {Wheatley}}, \bibinfo
  {author} {\bibfnamefont {H.}~\bibnamefont {Yonezawa}}, \bibinfo {author}
  {\bibfnamefont {A.}~\bibnamefont {Furusawa}}, \bibinfo {author}
  {\bibfnamefont {E.~H.}\ \bibnamefont {Huntington}}, \ and\ \bibinfo {author}
  {\bibfnamefont {M.}~\bibnamefont {Tsang}},\ }
  \href {\doibase 10.1103/PhysRevA.93.042121} {\bibfield  {journal} {\bibinfo  {journal} {Phys.
  Rev. A}\ }\textbf {\bibinfo {volume} {93}},\ \bibinfo {pages} {042121}
  (\bibinfo {year} {2016})}\BibitemShut {NoStop}%
\bibitem [{\citenamefont {Norris}\ \emph {et~al.}(2016)\citenamefont {Norris},
  \citenamefont {Paz-Silva},\ and\ \citenamefont {Viola}}]{Norris16}%
  \BibitemOpen
  \bibfield  {author} {\bibinfo {author} {\bibfnamefont {L.~M.}\ \bibnamefont
  {Norris}}, \bibinfo {author} {\bibfnamefont {G.~A.}\ \bibnamefont
  {Paz-Silva}}, \ and\ \bibinfo {author} {\bibfnamefont {L.}~\bibnamefont
  {Viola}},\ }\href {\doibase 10.1103/PhysRevLett.116.150503} {\bibfield
  {journal} {\bibinfo  {journal} {Phys. Rev. Lett.}\ }\textbf {\bibinfo
  {volume} {116}},\ \bibinfo {pages} {150503} (\bibinfo {year}
  {2016})}\BibitemShut {NoStop}%
\bibitem [{\citenamefont {Shi}\ and\ \citenamefont {Zhuang}(2023)}]{Shi23}%
  \BibitemOpen
  \bibfield  {author} {\bibinfo {author} {\bibfnamefont {H.}~\bibnamefont
  {Shi}}\ and\ \bibinfo {author} {\bibfnamefont {Q.}~\bibnamefont {Zhuang}},\
  }\href{\doibase 10.1038/s41534-023-00693-w} {\bibfield  {journal} {\bibinfo  {journal} {npj Quantum Inf}\
  }\textbf {\bibinfo {volume} {9}},\ \bibinfo {pages} {27} (\bibinfo {year}
  {2023})}\BibitemShut {NoStop}%
\bibitem [{\citenamefont {Sung}\ \emph {et~al.}(2019)\citenamefont {Sung},
  \citenamefont {Beaudoin}, \citenamefont {Norris}, \citenamefont {Yan},
  \citenamefont {Kim}, \citenamefont {Qiu}, \citenamefont {von L{\"u}pke},
  \citenamefont {Yoder}, \citenamefont {Orlando}, \citenamefont {Gustavsson}
  \emph {et~al.}}]{Sung19}%
  \BibitemOpen
  \bibfield  {author} {\bibinfo {author} {\bibfnamefont {Y.}~\bibnamefont
  {Sung}}, \bibinfo {author} {\bibfnamefont {F.}~\bibnamefont {Beaudoin}},
  \bibinfo {author} {\bibfnamefont {L.~M.}\ \bibnamefont {Norris}}, \bibinfo
  {author} {\bibfnamefont {F.}~\bibnamefont {Yan}}, \bibinfo {author}
  {\bibfnamefont {D.~K.}\ \bibnamefont {Kim}}, \bibinfo {author} {\bibfnamefont
  {J.~Y.}\ \bibnamefont {Qiu}}, \bibinfo {author} {\bibfnamefont
  {U.}~\bibnamefont {von L{\"u}pke}}, \bibinfo {author} {\bibfnamefont {J.~L.}\
  \bibnamefont {Yoder}}, \bibinfo {author} {\bibfnamefont {T.~P.}\ \bibnamefont
  {Orlando}}, \bibinfo {author} {\bibfnamefont {S.}~\bibnamefont {Gustavsson}},
   \emph {et~al.},\ }\href {\doibase 10.1038/s41467-019-11699-4} {\bibfield
  {journal} {\bibinfo  {journal} {Nature communications}\ }\textbf {\bibinfo
  {volume} {10}},\ \bibinfo {pages} {3715} (\bibinfo {year}
  {2019})}\BibitemShut {NoStop}%
\bibitem [{\citenamefont {Tsang}(2023)}]{Tsang23}%
  \BibitemOpen
  \bibfield  {author} {\bibinfo {author} {\bibfnamefont {M.}~\bibnamefont
  {Tsang}},\ }\href {\doibase 10.1103/PhysRevA.107.012611} {\bibfield
  {journal} {\bibinfo  {journal} {Phys. Rev. A}\ }\textbf {\bibinfo {volume}
  {107}},\ \bibinfo {pages} {012611} (\bibinfo {year} {2023})}\BibitemShut
  {NoStop}%
\bibitem [{\citenamefont {Dowling}\ and\ \citenamefont
  {Seshadreesan}(2015)}]{Dowling}%
  \BibitemOpen
  \bibfield  {author} {\bibinfo {author} {\bibfnamefont {J.~P.}\ \bibnamefont
  {Dowling}}\ and\ \bibinfo {author} {\bibfnamefont {K.~P.}\ \bibnamefont
  {Seshadreesan}},\ }\href {\doibase 10.1109/JLT.2014.2386795} {\bibfield
  {journal} {\bibinfo  {journal} {Journal of Lightwave Technology}\ }\textbf
  {\bibinfo {volume} {33}},\ \bibinfo {pages} {2359} (\bibinfo {year}
  {2015})}\BibitemShut {NoStop}%
\bibitem [{\citenamefont {Degen}\ \emph {et~al.}(2017)\citenamefont {Degen},
  \citenamefont {Reinhard},\ and\ \citenamefont {Cappellaro}}]{Degen2017}%
  \BibitemOpen
  \bibfield  {author} {\bibinfo {author} {\bibfnamefont {C.~L.}\ \bibnamefont
  {Degen}}, \bibinfo {author} {\bibfnamefont {F.}~\bibnamefont {Reinhard}}, \
  and\ \bibinfo {author} {\bibfnamefont {P.}~\bibnamefont {Cappellaro}},\
  }\href {\doibase 10.1103/RevModPhys.89.035002} {\bibfield  {journal}
  {\bibinfo  {journal} {Reviews of Modern Physics}\ }\textbf {\bibinfo {volume}
  {89}},\ \bibinfo {pages} {035002} (\bibinfo {year} {2017})}\BibitemShut
  {NoStop}%
\bibitem [{\citenamefont {Pezz\`e}\ \emph {et~al.}(2018)\citenamefont
  {Pezz\`e}, \citenamefont {Smerzi}, \citenamefont {Oberthaler}, \citenamefont
  {Schmied},\ and\ \citenamefont {Treutlein}}]{Pezze18}%
  \BibitemOpen
  \bibfield  {author} {\bibinfo {author} {\bibfnamefont {L.}~\bibnamefont
  {Pezz\`e}}, \bibinfo {author} {\bibfnamefont {A.}~\bibnamefont {Smerzi}},
  \bibinfo {author} {\bibfnamefont {M.~K.}\ \bibnamefont {Oberthaler}},
  \bibinfo {author} {\bibfnamefont {R.}~\bibnamefont {Schmied}}, \ and\
  \bibinfo {author} {\bibfnamefont {P.}~\bibnamefont {Treutlein}},\ }\href
  {\doibase 10.1103/RevModPhys.90.035005} {\bibfield  {journal} {\bibinfo
  {journal} {Rev. Mod. Phys.}\ }\textbf {\bibinfo {volume} {90}},\ \bibinfo
  {pages} {035005} (\bibinfo {year} {2018})}\BibitemShut {NoStop}%
\bibitem [{\citenamefont {Marciniak}\ \emph {et~al.}(2022)\citenamefont
  {Marciniak}, \citenamefont {Feldker}, \citenamefont {Pogorelov},
  \citenamefont {Kaubruegger}, \citenamefont {Vasilyev}, \citenamefont {van
  Bijnen}, \citenamefont {Schindler}, \citenamefont {Zoller}, \citenamefont
  {Blatt},\ and\ \citenamefont {Monz}}]{Marciniak2022}%
  \BibitemOpen
  \bibfield  {author} {\bibinfo {author} {\bibfnamefont {C.~D.}\ \bibnamefont
  {Marciniak}}, \bibinfo {author} {\bibfnamefont {T.}~\bibnamefont {Feldker}},
  \bibinfo {author} {\bibfnamefont {I.}~\bibnamefont {Pogorelov}}, \bibinfo
  {author} {\bibfnamefont {R.}~\bibnamefont {Kaubruegger}}, \bibinfo {author}
  {\bibfnamefont {D.~V.}\ \bibnamefont {Vasilyev}}, \bibinfo {author}
  {\bibfnamefont {R.}~\bibnamefont {van Bijnen}}, \bibinfo {author}
  {\bibfnamefont {P.}~\bibnamefont {Schindler}}, \bibinfo {author}
  {\bibfnamefont {P.}~\bibnamefont {Zoller}}, \bibinfo {author} {\bibfnamefont
  {R.}~\bibnamefont {Blatt}}, \ and\ \bibinfo {author} {\bibfnamefont
  {T.}~\bibnamefont {Monz}},\ }\href {\doibase 10.1038/s41586-022-04435-4}
  {\bibfield  {journal} {\bibinfo  {journal} {Nature}\ }\textbf {\bibinfo
  {volume} {603}},\ \bibinfo {pages} {604} (\bibinfo {year}
  {2022})}\BibitemShut {NoStop}%
\bibitem [{\citenamefont {Zwick}\ and\ \citenamefont
  {Álvarez}(2023)}]{Zwick2023}%
  \BibitemOpen
  \bibfield  {author} {\bibinfo {author} {\bibfnamefont {A.}~\bibnamefont
  {Zwick}}\ and\ \bibinfo {author} {\bibfnamefont {G.~A.}\ \bibnamefont
  {Álvarez}},\ }\href {\doibase 10.1016/j.jmro.2023.100113}
  {\bibfield  {journal} {\bibinfo  {journal} {Journal of Magnetic Resonance
  Open}\ }\textbf {\bibinfo {volume} {16-17}},\ \bibinfo {pages} {100113}
  (\bibinfo {year} {2023})}\BibitemShut {NoStop}%
\bibitem [{\citenamefont {Robinson}\ \emph {et~al.}(2024)\citenamefont
  {Robinson}, \citenamefont {Miklos},\ and\ \citenamefont {Tso}}]{Robinson24}%
  \BibitemOpen
  \bibfield  {author} {\bibinfo {author} {\bibfnamefont {J.}~\bibnamefont
  {Robinson}}, \bibinfo {author} {\bibfnamefont {M.}~\bibnamefont {Miklos}}, \
  and\ \bibinfo {author} {\bibfnamefont {Y.~e.~a.}\ \bibnamefont {Tso}},\
  }\href {\doibase 10.1038/s41567-023-02310-1} {\bibfield
  {journal} {\bibinfo  {journal} {Nat. Phys.}\ }\textbf {\bibinfo {volume}
  {20}},\ \bibinfo {pages} {208–213} (\bibinfo {year} {2024})}\BibitemShut
  {NoStop}%
\bibitem [{\citenamefont {Jones}\ \emph {et~al.}(2009)\citenamefont {Jones},
  \citenamefont {Karlen}, \citenamefont {Fitzsimons}, \citenamefont {Ardavan},
  \citenamefont {Benjamin}, \citenamefont {Briggs},\ and\ \citenamefont
  {Morton}}]{Jones2009}%
  \BibitemOpen
  \bibfield  {author} {\bibinfo {author} {\bibfnamefont {J.~A.}\ \bibnamefont
  {Jones}}, \bibinfo {author} {\bibfnamefont {S.~D.}\ \bibnamefont {Karlen}},
  \bibinfo {author} {\bibfnamefont {J.}~\bibnamefont {Fitzsimons}}, \bibinfo
  {author} {\bibfnamefont {A.}~\bibnamefont {Ardavan}}, \bibinfo {author}
  {\bibfnamefont {S.~C.}\ \bibnamefont {Benjamin}}, \bibinfo {author}
  {\bibfnamefont {G.~A.~D.}\ \bibnamefont {Briggs}}, \ and\ \bibinfo {author}
  {\bibfnamefont {J.~J.}\ \bibnamefont {Morton}},\ }
  \href {\doibase 10.1126/SCIENCE.1170730} {\bibfield  {journal} {\bibinfo  {journal}
  {Science}\ }\textbf {\bibinfo {volume} {324}},\ \bibinfo {pages} {1166}
  (\bibinfo {year} {2009})}\BibitemShut {NoStop}%
\bibitem [{\citenamefont {Amorós-Binefa}\ and\ \citenamefont
  {Kołodyński}(2021)}]{Jan2021}%
  \BibitemOpen
  \bibfield  {author} {\bibinfo {author} {\bibfnamefont {J.}~\bibnamefont
  {Amorós-Binefa}}\ and\ \bibinfo {author} {\bibfnamefont {J.}~\bibnamefont
  {Kołodyński}},\ }\href {\doibase 10.1088/1367-2630/ac3b71} {\bibfield
  {journal} {\bibinfo  {journal} {New Journal of Physics}\ }\textbf {\bibinfo
  {volume} {23}},\ \bibinfo {pages} {123030} (\bibinfo {year}
  {2021})}\BibitemShut {NoStop}%
\bibitem [{\citenamefont {Brask}\ \emph {et~al.}(2015)\citenamefont {Brask},
  \citenamefont {Chaves},\ and\ \citenamefont {Kolodynski}}]{Brask2015}%
  \BibitemOpen
  \bibfield  {author} {\bibinfo {author} {\bibfnamefont {J.~B.}\ \bibnamefont
  {Brask}}, \bibinfo {author} {\bibfnamefont {R.}~\bibnamefont {Chaves}}, \
  and\ \bibinfo {author} {\bibfnamefont {J.}~\bibnamefont {Kolodynski}},\
  }\href {\doibase 10.1103/PHYSREVX.5.031010} {\bibfield  {journal} {\bibinfo
  {journal} {Physical Review X}\ }\textbf {\bibinfo {volume} {5}},\ \bibinfo
  {pages} {031010} (\bibinfo {year} {2015})}\BibitemShut {NoStop}%
\bibitem [{\citenamefont {Albarelli}\ \emph {et~al.}(2017)\citenamefont
  {Albarelli}, \citenamefont {Rossi}, \citenamefont {Paris},\ and\
  \citenamefont {Genoni}}]{ARPG17}%
  \BibitemOpen
  \bibfield  {author} {\bibinfo {author} {\bibfnamefont {F.}~\bibnamefont
  {Albarelli}}, \bibinfo {author} {\bibfnamefont {M.~A.~C.}\ \bibnamefont
  {Rossi}}, \bibinfo {author} {\bibfnamefont {M.~G.~A.}\ \bibnamefont {Paris}},
  \ and\ \bibinfo {author} {\bibfnamefont {M.~G.}\ \bibnamefont {Genoni}},\
  }\href {\doibase 10.1088/1367-2630/aa9840} {\bibfield
  {journal} {\bibinfo  {journal} {New J. Phys.}\ }\textbf {\bibinfo {volume}
  {19}},\ \bibinfo {pages} {123011} (\bibinfo {year} {2017})}\BibitemShut
  {NoStop}%
\bibitem [{\citenamefont {Aslam}\ \emph {et~al.}(2023)\citenamefont {Aslam},
  \citenamefont {Zhou},\ and\ \citenamefont {Urbach}}]{Aslam2023}%
  \BibitemOpen
  \bibfield  {author} {\bibinfo {author} {\bibfnamefont {N.}~\bibnamefont
  {Aslam}}, \bibinfo {author} {\bibfnamefont {H.}~\bibnamefont {Zhou}}, \ and\
  \bibinfo {author} {\bibfnamefont {E.~e.~a.}\ \bibnamefont {Urbach}},\ }\href{\doibase 10.1038/s42254-023-00558-3} {\bibfield  {journal}
  {\bibinfo  {journal} {Nat Rev Phys}\ }\textbf {\bibinfo {volume} {5}},\
  \bibinfo {pages} {157–169} (\bibinfo {year} {2023})}\BibitemShut {NoStop}%
\bibitem [{\citenamefont {Correa}\ \emph {et~al.}(2015)\citenamefont {Correa},
  \citenamefont {Mehboudi}, \citenamefont {Adesso},\ and\ \citenamefont
  {Sanpera}}]{Correa2015}%
  \BibitemOpen
  \bibfield  {author} {\bibinfo {author} {\bibfnamefont {L.~A.}\ \bibnamefont
  {Correa}}, \bibinfo {author} {\bibfnamefont {M.}~\bibnamefont {Mehboudi}},
  \bibinfo {author} {\bibfnamefont {G.}~\bibnamefont {Adesso}}, \ and\ \bibinfo
  {author} {\bibfnamefont {A.}~\bibnamefont {Sanpera}},\ }\href {\doibase 10.1103/PHYSREVLETT.114.220405} {\bibfield  {journal} {\bibinfo  {journal}
  {Physical Review Letters}\ }\textbf {\bibinfo {volume} {114}},\ \bibinfo
  {pages} {220405} (\bibinfo {year} {2015})}\BibitemShut {NoStop}%
\bibitem [{\citenamefont {Mehboudi}\ \emph {et~al.}(2019)\citenamefont
  {Mehboudi}, \citenamefont {Sanpera},\ and\ \citenamefont
  {Correa}}]{Mehboudi2019}%
  \BibitemOpen
  \bibfield  {author} {\bibinfo {author} {\bibfnamefont {M.}~\bibnamefont
  {Mehboudi}}, \bibinfo {author} {\bibfnamefont {A.}~\bibnamefont {Sanpera}}, \
  and\ \bibinfo {author} {\bibfnamefont {L.~A.}\ \bibnamefont {Correa}},\
  }\href {\doibase 10.1088/1751-8121/AB2828} {\bibfield  {journal} {\bibinfo
  {journal} {Journal of Physics A: Mathematical and Theoretical}\ }\textbf
  {\bibinfo {volume} {52}},\ \bibinfo {pages} {303001} (\bibinfo {year}
  {2019})}\BibitemShut {NoStop}%
\bibitem [{\citenamefont {Abbott}\ \emph {et~al.}(2016)\citenamefont {Abbott}
  \emph {et~al.}}]{GW1}%
  \BibitemOpen
  \bibfield  {author} {\bibinfo {author} {\bibfnamefont {B.~P.}\ \bibnamefont
  {Abbott}} \emph {et~al.},\ }\href {\doibase 10.1103/PhysRevLett.116.061102}
  {\bibfield  {journal} {\bibinfo  {journal} {Physical Review Letters}\
  }\textbf {\bibinfo {volume} {116}},\ \bibinfo {pages} {061102} (\bibinfo
  {year} {2016})}\BibitemShut {NoStop}%
\bibitem [{\citenamefont {Tse}\ \emph {et~al.}(2019)\citenamefont {Tse} \emph
  {et~al.}}]{GW2}%
  \BibitemOpen
  \bibfield  {author} {\bibinfo {author} {\bibfnamefont {M.}~\bibnamefont
  {Tse}} \emph {et~al.},\ }\href {\doibase 10.1103/PhysRevLett.123.231107}
  {\bibfield  {journal} {\bibinfo  {journal} {Physical Review Letters}\
  }\textbf {\bibinfo {volume} {123}},\ \bibinfo {pages} {231107} (\bibinfo
  {year} {2019})}\BibitemShut {NoStop}%
\bibitem [{\citenamefont {Bailes}\ \emph {et~al.}(2021)\citenamefont {Bailes}
  \emph {et~al.}}]{GW3}%
  \BibitemOpen
  \bibfield  {author} {\bibinfo {author} {\bibfnamefont {M.}~\bibnamefont
  {Bailes}} \emph {et~al.},\ }\href {\doibase 10.1038/s42254-021-00303-8}
  {\bibfield  {journal} {\bibinfo  {journal} {Nature Reviews Physics 2021 3:5}\
  }\textbf {\bibinfo {volume} {3}},\ \bibinfo {pages} {344} (\bibinfo {year}
  {2021})}\BibitemShut {NoStop}%
\bibitem [{\citenamefont {Christensen}\ and\ \citenamefont
  {Meyer}(2022)}]{GW4}%
  \BibitemOpen
  \bibfield  {author} {\bibinfo {author} {\bibfnamefont {N.}~\bibnamefont
  {Christensen}}\ and\ \bibinfo {author} {\bibfnamefont {R.}~\bibnamefont
  {Meyer}},\ }\href {\doibase 10.1103/REVMODPHYS.94.025001} {\bibfield
  {journal} {\bibinfo  {journal} {Reviews of Modern Physics}\ }\textbf
  {\bibinfo {volume} {94}} (\bibinfo {year} {2022}),\
  10.1103/REVMODPHYS.94.025001}\BibitemShut {NoStop}%
\bibitem [{\citenamefont {Belavkin}(1976)}]{Belavkin76}%
  \BibitemOpen
  \bibfield  {author} {\bibinfo {author} {\bibfnamefont {V.}~\bibnamefont
  {Belavkin}},\ }\href {\doibase 10.1007/BF01032091} {\bibfield  {journal}
  {\bibinfo  {journal} {Theoretical and Mathematical Physics}\ }\textbf
  {\bibinfo {volume} {26}},\ \bibinfo {pages} {213} (\bibinfo {year}
  {1976})}\BibitemShut {NoStop}%
\bibitem [{\citenamefont {Braunstein}\ and\ \citenamefont
  {Caves}(1994)}]{QCR1}%
  \BibitemOpen
  \bibfield  {author} {\bibinfo {author} {\bibfnamefont {S.~L.}\ \bibnamefont
  {Braunstein}}\ and\ \bibinfo {author} {\bibfnamefont {C.~M.}\ \bibnamefont
  {Caves}},\ }\href {\doibase 10.1103/PhysRevLett.72.3439} {\bibfield
  {journal} {\bibinfo  {journal} {Phys. Rev. Lett.}\ }\textbf {\bibinfo
  {volume} {72}},\ \bibinfo {pages} {3439} (\bibinfo {year}
  {1994})}\BibitemShut {NoStop}%
\bibitem [{\citenamefont {Haroche}\ and\ \citenamefont
  {Raimond}(2006)}]{Haroche}%
  \BibitemOpen
  \bibfield  {author} {\bibinfo {author} {\bibfnamefont {S.}~\bibnamefont
  {Haroche}}\ and\ \bibinfo {author} {\bibfnamefont {J.-M.}\ \bibnamefont
  {Raimond}},\ }\href{\doibase 10.1093/acprof:oso/9780198509141.001.0001} {\emph {\bibinfo {title} {Exploring the Quantum:
  atoms, cavities and photons}}}\ (\bibinfo  {publisher} {Oxford University
  Press},\ \bibinfo {year} {2006})\BibitemShut {NoStop}%
\bibitem [{\citenamefont {Ciccarello}\ \emph {et~al.}(2022)\citenamefont
  {Ciccarello}, \citenamefont {Lorenzo}, \citenamefont {Giovannetti},\ and\
  \citenamefont {Palma}}]{Ciccarello22}%
  \BibitemOpen
  \bibfield  {author} {\bibinfo {author} {\bibfnamefont {F.}~\bibnamefont
  {Ciccarello}}, \bibinfo {author} {\bibfnamefont {S.}~\bibnamefont {Lorenzo}},
  \bibinfo {author} {\bibfnamefont {V.}~\bibnamefont {Giovannetti}}, \ and\
  \bibinfo {author} {\bibfnamefont {G.~M.}\ \bibnamefont {Palma}},\ }\href{\doibase 10.1016/j.physrep.2022.01.001.} {\bibfield
  {journal} {\bibinfo  {journal} {Physics Reports}\ }\textbf {\bibinfo {volume}
  {954}},\ \bibinfo {pages} {1} (\bibinfo {year} {2022})}\BibitemShut {NoStop}%
\bibitem [{\citenamefont {Perez-Garcia}\ \emph {et~al.}(2007)\citenamefont
  {Perez-Garcia}, \citenamefont {Verstraete}, \citenamefont {Wolf},\ and\
  \citenamefont {Cirac}}]{PerezGarciaWolfCirac}%
  \BibitemOpen
  \bibfield  {author} {\bibinfo {author} {\bibfnamefont {D.}~\bibnamefont
  {Perez-Garcia}}, \bibinfo {author} {\bibfnamefont {F.}~\bibnamefont
  {Verstraete}}, \bibinfo {author} {\bibfnamefont {M.~M.}\ \bibnamefont
  {Wolf}}, \ and\ \bibinfo {author} {\bibfnamefont {J.~I.}\ \bibnamefont
  {Cirac}},\ }\href{\doibase 10.26421/QIC7.5-6-1} {\bibfield  {journal} {\bibinfo  {journal} {Quantum
  Info. Comput.}\ }\textbf {\bibinfo {volume} {7}},\ \bibinfo {pages}
  {401–430} (\bibinfo {year} {2007})}\BibitemShut {NoStop}%
\bibitem [{\citenamefont {Sch\"on}\ \emph {et~al.}(2005)\citenamefont
  {Sch\"on}, \citenamefont {Solano}, \citenamefont {Verstraete}, \citenamefont
  {Cirac},\ and\ \citenamefont {Wolf}}]{SchonSolanoVerstraeteCiracWolf}%
  \BibitemOpen
  \bibfield  {author} {\bibinfo {author} {\bibfnamefont {C.}~\bibnamefont
  {Sch\"on}}, \bibinfo {author} {\bibfnamefont {E.}~\bibnamefont {Solano}},
  \bibinfo {author} {\bibfnamefont {F.}~\bibnamefont {Verstraete}}, \bibinfo
  {author} {\bibfnamefont {J.~I.}\ \bibnamefont {Cirac}}, \ and\ \bibinfo
  {author} {\bibfnamefont {M.~M.}\ \bibnamefont {Wolf}},\ }\href {\doibase 10.1103/PhysRevLett.95.110503} {\bibfield  {journal} {\bibinfo  {journal}
  {Phys. Rev. Lett.}\ }\textbf {\bibinfo {volume} {95}},\ \bibinfo {pages}
  {110503} (\bibinfo {year} {2005})}\BibitemShut {NoStop}%
\bibitem [{\citenamefont {Fannes}\ \emph {et~al.}(1992)\citenamefont {Fannes},
  \citenamefont {Nachtergaele},\ and\ \citenamefont
  {Werner}}]{FannesNachtergaeleWerner}%
  \BibitemOpen
  \bibfield  {author} {\bibinfo {author} {\bibfnamefont {M.}~\bibnamefont
  {Fannes}}, \bibinfo {author} {\bibfnamefont {B.}~\bibnamefont
  {Nachtergaele}}, \ and\ \bibinfo {author} {\bibfnamefont {R.~F.}\
  \bibnamefont {Werner}},\ }\href{\doibase 10.1007/BF02099178} {\bibfield  {journal} {\bibinfo
  {journal} {Communications in Mathematical Physics}\ }\textbf {\bibinfo
  {volume} {144}},\ \bibinfo {pages} {443} (\bibinfo {year}
  {1992})}\BibitemShut {NoStop}%
\bibitem [{\citenamefont {Fannes}\ \emph {et~al.}(1994)\citenamefont {Fannes},
  \citenamefont {Nachtergaele},\ and\ \citenamefont
  {Werner}}]{FannesNachtergaeleWerner2}%
  \BibitemOpen
  \bibfield  {author} {\bibinfo {author} {\bibfnamefont {M.}~\bibnamefont
  {Fannes}}, \bibinfo {author} {\bibfnamefont {B.}~\bibnamefont
  {Nachtergaele}}, \ and\ \bibinfo {author} {\bibfnamefont {R.~F.}\
  \bibnamefont {Werner}},\ }\href{\doibase 10.1006/jfan.1994.1041} {\bibfield  {journal} {\bibinfo
  {journal} {Journal of Functional Analysis}\ }\textbf {\bibinfo {volume}
  {120}},\ \bibinfo {pages} {511} (\bibinfo {year} {1994})}\BibitemShut
  {NoStop}%
\bibitem [{\citenamefont {Attal}\ and\ \citenamefont
  {Pautrat}(2006{\natexlab{a}})}]{AttalPautrat}%
  \BibitemOpen
  \bibfield  {author} {\bibinfo {author} {\bibfnamefont {S.}~\bibnamefont
  {Attal}}\ and\ \bibinfo {author} {\bibfnamefont {Y.}~\bibnamefont
  {Pautrat}},\ }\href{\doibase 10.1007/s00023-005-0242-8} {\bibfield  {journal} {\bibinfo  {journal} {Ann.
  Henri Poincar\'{e}}\ }\textbf {\bibinfo {volume} {7}},\ \bibinfo {pages} {59}
  (\bibinfo {year} {2006}{\natexlab{a}})}\BibitemShut {NoStop}%
\bibitem [{\citenamefont {Verstraete}\ and\ \citenamefont
  {Cirac}(2010)}]{VerstraeteCirac}%
  \BibitemOpen
  \bibfield  {author} {\bibinfo {author} {\bibfnamefont {F.}~\bibnamefont
  {Verstraete}}\ and\ \bibinfo {author} {\bibfnamefont {J.~I.}\ \bibnamefont
  {Cirac}},\ }\href {\doibase 10.1103/PhysRevLett.104.190405} {\bibfield
  {journal} {\bibinfo  {journal} {Phys. Rev. Lett.}\ }\textbf {\bibinfo
  {volume} {104}},\ \bibinfo {pages} {190405} (\bibinfo {year}
  {2010})}\BibitemShut {NoStop}%
\bibitem [{\citenamefont {Gough}(2004)}]{Gough04}%
  \BibitemOpen
  \bibfield  {author} {\bibinfo {author} {\bibfnamefont {J.}~\bibnamefont
  {Gough}},\ }\href{\doibase 10.1023/B:MATH.0000035039.56638.e1} {\bibfield  {journal} {\bibinfo  {journal} {Lett.
  Math. Phys.}\ }\textbf {\bibinfo {volume} {67}},\ \bibinfo {pages} {207}
  (\bibinfo {year} {2004})}\BibitemShut {NoStop}%
\bibitem [{\citenamefont {Huelga}\ \emph {et~al.}(1997)\citenamefont {Huelga},
  \citenamefont {Macchiavello}, \citenamefont {Pellizzari}, \citenamefont
  {Ekert}, \citenamefont {Plenio},\ and\ \citenamefont {Cirac}}]{Huelga97}%
  \BibitemOpen
  \bibfield  {author} {\bibinfo {author} {\bibfnamefont {S.~F.}\ \bibnamefont
  {Huelga}}, \bibinfo {author} {\bibfnamefont {C.}~\bibnamefont
  {Macchiavello}}, \bibinfo {author} {\bibfnamefont {T.}~\bibnamefont
  {Pellizzari}}, \bibinfo {author} {\bibfnamefont {A.~K.}\ \bibnamefont
  {Ekert}}, \bibinfo {author} {\bibfnamefont {M.~B.}\ \bibnamefont {Plenio}}, \
  and\ \bibinfo {author} {\bibfnamefont {J.~I.}\ \bibnamefont {Cirac}},\ }\href
  {\doibase 10.1103/PhysRevLett.79.3865} {\bibfield  {journal} {\bibinfo
  {journal} {Phys. Rev. Lett.}\ }\textbf {\bibinfo {volume} {79}},\ \bibinfo
  {pages} {3865} (\bibinfo {year} {1997})}\BibitemShut {NoStop}%
\bibitem [{\citenamefont {Benatti}\ \emph {et~al.}(2014)\citenamefont
  {Benatti}, \citenamefont {Alipour},\ and\ \citenamefont
  {Rezakhani}}]{Benatti14}%
  \BibitemOpen
  \bibfield  {author} {\bibinfo {author} {\bibfnamefont {F.}~\bibnamefont
  {Benatti}}, \bibinfo {author} {\bibfnamefont {S.}~\bibnamefont {Alipour}}, \
  and\ \bibinfo {author} {\bibfnamefont {A.~T.}\ \bibnamefont {Rezakhani}},\
  }\href {\doibase 10.1088/1367-2630/16/1/015023} {\bibfield
  {journal} {\bibinfo  {journal} {New J. Phys.}\ }\textbf {\bibinfo {volume}
  {16}},\ \bibinfo {pages} {015023} (\bibinfo {year} {2014})}\BibitemShut
  {NoStop}%
\bibitem [{\citenamefont {Sekatski}\ \emph {et~al.}(2017)\citenamefont
  {Sekatski}, \citenamefont {Skotiniotis}, \citenamefont
  {Ko{\l{}}ody{\'{n}}ski},\ and\ \citenamefont {D{\"{u}}r}}]{Sekatski2017}%
  \BibitemOpen
  \bibfield  {author} {\bibinfo {author} {\bibfnamefont {P.}~\bibnamefont
  {Sekatski}}, \bibinfo {author} {\bibfnamefont {M.}~\bibnamefont
  {Skotiniotis}}, \bibinfo {author} {\bibfnamefont {J.}~\bibnamefont
  {Ko{\l{}}ody{\'{n}}ski}}, \ and\ \bibinfo {author} {\bibfnamefont
  {W.}~\bibnamefont {D{\"{u}}r}},\ }\href {\doibase 10.22331/q-2017-09-06-27}
  {\bibfield  {journal} {\bibinfo  {journal} {{Quantum}}\ }\textbf {\bibinfo
  {volume} {1}},\ \bibinfo {pages} {27} (\bibinfo {year} {2017})}\BibitemShut
  {NoStop}%
\bibitem [{\citenamefont {Demkowicz-Dobrza\ifmmode~\acute{n}\else
  \'{n}\fi{}ski}\ \emph {et~al.}(2017)\citenamefont
  {Demkowicz-Dobrza\ifmmode~\acute{n}\else \'{n}\fi{}ski}, \citenamefont
  {Czajkowski},\ and\ \citenamefont {Sekatski}}]{DDCS}%
  \BibitemOpen
  \bibfield  {author} {\bibinfo {author} {\bibfnamefont {R.}~\bibnamefont
  {Demkowicz-Dobrza\ifmmode~\acute{n}\else \'{n}\fi{}ski}}, \bibinfo {author}
  {\bibfnamefont {J.}~\bibnamefont {Czajkowski}}, \ and\ \bibinfo {author}
  {\bibfnamefont {P.}~\bibnamefont {Sekatski}},\ }\href {\doibase 10.1103/PhysRevX.7.041009} {\bibfield  {journal} {\bibinfo  {journal} {Phys.
  Rev. X}\ }\textbf {\bibinfo {volume} {7}},\ \bibinfo {pages} {041009}
  (\bibinfo {year} {2017})}\BibitemShut {NoStop}%
\bibitem [{\citenamefont {Zhou}\ \emph {et~al.}(2018)\citenamefont {Zhou},
  \citenamefont {Zhang}, \citenamefont {Preskill},\ and\ \citenamefont
  {Jiang}}]{ZhouZhangPreskillJuang}%
  \BibitemOpen
  \bibfield  {author} {\bibinfo {author} {\bibfnamefont {S.}~\bibnamefont
  {Zhou}}, \bibinfo {author} {\bibfnamefont {M.}~\bibnamefont {Zhang}},
  \bibinfo {author} {\bibfnamefont {J.}~\bibnamefont {Preskill}}, \ and\
  \bibinfo {author} {\bibfnamefont {L.}~\bibnamefont {Jiang}},\ }\href{\doibase 10.1038/s41467-017-02510-3} {\bibfield  {journal}
  {\bibinfo  {journal} {Nat Commun}\ }\textbf {\bibinfo {volume} {9}},\
  \bibinfo {pages} {78} (\bibinfo {year} {2018})}\BibitemShut {NoStop}%
\bibitem [{\citenamefont {G{\'{o}}recki}\ \emph {et~al.}(2020)\citenamefont
  {G{\'{o}}recki}, \citenamefont {Zhou}, \citenamefont {Jiang},\ and\
  \citenamefont {Demkowicz-Dobrza{\'{n}}ski}}]{GoreckiZhouSisiJiangDemko}%
  \BibitemOpen
  \bibfield  {author} {\bibinfo {author} {\bibfnamefont {W.}~\bibnamefont
  {G{\'{o}}recki}}, \bibinfo {author} {\bibfnamefont {S.}~\bibnamefont {Zhou}},
  \bibinfo {author} {\bibfnamefont {L.}~\bibnamefont {Jiang}}, \ and\ \bibinfo
  {author} {\bibfnamefont {R.}~\bibnamefont {Demkowicz-Dobrza{\'{n}}ski}},\
  }\href {\doibase 10.22331/q-2020-07-02-288} {\bibfield  {journal} {\bibinfo
  {journal} {{Quantum}}\ }\textbf {\bibinfo {volume} {4}},\ \bibinfo {pages}
  {288} (\bibinfo {year} {2020})}\BibitemShut {NoStop}%
\bibitem [{\citenamefont {Gardiner}\ and\ \citenamefont
  {Zoller}(2004)}]{GardinerZoller}%
  \BibitemOpen
  \bibfield  {author} {\bibinfo {author} {\bibfnamefont {C.~W.}\ \bibnamefont
  {Gardiner}}\ and\ \bibinfo {author} {\bibfnamefont {P.}~\bibnamefont
  {Zoller}},\ }\href@noop {} {\emph {\bibinfo {title} {Quantum Noise: A
  Handbook of Markovian and Non-Markovian Quantum Stochastic Methods with
  Applications to Quantum Optics}}},\ \bibinfo {edition} {3rd}\ ed.\ (\bibinfo
  {publisher} {Springer-Verlag},\ \bibinfo {address} {Berlin and New York},\
  \bibinfo {year} {2004})\BibitemShut {NoStop}%
\bibitem [{\citenamefont {Wiseman}\ and\ \citenamefont
  {Milburn}(2010)}]{WisemanMilburn}%
  \BibitemOpen
  \bibfield  {author} {\bibinfo {author} {\bibfnamefont {H.~M.}\ \bibnamefont
  {Wiseman}}\ and\ \bibinfo {author} {\bibfnamefont {G.~J.}\ \bibnamefont
  {Milburn}},\ }\href{\doibase 10.1017/CBO9780511813948} {\emph {\bibinfo {title} {Quantum Measurement and
  Control}}}\ (\bibinfo  {publisher} {Cambridge University Press},\ \bibinfo
  {year} {2010})\BibitemShut {NoStop}%
\bibitem [{\citenamefont {Gough}\ and\ \citenamefont {James}(2009)}]{GJ09}%
  \BibitemOpen
  \bibfield  {author} {\bibinfo {author} {\bibfnamefont {J.}~\bibnamefont
  {Gough}}\ and\ \bibinfo {author} {\bibfnamefont {M.~R.}\ \bibnamefont
  {James}},\ }\href{\doibase 10.1109/TAC.2009.2031205} {\bibfield  {journal} {\bibinfo  {journal} {IEEE
  Trans. Automat. Control}\ }\textbf {\bibinfo {volume} {54}},\ \bibinfo
  {pages} {2530} (\bibinfo {year} {2009})}\BibitemShut {NoStop}%
\bibitem [{\citenamefont {Combes}\ \emph {et~al.}(2017)\citenamefont {Combes},
  \citenamefont {Kerckhoff},\ and\ \citenamefont {Sarovar}}]{CKS17}%
  \BibitemOpen
  \bibfield  {author} {\bibinfo {author} {\bibfnamefont {J.}~\bibnamefont
  {Combes}}, \bibinfo {author} {\bibfnamefont {J.}~\bibnamefont {Kerckhoff}}, \
  and\ \bibinfo {author} {\bibfnamefont {M.}~\bibnamefont {Sarovar}},\ }\href
  {\doibase 10.1080/23746149.2017.1343097} {\bibfield  {journal} {\bibinfo
  {journal} {Adv. Phys. X}\ }\textbf {\bibinfo {volume} {2}},\ \bibinfo {pages}
  {784} (\bibinfo {year} {2017})}\BibitemShut {NoStop}%
\bibitem [{\citenamefont {Belavkin}(1994)}]{Belavkin94}%
  \BibitemOpen
  \bibfield  {author} {\bibinfo {author} {\bibfnamefont {V.~P.}\ \bibnamefont
  {Belavkin}},\ }\href {\doibase 10.1007/BF02054669} {\bibfield
   {journal} {\bibinfo  {journal} {Found. Phys.}\ }\textbf {\bibinfo {volume}
  {24}},\ \bibinfo {pages} {685} (\bibinfo {year} {1994})}\BibitemShut
  {NoStop}%
\bibitem [{\citenamefont {Dalibard}\ \emph {et~al.}(1992)\citenamefont
  {Dalibard}, \citenamefont {Castin},\ and\ \citenamefont {M\o{}lmer}}]{DCM}%
  \BibitemOpen
  \bibfield  {author} {\bibinfo {author} {\bibfnamefont {J.}~\bibnamefont
  {Dalibard}}, \bibinfo {author} {\bibfnamefont {Y.}~\bibnamefont {Castin}}, \
  and\ \bibinfo {author} {\bibfnamefont {K.}~\bibnamefont {M\o{}lmer}},\ }\href
  {\doibase 10.1103/PhysRevLett.68.580} {\bibfield  {journal} {\bibinfo
  {journal} {Phys. Rev. Lett.}\ }\textbf {\bibinfo {volume} {68}},\ \bibinfo
  {pages} {580} (\bibinfo {year} {1992})}\BibitemShut {NoStop}%
\bibitem [{\citenamefont {Wiseman}\ and\ \citenamefont {Milburn}(1993)}]{WM93}%
  \BibitemOpen
  \bibfield  {author} {\bibinfo {author} {\bibfnamefont {H.~M.}\ \bibnamefont
  {Wiseman}}\ and\ \bibinfo {author} {\bibfnamefont {G.~J.}\ \bibnamefont
  {Milburn}},\ }\href {\doibase 10.1103/PhysRevA.47.642} {\bibfield  {journal}
  {\bibinfo  {journal} {Phys. Rev. A}\ }\textbf {\bibinfo {volume} {47}},\
  \bibinfo {pages} {642} (\bibinfo {year} {1993})}\BibitemShut {NoStop}%
\bibitem [{\citenamefont {Bouten}\ \emph {et~al.}(2007)\citenamefont {Bouten},
  \citenamefont {{van Handel}},\ and\ \citenamefont {James}}]{BvHJ07}%
  \BibitemOpen
  \bibfield  {author} {\bibinfo {author} {\bibfnamefont {L.}~\bibnamefont
  {Bouten}}, \bibinfo {author} {\bibfnamefont {R.}~\bibnamefont {{van
  Handel}}}, \ and\ \bibinfo {author} {\bibfnamefont {M.~R.}\ \bibnamefont
  {James}},\ }\href{\doibase 10.1137/060651239} {\bibfield  {journal} {\bibinfo  {journal} {SIAM J.
  Control Optim.}\ }\textbf {\bibinfo {volume} {46}},\ \bibinfo {pages} {2199}
  (\bibinfo {year} {2007})}\BibitemShut {NoStop}%
\bibitem [{\citenamefont {Carmichael}(1993)}]{Carmichael}%
  \BibitemOpen
  \bibfield  {author} {\bibinfo {author} {\bibfnamefont {H.~J.}\ \bibnamefont
  {Carmichael}},\ }\href {\doibase 10.1007/978-3-540-47620-7}
  {\emph {\bibinfo {title} {An Open Systems Approach to Quantum Optics}}}\
  (\bibinfo  {publisher} {Springer-Verlag, Berlin},\ \bibinfo {year}
  {1993})\BibitemShut {NoStop}%
\bibitem [{\citenamefont {Mabuchi}(1996)}]{Mab96}%
  \BibitemOpen
  \bibfield  {author} {\bibinfo {author} {\bibfnamefont {H.}~\bibnamefont
  {Mabuchi}},\ }\href{\doibase 10.1088/1355-5111/8/6/002} {\bibfield  {journal} {\bibinfo  {journal}
  {Quantum Semiclass. Opt.}\ }\textbf {\bibinfo {volume} {8}},\ \bibinfo
  {pages} {1103} (\bibinfo {year} {1996})}\BibitemShut {NoStop}%
\bibitem [{\citenamefont {Gambetta}\ and\ \citenamefont
  {Wiseman}(2001)}]{GW01}%
  \BibitemOpen
  \bibfield  {author} {\bibinfo {author} {\bibfnamefont {J.}~\bibnamefont
  {Gambetta}}\ and\ \bibinfo {author} {\bibfnamefont {H.~M.}\ \bibnamefont
  {Wiseman}},\ }\href {https://doi.org/10.1103/PhysRevA.64.042105} {\bibfield  {journal} {\bibinfo  {journal} {Phys.
  Rev. A}\ }\textbf {\bibinfo {volume} {64}},\ \bibinfo {pages} {042105}
  (\bibinfo {year} {2001})}\BibitemShut {NoStop}%
\bibitem [{\citenamefont {Berry}\ and\ \citenamefont
  {Wiseman}(2002)}]{Berry02}%
  \BibitemOpen
  \bibfield  {author} {\bibinfo {author} {\bibfnamefont {D.~W.}\ \bibnamefont
  {Berry}}\ and\ \bibinfo {author} {\bibfnamefont {H.~M.}\ \bibnamefont
  {Wiseman}},\ }\href {https://doi.org/10.1103/PhysRevA.65.043803}
  {\bibfield  {journal} {\bibinfo  {journal} {Phys. Rev. A}\ }\textbf {\bibinfo
  {volume} {65}},\ \bibinfo {pages} {043803} (\bibinfo {year}
  {2002})}\BibitemShut {NoStop}%
\bibitem [{\citenamefont {Pope}\ \emph {et~al.}(2004)\citenamefont {Pope},
  \citenamefont {Wiseman},\ and\ \citenamefont {Langford}}]{Wiseman04}%
  \BibitemOpen
  \bibfield  {author} {\bibinfo {author} {\bibfnamefont {D.~T.}\ \bibnamefont
  {Pope}}, \bibinfo {author} {\bibfnamefont {H.~M.}\ \bibnamefont {Wiseman}}, \
  and\ \bibinfo {author} {\bibfnamefont {N.~K.}\ \bibnamefont {Langford}},\
  }\href {https://doi.org/10.1103/PhysRevA.70.043812} {\bibfield
  {journal} {\bibinfo  {journal} {Phys. Rev. A}\ }\textbf {\bibinfo {volume}
  {70}},\ \bibinfo {pages} {043812} (\bibinfo {year} {2004})}\BibitemShut
  {NoStop}%
\bibitem [{\citenamefont {Ralph}\ \emph {et~al.}(2011)\citenamefont {Ralph},
  \citenamefont {Jacobs},\ and\ \citenamefont {Hill}}]{Ralph11}%
  \BibitemOpen
  \bibfield  {author} {\bibinfo {author} {\bibfnamefont {J.~F.}\ \bibnamefont
  {Ralph}}, \bibinfo {author} {\bibfnamefont {K.}~\bibnamefont {Jacobs}}, \
  and\ \bibinfo {author} {\bibfnamefont {C.~D.}\ \bibnamefont {Hill}},\ }\href{https://doi.org/10.1103/PhysRevA.84.052119} {\bibfield  {journal}
  {\bibinfo  {journal} {Phys. Rev. A}\ }\textbf {\bibinfo {volume} {84}},\
  \bibinfo {pages} {052119} (\bibinfo {year} {2011})}\BibitemShut {NoStop}%
\bibitem [{\citenamefont {Chase}\ and\ \citenamefont {Geremia}(2009)}]{CG09}%
  \BibitemOpen
  \bibfield  {author} {\bibinfo {author} {\bibfnamefont {B.}~\bibnamefont
  {Chase}}\ and\ \bibinfo {author} {\bibfnamefont {J.~M.}\ \bibnamefont
  {Geremia}},\ }\href {https://doi.org/10.1103/PhysRevA.79.022314} {\bibfield
  {journal} {\bibinfo  {journal} {Phys. Rev. A}\ }\textbf {\bibinfo {volume}
  {79}},\ \bibinfo {pages} {022314} (\bibinfo {year} {2009})}\BibitemShut
  {NoStop}%
\bibitem [{\citenamefont {Six}\ \emph {et~al.}(2015)\citenamefont {Six},
  \citenamefont {Campagne-Ibarcq}, \citenamefont {Bretheau}, \citenamefont
  {Huard},\ and\ \citenamefont {Rouchon}}]{Six15}%
  \BibitemOpen
  \bibfield  {author} {\bibinfo {author} {\bibfnamefont {P.}~\bibnamefont
  {Six}}, \bibinfo {author} {\bibfnamefont {P.}~\bibnamefont
  {Campagne-Ibarcq}}, \bibinfo {author} {\bibfnamefont {L.}~\bibnamefont
  {Bretheau}}, \bibinfo {author} {\bibfnamefont {B.}~\bibnamefont {Huard}}, \
  and\ \bibinfo {author} {\bibfnamefont {P.}~\bibnamefont {Rouchon}},\ }in\
  \href {https://doi.org/10.1109/CDC.2015.7403443} {\emph {\bibinfo {booktitle} {2015
  54th IEEE Conference on Decision and Control (CDC)}}}\ (\bibinfo {year}
  {2015})\ pp.\ \bibinfo {pages} {7742--7748}\BibitemShut {NoStop}%
\bibitem [{\citenamefont {Macieszczak}\ \emph {et~al.}(2016)\citenamefont
  {Macieszczak}, \citenamefont {Guta}, \citenamefont {Lesanovsky},\ and\
  \citenamefont {Garrahan}}]{GutaMacieszczakGarrahanLesanovsky}%
  \BibitemOpen
  \bibfield  {author} {\bibinfo {author} {\bibfnamefont {K.}~\bibnamefont
  {Macieszczak}}, \bibinfo {author} {\bibfnamefont {M.}~\bibnamefont {Guta}},
  \bibinfo {author} {\bibfnamefont {I.}~\bibnamefont {Lesanovsky}}, \ and\
  \bibinfo {author} {\bibfnamefont {J.~P.}\ \bibnamefont {Garrahan}},\ }\href{https://doi.org/10.1103/PhysRevA.93.022103} {\bibfield  {journal}
  {\bibinfo  {journal} {Physical Review A}\ }\textbf {\bibinfo {volume} {93}},\
  \bibinfo {pages} {022103} (\bibinfo {year} {2016})}\BibitemShut {NoStop}%
\bibitem [{\citenamefont {Albarelli}\ \emph {et~al.}(2018)\citenamefont
  {Albarelli}, \citenamefont {Rossi}, \citenamefont {Tamascelli},\ and\
  \citenamefont {Genoni}}]{Genoni18}%
  \BibitemOpen
  \bibfield  {author} {\bibinfo {author} {\bibfnamefont {F.}~\bibnamefont
  {Albarelli}}, \bibinfo {author} {\bibfnamefont {M.~A.~C.}\ \bibnamefont
  {Rossi}}, \bibinfo {author} {\bibfnamefont {D.}~\bibnamefont {Tamascelli}}, \
  and\ \bibinfo {author} {\bibfnamefont {M.~G.}\ \bibnamefont {Genoni}},\
  }\href {https://doi.org/10.22331/q-2018-12-03-110} {\bibfield
  {journal} {\bibinfo  {journal} {Quantum}\ }\textbf {\bibinfo {volume} {2}},\
  \bibinfo {pages} {110} (\bibinfo {year} {2018})}\BibitemShut {NoStop}%
\bibitem [{\citenamefont {Ilias}\ \emph {et~al.}(2022)\citenamefont {Ilias},
  \citenamefont {Yang}, \citenamefont {Huelga},\ and\ \citenamefont
  {Plenio}}]{Ilias22}%
  \BibitemOpen
  \bibfield  {author} {\bibinfo {author} {\bibfnamefont {T.}~\bibnamefont
  {Ilias}}, \bibinfo {author} {\bibfnamefont {D.}~\bibnamefont {Yang}},
  \bibinfo {author} {\bibfnamefont {S.~F.}\ \bibnamefont {Huelga}}, \ and\
  \bibinfo {author} {\bibfnamefont {M.~B.}\ \bibnamefont {Plenio}},\ }\href
  {\doibase 10.1103/PRXQuantum.3.010354} {\bibfield  {journal} {\bibinfo
  {journal} {PRX Quantum}\ }\textbf {\bibinfo {volume} {3}},\ \bibinfo {pages}
  {010354} (\bibinfo {year} {2022})}\BibitemShut {NoStop}%
\bibitem [{\citenamefont {Plenio}\ and\ \citenamefont
  {Huelga}(2016)}]{Plenio16}%
  \BibitemOpen
  \bibfield  {author} {\bibinfo {author} {\bibfnamefont {M.~B.}\ \bibnamefont
  {Plenio}}\ and\ \bibinfo {author} {\bibfnamefont {S.~F.}\ \bibnamefont
  {Huelga}},\ }\href {https://doi.org/10.1103/PhysRevA.93.032123} {\bibfield
  {journal} {\bibinfo  {journal} {Phys. Rev. A}\ }\textbf {\bibinfo {volume}
  {93}},\ \bibinfo {pages} {032123} (\bibinfo {year} {2016})}\BibitemShut
  {NoStop}%
\bibitem [{\citenamefont {Gammelmark}\ and\ \citenamefont
  {M\o{}lmer}(2013)}]{GM13}%
  \BibitemOpen
  \bibfield  {author} {\bibinfo {author} {\bibfnamefont {S.}~\bibnamefont
  {Gammelmark}}\ and\ \bibinfo {author} {\bibfnamefont {K.}~\bibnamefont
  {M\o{}lmer}},\ }\href {https://doi.org/10.1103/PhysRevA.87.032115}
  {\bibfield  {journal} {\bibinfo  {journal} {Phys. Rev. A}\ }\textbf {\bibinfo
  {volume} {87}},\ \bibinfo {pages} {032115} (\bibinfo {year}
  {2013})}\BibitemShut {NoStop}%
\bibitem [{\citenamefont {Negretti}\ and\ \citenamefont
  {M\o{}lmer}(2013)}]{Negretti13}%
  \BibitemOpen
  \bibfield  {author} {\bibinfo {author} {\bibfnamefont {A.}~\bibnamefont
  {Negretti}}\ and\ \bibinfo {author} {\bibfnamefont {K.}~\bibnamefont
  {M\o{}lmer}},\ }\href {https://doi.org/10.1088/1367-2630/15/12/125002} {\bibfield  {journal}
  {\bibinfo  {journal} {New J. Phys.}\ }\textbf {\bibinfo {volume} {15}},\
  \bibinfo {pages} {125002} (\bibinfo {year} {2013})}\BibitemShut {NoStop}%
\bibitem [{\citenamefont {Kiilerich}\ and\ \citenamefont
  {M\o{}lmer}(2016)}]{KM16}%
  \BibitemOpen
  \bibfield  {author} {\bibinfo {author} {\bibfnamefont {A.~H.}\ \bibnamefont
  {Kiilerich}}\ and\ \bibinfo {author} {\bibfnamefont {K.}~\bibnamefont
  {M\o{}lmer}},\ }\href {https://doi.org/10.1103/PhysRevA.94.032103}
  {\bibfield  {journal} {\bibinfo  {journal} {Phys. Rev. A}\ }\textbf {\bibinfo
  {volume} {94}},\ \bibinfo {pages} {032103} (\bibinfo {year}
  {2016})}\BibitemShut {NoStop}%
\bibitem [{\citenamefont {Ralph}\ \emph {et~al.}(2017)\citenamefont {Ralph},
  \citenamefont {Maskell},\ and\ \citenamefont {Jacobs}}]{Ralph17}%
  \BibitemOpen
  \bibfield  {author} {\bibinfo {author} {\bibfnamefont {J.~F.}\ \bibnamefont
  {Ralph}}, \bibinfo {author} {\bibfnamefont {S.}~\bibnamefont {Maskell}}, \
  and\ \bibinfo {author} {\bibfnamefont {K.}~\bibnamefont {Jacobs}},\ }\href{https://doi.org/10.1103/PhysRevA.96.052306} {\bibfield  {journal}
  {\bibinfo  {journal} {Phys. Rev. A}\ }\textbf {\bibinfo {volume} {96}},\
  \bibinfo {pages} {052306} (\bibinfo {year} {2017})}\BibitemShut {NoStop}%
\bibitem [{\citenamefont {Zhang}\ \emph {et~al.}(2019)\citenamefont {Zhang},
  \citenamefont {Zhou}, \citenamefont {Feng},\ and\ \citenamefont
  {Li}}]{Zhang19}%
  \BibitemOpen
  \bibfield  {author} {\bibinfo {author} {\bibfnamefont {C.}~\bibnamefont
  {Zhang}}, \bibinfo {author} {\bibfnamefont {K.}~\bibnamefont {Zhou}},
  \bibinfo {author} {\bibfnamefont {W.}~\bibnamefont {Feng}}, \ and\ \bibinfo
  {author} {\bibfnamefont {X.-Q.}\ \bibnamefont {Li}},\ }\href {https://doi.org/10.1103/PhysRevA.99.022114} {\bibfield  {journal} {\bibinfo
  {journal} {Phys. Rev. A}\ }\textbf {\bibinfo {volume} {99}},\ \bibinfo
  {pages} {022114} (\bibinfo {year} {2019})}\BibitemShut {NoStop}%
\bibitem [{\citenamefont {Tsang}(2009{\natexlab{a}})}]{Tsang09}%
  \BibitemOpen
  \bibfield  {author} {\bibinfo {author} {\bibfnamefont {M.}~\bibnamefont
  {Tsang}},\ }\href {https://doi.org/10.1103/PhysRevLett.102.250403} {\bibfield
  {journal} {\bibinfo  {journal} {Phys. Rev. Lett.}\ }\textbf {\bibinfo
  {volume} {102}},\ \bibinfo {pages} {250403} (\bibinfo {year}
  {2009}{\natexlab{a}})}\BibitemShut {NoStop}%
\bibitem [{\citenamefont {Tsang}(2009{\natexlab{b}})}]{T09}%
  \BibitemOpen
  \bibfield  {author} {\bibinfo {author} {\bibfnamefont {M.}~\bibnamefont
  {Tsang}},\ }\href {https://doi.org/10.1103/PhysRevA.80.033840}
  {\bibfield  {journal} {\bibinfo  {journal} {Phys. Rev. A}\ }\textbf {\bibinfo
  {volume} {80}},\ \bibinfo {pages} {033840} (\bibinfo {year}
  {2009}{\natexlab{b}})}\BibitemShut {NoStop}%
\bibitem [{\citenamefont {Tsang}(2010)}]{Tsang10}%
  \BibitemOpen
  \bibfield  {author} {\bibinfo {author} {\bibfnamefont {M.}~\bibnamefont
  {Tsang}},\ }\href {https://doi.org/10.1103/PhysRevA.81.013824}
  {\bibfield  {journal} {\bibinfo  {journal} {Phys. Rev. A}\ }\textbf {\bibinfo
  {volume} {81}},\ \bibinfo {pages} {013824} (\bibinfo {year}
  {2010})}\BibitemShut {NoStop}%
\bibitem [{\citenamefont {Guevara}\ and\ \citenamefont
  {Wiseman}(2015)}]{Guevara15}%
  \BibitemOpen
  \bibfield  {author} {\bibinfo {author} {\bibfnamefont {I.}~\bibnamefont
  {Guevara}}\ and\ \bibinfo {author} {\bibfnamefont {H.}~\bibnamefont
  {Wiseman}},\ }\href {https://doi.org/10.1103/PhysRevLett.115.180407}
  {\bibfield  {journal} {\bibinfo  {journal} {Phys. Rev. Lett.}\ }\textbf
  {\bibinfo {volume} {115}},\ \bibinfo {pages} {180407} (\bibinfo {year}
  {2015})}\BibitemShut {NoStop}%
\bibitem [{\citenamefont {Gu\c{t}\v{a}}\ and\ \citenamefont
  {Yamamoto}(2016)}]{GY16}%
  \BibitemOpen
  \bibfield  {author} {\bibinfo {author} {\bibfnamefont {M.}~\bibnamefont
  {Gu\c{t}\v{a}}}\ and\ \bibinfo {author} {\bibfnamefont {N.}~\bibnamefont
  {Yamamoto}},\ }\href {https://doi.org/10.1109/TAC.2015.2448491}
  {\bibfield  {journal} {\bibinfo  {journal} {IEEE Trans. Automat. Contr.}\
  }\textbf {\bibinfo {volume} {61}},\ \bibinfo {pages} {921} (\bibinfo {year}
  {2016})}\BibitemShut {NoStop}%
\bibitem [{\citenamefont {Genoni}(2017)}]{Genoni17}%
  \BibitemOpen
  \bibfield  {author} {\bibinfo {author} {\bibfnamefont {M.~G.}\ \bibnamefont
  {Genoni}},\ }\href {https://doi.org/10.1103/PhysRevA.95.012116}
  {\bibfield  {journal} {\bibinfo  {journal} {Phys. Rev. A}\ }\textbf {\bibinfo
  {volume} {95}},\ \bibinfo {pages} {012116} (\bibinfo {year}
  {2017})}\BibitemShut {NoStop}%
\bibitem [{\citenamefont {Levitt}\ and\ \citenamefont
  {Gu\c{t}\u{a}}(2017)}]{Levitt17}%
  \BibitemOpen
  \bibfield  {author} {\bibinfo {author} {\bibfnamefont {M.}~\bibnamefont
  {Levitt}}\ and\ \bibinfo {author} {\bibfnamefont {M.}~\bibnamefont
  {Gu\c{t}\u{a}}},\ }\href{https://doi.org/10.1103/PhysRevA.95.033825} {\bibfield  {journal} {\bibinfo
  {journal} {Phys. Rev. A}\ }\textbf {\bibinfo {volume} {95}},\ \bibinfo
  {pages} {033825} (\bibinfo {year} {2017})}\BibitemShut {NoStop}%
\bibitem [{\citenamefont {Levitt}\ \emph {et~al.}(2018)\citenamefont {Levitt},
  \citenamefont {Gu\c{t}\v{a}},\ and\ \citenamefont {Nurdin}}]{LGN18}%
  \BibitemOpen
  \bibfield  {author} {\bibinfo {author} {\bibfnamefont {M.}~\bibnamefont
  {Levitt}}, \bibinfo {author} {\bibfnamefont {M.}~\bibnamefont
  {Gu\c{t}\v{a}}}, \ and\ \bibinfo {author} {\bibfnamefont {H.~I.}\
  \bibnamefont {Nurdin}},\ }\href {https://doi.org/10.1016/j.automatica.2017.12.037} {\bibfield  {journal}
  {\bibinfo  {journal} {Automatica}\ }\textbf {\bibinfo {volume} {90}},\
  \bibinfo {pages} {255} (\bibinfo {year} {2018})}\BibitemShut {NoStop}%
\bibitem [{\citenamefont {Garrahan}\ and\ \citenamefont
  {Lesanovsky}(2010)}]{Garrahan10}%
  \BibitemOpen
  \bibfield  {author} {\bibinfo {author} {\bibfnamefont {J.~P.}\ \bibnamefont
  {Garrahan}}\ and\ \bibinfo {author} {\bibfnamefont {I.}~\bibnamefont
  {Lesanovsky}},\ }\href {https://doi.org/10.1103/PhysRevLett.104.160601} {\bibfield
  {journal} {\bibinfo  {journal} {Phys. Rev. Lett.}\ }\textbf {\bibinfo
  {volume} {104}},\ \bibinfo {pages} {160601} (\bibinfo {year}
  {2010})}\BibitemShut {NoStop}%
\bibitem [{\citenamefont {van Horssen}\ and\ \citenamefont
  {Guţă}(2015)}]{GutavanHorssen15}%
  \BibitemOpen
  \bibfield  {author} {\bibinfo {author} {\bibfnamefont {M.}~\bibnamefont {van
  Horssen}}\ and\ \bibinfo {author} {\bibfnamefont {M.}~\bibnamefont
  {Guţă}},\ }\href {https://doi.org/10.1063/1.4907995} {\bibfield  {journal} {\bibinfo  {journal}
  {Journal of Mathematical Physics}\ }\textbf {\bibinfo {volume} {56}},\
  \bibinfo {pages} {022109} (\bibinfo {year} {2015})}\BibitemShut {NoStop}%
\bibitem [{\citenamefont {Burgarth}\ \emph {et~al.}(2015)\citenamefont
  {Burgarth}, \citenamefont {Giovannetti}, \citenamefont {Kato},\ and\
  \citenamefont {Yuasa}}]{Burgarth15}%
  \BibitemOpen
  \bibfield  {author} {\bibinfo {author} {\bibfnamefont {D.}~\bibnamefont
  {Burgarth}}, \bibinfo {author} {\bibfnamefont {V.}~\bibnamefont
  {Giovannetti}}, \bibinfo {author} {\bibfnamefont {A.~N.}\ \bibnamefont
  {Kato}}, \ and\ \bibinfo {author} {\bibfnamefont {K.}~\bibnamefont {Yuasa}},\
  }\href {https://doi.org/10.1088/1367-2630/17/11/113055} {\bibfield  {journal} {\bibinfo  {journal} {New Journal of
  Physics}\ }\textbf {\bibinfo {volume} {17}},\ \bibinfo {pages} {113055}
  (\bibinfo {year} {2015})}\BibitemShut {NoStop}%
\bibitem [{\citenamefont {Carollo}\ \emph {et~al.}(2019)\citenamefont
  {Carollo}, \citenamefont {Jack},\ and\ \citenamefont
  {Garrahan}}]{Garrahan19}%
  \BibitemOpen
  \bibfield  {author} {\bibinfo {author} {\bibfnamefont {F.}~\bibnamefont
  {Carollo}}, \bibinfo {author} {\bibfnamefont {R.~L.}\ \bibnamefont {Jack}}, \
  and\ \bibinfo {author} {\bibfnamefont {J.~P.}\ \bibnamefont {Garrahan}},\
  }\href {https://doi.org/10.1103/PhysRevLett.122.130605} {\bibfield  {journal}
  {\bibinfo  {journal} {Phys. Rev. Lett.}\ }\textbf {\bibinfo {volume} {122}},\
  \bibinfo {pages} {130605} (\bibinfo {year} {2019})}\BibitemShut {NoStop}%
\bibitem [{\citenamefont {Benoist}\ \emph {et~al.}(2022)\citenamefont
  {Benoist}, \citenamefont {H{\"{a}}nggli},\ and\ \citenamefont
  {Rouz{\'{e}}}}]{BR21}%
  \BibitemOpen
  \bibfield  {author} {\bibinfo {author} {\bibfnamefont {T.}~\bibnamefont
  {Benoist}}, \bibinfo {author} {\bibfnamefont {L.}~\bibnamefont
  {H{\"{a}}nggli}}, \ and\ \bibinfo {author} {\bibfnamefont {C.}~\bibnamefont
  {Rouz{\'{e}}}},\ }\href {https://doi.org/10.22331/q-2022-08-04-772} {\bibfield
  {journal} {\bibinfo  {journal} {{Quantum}}\ }\textbf {\bibinfo {volume}
  {6}},\ \bibinfo {pages} {772} (\bibinfo {year} {2022})}\BibitemShut {NoStop}%
\bibitem [{\citenamefont {Girotti}\ \emph {et~al.}(2023)\citenamefont
  {Girotti}, \citenamefont {Garrahan},\ and\ \citenamefont {Gu{\c t}{\u
  a}}}]{GiGaGu23}%
  \BibitemOpen
  \bibfield  {author} {\bibinfo {author} {\bibfnamefont {F.}~\bibnamefont
  {Girotti}}, \bibinfo {author} {\bibfnamefont {J.~P.}\ \bibnamefont
  {Garrahan}}, \ and\ \bibinfo {author} {\bibfnamefont {M.}~\bibnamefont {Gu{\c
  t}{\u a}}},\ }\href {https://doi.org/10.1007/s00023-023-01286-1} {\bibfield
  {journal} {\bibinfo  {journal} {Annales Henri Poincar{\'e}}\ } (\bibinfo
  {year} {2023}),\ 10.1007/s00023-023-01286-1}\BibitemShut {NoStop}%
\bibitem [{\citenamefont {Fallani}\ \emph {et~al.}(2022)\citenamefont
  {Fallani}, \citenamefont {Rossi}, \citenamefont {Tamascelli},\ and\
  \citenamefont {Genoni}}]{Fallani22}%
  \BibitemOpen
  \bibfield  {author} {\bibinfo {author} {\bibfnamefont {A.}~\bibnamefont
  {Fallani}}, \bibinfo {author} {\bibfnamefont {M.~A.~C.}\ \bibnamefont
  {Rossi}}, \bibinfo {author} {\bibfnamefont {D.}~\bibnamefont {Tamascelli}}, \
  and\ \bibinfo {author} {\bibfnamefont {M.~G.}\ \bibnamefont {Genoni}},\
  }\href {https://doi.org/10.1103/PRXQuantum.3.020310} {\bibfield  {journal} {\bibinfo
   {journal} {PRX Quantum}\ }\textbf {\bibinfo {volume} {3}},\ \bibinfo {pages}
  {020310} (\bibinfo {year} {2022})}\BibitemShut {NoStop}%
\bibitem [{\citenamefont {Gu\c{t}\u{a}}(2011)}]{Guta2011}%
  \BibitemOpen
  \bibfield  {author} {\bibinfo {author} {\bibfnamefont {M.}~\bibnamefont
  {Gu\c{t}\u{a}}},\ }\href {https://doi.org/10.1103/PhysRevA.83.062324} {\bibfield  {journal} {\bibinfo
  {journal} {Physical Review A}\ }\textbf {\bibinfo {volume} {83}},\ \bibinfo
  {pages} {062324} (\bibinfo {year} {2011})}\BibitemShut {NoStop}%
\bibitem [{\citenamefont {Gammelmark}\ and\ \citenamefont
  {M\o{}lmer}(2014)}]{Molmer14}%
  \BibitemOpen
  \bibfield  {author} {\bibinfo {author} {\bibfnamefont {S.}~\bibnamefont
  {Gammelmark}}\ and\ \bibinfo {author} {\bibfnamefont {K.}~\bibnamefont
  {M\o{}lmer}},\ }\href {https://doi.org/10.1103/PhysRevLett.112.170401} {\bibfield
  {journal} {\bibinfo  {journal} {Phys. Rev. Lett.}\ }\textbf {\bibinfo
  {volume} {112}},\ \bibinfo {pages} {170401} (\bibinfo {year}
  {2014})}\BibitemShut {NoStop}%
\bibitem [{\citenamefont {Catana}\ \emph {et~al.}(2015)\citenamefont {Catana},
  \citenamefont {Bouten},\ and\ \citenamefont {Gu\c{t}\u{a}}}]{GutaCB15}%
  \BibitemOpen
  \bibfield  {author} {\bibinfo {author} {\bibfnamefont {C.}~\bibnamefont
  {Catana}}, \bibinfo {author} {\bibfnamefont {L.}~\bibnamefont {Bouten}}, \
  and\ \bibinfo {author} {\bibfnamefont {M.}~\bibnamefont {Gu\c{t}\u{a}}},\
  }\href {https://doi.org/10.1088/1751-8113/48/36/365301} {\bibfield  {journal}
  {\bibinfo  {journal} {Journal of Physics A: Mathematical and Theoretical}\
  }\textbf {\bibinfo {volume} {48}},\ \bibinfo {pages} {365301} (\bibinfo
  {year} {2015})}\BibitemShut {NoStop}%
\bibitem [{\citenamefont {Gu\c{t}\u{a}}\ and\ \citenamefont
  {Kiukas}(2015)}]{Guta_2015}%
  \BibitemOpen
  \bibfield  {author} {\bibinfo {author} {\bibfnamefont {M.}~\bibnamefont
  {Gu\c{t}\u{a}}}\ and\ \bibinfo {author} {\bibfnamefont {J.}~\bibnamefont
  {Kiukas}},\ }\href {https://doi.org/10.1007/S00220-014-2253-0} {\bibfield  {journal}
  {\bibinfo  {journal} {Communications in Mathematical Physics}\ }\textbf
  {\bibinfo {volume} {335}},\ \bibinfo {pages} {1397} (\bibinfo {year}
  {2015})}\BibitemShut {NoStop}%
\bibitem [{\citenamefont {Gu\c{t}\u{a}}\ and\ \citenamefont
  {Kiukas}(2017)}]{Guta_2017}%
  \BibitemOpen
  \bibfield  {author} {\bibinfo {author} {\bibfnamefont {M.}~\bibnamefont
  {Gu\c{t}\u{a}}}\ and\ \bibinfo {author} {\bibfnamefont {J.}~\bibnamefont
  {Kiukas}},\ }\href {https://doi.org/10.1063/1.4982958} {\bibfield  {journal}
  {\bibinfo  {journal} {Journal of Mathematical Physics}\ }\textbf {\bibinfo
  {volume} {58}},\ \bibinfo {pages} {052201} (\bibinfo {year}
  {2017})}\BibitemShut {NoStop}%
\bibitem [{\citenamefont {Godley}\ and\ \citenamefont
  {Gu\c{t}\u{a}}(2023)}]{Godley2023}%
  \BibitemOpen
  \bibfield  {author} {\bibinfo {author} {\bibfnamefont {A.}~\bibnamefont
  {Godley}}\ and\ \bibinfo {author} {\bibfnamefont {M.}~\bibnamefont
  {Gu\c{t}\u{a}}},\ }\href {https://doi.org/10.22331/q-2023-04-06-973} {\bibfield
  {journal} {\bibinfo  {journal} {Quantum}\ }\textbf {\bibinfo {volume} {7}},\
  \bibinfo {pages} {973} (\bibinfo {year} {2023})}\BibitemShut {NoStop}%
\bibitem [{\citenamefont {Stannigel}\ \emph {et~al.}(2012)\citenamefont
  {Stannigel}, \citenamefont {Rabl},\ and\ \citenamefont
  {Zoller}}]{StannigelRAblZoller}%
  \BibitemOpen
  \bibfield  {author} {\bibinfo {author} {\bibfnamefont {K.}~\bibnamefont
  {Stannigel}}, \bibinfo {author} {\bibfnamefont {P.}~\bibnamefont {Rabl}}, \
  and\ \bibinfo {author} {\bibfnamefont {P.}~\bibnamefont {Zoller}},\ }\href
  {https://doi.org/10.1088/1367-2630/14/6/063014} {\bibfield  {journal} {\bibinfo
  {journal} {New J. Phys.}\ }\textbf {\bibinfo {volume} {12}},\ \bibinfo
  {pages} {063014} (\bibinfo {year} {2012})}\BibitemShut {NoStop}%
\bibitem [{\citenamefont {Gu\c{t}\u{a}}\ and\ \citenamefont
  {Kahn}(2006)}]{LAN1}%
  \BibitemOpen
  \bibfield  {author} {\bibinfo {author} {\bibfnamefont {M.}~\bibnamefont
  {Gu\c{t}\u{a}}}\ and\ \bibinfo {author} {\bibfnamefont {J.}~\bibnamefont
  {Kahn}},\ }\href {https://doi.org/10.1103/PhysRevA.73.052108} {\bibfield  {journal}
  {\bibinfo  {journal} {Physical Review A}\ }\textbf {\bibinfo {volume} {73}},\
  \bibinfo {pages} {052108} (\bibinfo {year} {2006})}\BibitemShut {NoStop}%
\bibitem [{\citenamefont {Kahn}\ and\ \citenamefont
  {Gu\c{t}\u{a}}(2009)}]{LAN3}%
  \BibitemOpen
  \bibfield  {author} {\bibinfo {author} {\bibfnamefont {J.}~\bibnamefont
  {Kahn}}\ and\ \bibinfo {author} {\bibfnamefont {M.}~\bibnamefont
  {Gu\c{t}\u{a}}},\ }\href {https://doi.org/10.1007/s00220-009-0787-3} {\bibfield
  {journal} {\bibinfo  {journal} {Communications in Mathematical Physics}\
  }\textbf {\bibinfo {volume} {289}},\ \bibinfo {pages} {597} (\bibinfo {year}
  {2009})}\BibitemShut {NoStop}%
\bibitem [{\citenamefont {Gu\c{t}\u{a}}\ \emph {et~al.}(2008)\citenamefont
  {Gu\c{t}\u{a}}, \citenamefont {Janssens},\ and\ \citenamefont {Kahn}}]{LAN5}%
  \BibitemOpen
  \bibfield  {author} {\bibinfo {author} {\bibfnamefont {M.}~\bibnamefont
  {Gu\c{t}\u{a}}}, \bibinfo {author} {\bibfnamefont {B.}~\bibnamefont
  {Janssens}}, \ and\ \bibinfo {author} {\bibfnamefont {J.}~\bibnamefont
  {Kahn}},\ }\href {https://doi.org/10.1007/s00220-007-0357-5} {\bibfield  {journal}
  {\bibinfo  {journal} {Commun. Math. Phys.}\ }\textbf {\bibinfo {volume}
  {277}},\ \bibinfo {pages} {127} (\bibinfo {year} {2008})}\BibitemShut
  {NoStop}%
\bibitem [{\citenamefont {Butucea}\ \emph {et~al.}(2016)\citenamefont
  {Butucea}, \citenamefont {Gu\c{t}\u{a}},\ and\ \citenamefont
  {Nussbaum}}]{LAN6}%
  \BibitemOpen
  \bibfield  {author} {\bibinfo {author} {\bibfnamefont {C.}~\bibnamefont
  {Butucea}}, \bibinfo {author} {\bibfnamefont {M.}~\bibnamefont
  {Gu\c{t}\u{a}}}, \ and\ \bibinfo {author} {\bibfnamefont {M.}~\bibnamefont
  {Nussbaum}},\ }\href {https://doi.org/10.1214/17-AOS1672} {\bibfield
   {journal} {\bibinfo  {journal} {Annals Statist.}\ }\textbf {\bibinfo
  {volume} {46}},\ \bibinfo {pages} {3676} (\bibinfo {year}
  {2016})}\BibitemShut {NoStop}%
\bibitem [{\citenamefont {Gu\c{t}\u{a}}\ and\ \citenamefont
  {Jencova}(2007)}]{LAN2}%
  \BibitemOpen
  \bibfield  {author} {\bibinfo {author} {\bibfnamefont {M.}~\bibnamefont
  {Gu\c{t}\u{a}}}\ and\ \bibinfo {author} {\bibfnamefont {A.}~\bibnamefont
  {Jencova}},\ }\href {https://doi.org/10.1007/s00220-007-0340-1} {\bibfield
  {journal} {\bibinfo  {journal} {Communications in Mathematical Physics}\
  }\textbf {\bibinfo {volume} {276}},\ \bibinfo {pages} {341} (\bibinfo {year}
  {2007})}\BibitemShut {NoStop}%
\bibitem [{\citenamefont {Gill}\ and\ \citenamefont
  {Gu\c{t}\u{a}}(2013)}]{LAN4}%
  \BibitemOpen
  \bibfield  {author} {\bibinfo {author} {\bibfnamefont {R.~D.}\ \bibnamefont
  {Gill}}\ and\ \bibinfo {author} {\bibfnamefont {M.}~\bibnamefont
  {Gu\c{t}\u{a}}},\ }\href {https://doi.org/10.1214/12-IMSCOLL909} {\bibfield
  {journal} {\bibinfo  {journal} {IMS Collections}\ }\textbf {\bibinfo {volume}
  {9}},\ \bibinfo {pages} {105} (\bibinfo {year} {2013})}\BibitemShut {NoStop}%
\bibitem [{\citenamefont {Yamagata}\ \emph {et~al.}(2013)\citenamefont
  {Yamagata}, \citenamefont {Fujiwara},\ and\ \citenamefont
  {Gill}}]{Yamagata13}%
  \BibitemOpen
  \bibfield  {author} {\bibinfo {author} {\bibfnamefont {K.}~\bibnamefont
  {Yamagata}}, \bibinfo {author} {\bibfnamefont {A.}~\bibnamefont {Fujiwara}},
  \ and\ \bibinfo {author} {\bibfnamefont {R.~D.}\ \bibnamefont {Gill}},\
  }\href {https://doi.org/10.1214/13-AOS1147} {\bibfield  {journal} {\bibinfo
  {journal} {The Annals of Statistics}\ }\textbf {\bibinfo {volume} {41}},\
  \bibinfo {pages} {2197} (\bibinfo {year} {2013})}\BibitemShut {NoStop}%
\bibitem [{\citenamefont {Fujiwara}\ and\ \citenamefont
  {Yamagata}(2020)}]{Fujiwara20}%
  \BibitemOpen
  \bibfield  {author} {\bibinfo {author} {\bibfnamefont {A.}~\bibnamefont
  {Fujiwara}}\ and\ \bibinfo {author} {\bibfnamefont {K.}~\bibnamefont
  {Yamagata}},\ }\href {https://doi.org/10.3150/19-BEJ1185} {\bibfield  {journal}
  {\bibinfo  {journal} {Bernoulli}\ }\textbf {\bibinfo {volume} {26}},\
  \bibinfo {pages} {2105 } (\bibinfo {year} {2020})}\BibitemShut {NoStop}%
\bibitem [{\citenamefont {Fujiwara}\ and\ \citenamefont
  {Yamagata}(2023)}]{Fujiwara22}%
  \BibitemOpen
  \bibfield  {author} {\bibinfo {author} {\bibfnamefont {A.}~\bibnamefont
  {Fujiwara}}\ and\ \bibinfo {author} {\bibfnamefont {K.}~\bibnamefont
  {Yamagata}},\ }\href {https://doi.org/10.1214/23-aos2285} {\bibfield  {journal}
  {\bibinfo  {journal} {Ann. Statist.}\ }\textbf {\bibinfo {volume} {51}},\
  \bibinfo {pages} {1159} (\bibinfo {year} {2023})}\BibitemShut {NoStop}%
\bibitem [{\citenamefont {Yuen}\ and\ \citenamefont {Lax}(1973)}]{YuenLax76}%
  \BibitemOpen
  \bibfield  {author} {\bibinfo {author} {\bibfnamefont {H.}~\bibnamefont
  {Yuen}}\ and\ \bibinfo {author} {\bibfnamefont {M.}~\bibnamefont {Lax}},\
  }\href {https://doi.org/10.1109/TIT.1973.1055103} {\bibfield  {journal} {\bibinfo
  {journal} {IEEE Trans. Inf. Theory}\ }\textbf {\bibinfo {volume} {19}},\
  \bibinfo {pages} {740} (\bibinfo {year} {1973})}\BibitemShut {NoStop}%
\bibitem [{\citenamefont {Braunstein}\ \emph {et~al.}(1996)\citenamefont
  {Braunstein}, \citenamefont {Caves},\ and\ \citenamefont {Milburn}}]{QCR2}%
  \BibitemOpen
  \bibfield  {author} {\bibinfo {author} {\bibfnamefont {S.~L.}\ \bibnamefont
  {Braunstein}}, \bibinfo {author} {\bibfnamefont {C.~M.}\ \bibnamefont
  {Caves}}, \ and\ \bibinfo {author} {\bibfnamefont {G.~J.}\ \bibnamefont
  {Milburn}},\ }\href {https://doi.org/10.1006/aphy.1996.0040} {\bibfield  {journal}
  {\bibinfo  {journal} {Annals of Physics}\ }\textbf {\bibinfo {volume}
  {247}},\ \bibinfo {pages} {135} (\bibinfo {year} {1996})}\BibitemShut
  {NoStop}%
\bibitem [{\citenamefont {Gill}\ and\ \citenamefont
  {Massar}(2000)}]{GillMassar}%
  \BibitemOpen
  \bibfield  {author} {\bibinfo {author} {\bibfnamefont {R.~D.}\ \bibnamefont
  {Gill}}\ and\ \bibinfo {author} {\bibfnamefont {S.}~\bibnamefont {Massar}},\
  }\href {https://doi.org/10.1103/PhysRevA.61.042312} {\bibfield  {journal} {\bibinfo
  {journal} {Phys. Rev. A}\ }\textbf {\bibinfo {volume} {61}},\ \bibinfo
  {pages} {042312} (\bibinfo {year} {2000})}\BibitemShut {NoStop}%
\bibitem [{\citenamefont {Pezze}\ and\ \citenamefont
  {Smerzi}(2014)}]{NullQFI1}%
  \BibitemOpen
  \bibfield  {author} {\bibinfo {author} {\bibfnamefont {L.}~\bibnamefont
  {Pezze}}\ and\ \bibinfo {author} {\bibfnamefont {A.}~\bibnamefont {Smerzi}},\
  }\href {https://doi.org/10.3254/978-1-61499-448-0-691} {\bibfield  {journal}
  {\bibinfo  {journal} {Atom Interferometry}\ }\textbf {\bibinfo {volume}
  {188}},\ \bibinfo {pages} {691} (\bibinfo {year} {2014})}\BibitemShut
  {NoStop}%
\bibitem [{\citenamefont {Pezz\`e}\ \emph {et~al.}(2017)\citenamefont
  {Pezz\`e}, \citenamefont {Ciampini}, \citenamefont {Spagnolo}, \citenamefont
  {Humphreys}, \citenamefont {Datta}, \citenamefont {Walmsley}, \citenamefont
  {Barbieri}, \citenamefont {Sciarrino},\ and\ \citenamefont
  {Smerzi}}]{NullQFI2}%
  \BibitemOpen
  \bibfield  {author} {\bibinfo {author} {\bibfnamefont {L.}~\bibnamefont
  {Pezz\`e}}, \bibinfo {author} {\bibfnamefont {M.~A.}\ \bibnamefont
  {Ciampini}}, \bibinfo {author} {\bibfnamefont {N.}~\bibnamefont {Spagnolo}},
  \bibinfo {author} {\bibfnamefont {P.~C.}\ \bibnamefont {Humphreys}}, \bibinfo
  {author} {\bibfnamefont {A.}~\bibnamefont {Datta}}, \bibinfo {author}
  {\bibfnamefont {I.~A.}\ \bibnamefont {Walmsley}}, \bibinfo {author}
  {\bibfnamefont {M.}~\bibnamefont {Barbieri}}, \bibinfo {author}
  {\bibfnamefont {F.}~\bibnamefont {Sciarrino}}, \ and\ \bibinfo {author}
  {\bibfnamefont {A.}~\bibnamefont {Smerzi}},\ }
  \href {https://doi.org/10.1103/PhysRevLett.119.130504} {\bibfield  {journal} {\bibinfo  {journal}
  {Phys. Rev. Lett.}\ }\textbf {\bibinfo {volume} {119}},\ \bibinfo {pages}
  {130504} (\bibinfo {year} {2017})}\BibitemShut {NoStop}%
\bibitem [{\citenamefont {Liu}\ \emph {et~al.}(2019)\citenamefont {Liu},
  \citenamefont {Yuan}, \citenamefont {Lu},\ and\ \citenamefont
  {Wang}}]{NullQFI3}%
  \BibitemOpen
  \bibfield  {author} {\bibinfo {author} {\bibfnamefont {J.}~\bibnamefont
  {Liu}}, \bibinfo {author} {\bibfnamefont {H.}~\bibnamefont {Yuan}}, \bibinfo
  {author} {\bibfnamefont {X.-M.}\ \bibnamefont {Lu}}, \ and\ \bibinfo {author}
  {\bibfnamefont {X.}~\bibnamefont {Wang}},\ }
  \href {https://doi.org/10.1088/1751-8121/ab5d4d} {\bibfield  {journal} {\bibinfo  {journal} {Journal
  of Physics A: Mathematical and Theoretical}\ }\textbf {\bibinfo {volume}
  {53}},\ \bibinfo {pages} {023001}\,(\bibinfo {year} {2019}) }\BibitemShut
  {NoStop}%
\bibitem [{\citenamefont {Attal}\ and\ \citenamefont
  {Pautrat}(2006{\natexlab{b}})}]{PautratAttal}%
  \BibitemOpen
  \bibfield  {author} {\bibinfo {author} {\bibfnamefont {S.}~\bibnamefont
  {Attal}}\ and\ \bibinfo {author} {\bibfnamefont {Y.}~\bibnamefont
  {Pautrat}},\ }\href {https://doi.org/10.1007/s00023-005-0242-8} {\bibfield  {journal} {\bibinfo  {journal} {Ann.
  Henri Poincar\'{e}}\ }\textbf {\bibinfo {volume} {7}},\ \bibinfo {pages} {59}
  ,(\bibinfo {year} {2006}{\natexlab{b}})}\BibitemShut {NoStop}%
\bibitem [{\citenamefont {Petz}(1986)}]{Petz86}%
  \BibitemOpen
  \bibfield  {author} {\bibinfo {author} {\bibfnamefont {D.}~\bibnamefont
  {Petz}},\ }\href {https://doi.org/10.1007/BF01212345} {\bibfield  {journal} {\bibinfo  {journal}
  {Communications in Mathematical Physics}\ }\textbf {\bibinfo {volume}
  {105}},\ \bibinfo {pages} {123–131}\,(\bibinfo {year} {1986}) }\BibitemShut
  {NoStop}%
\bibitem [{\citenamefont {Petz}(1988)}]{Petz88}%
  \BibitemOpen
  \bibfield  {author} {\bibinfo {author} {\bibfnamefont {D.}~\bibnamefont
  {Petz}},\ }\href {https://doi.org/10.1093/qmath/39.1.97} {\bibfield  {journal} {\bibinfo  {journal}
  {Quarterly Journal of Mathematics}\ }\textbf {\bibinfo {volume} {39}},\
  \bibinfo {pages} {97–108}\,(\bibinfo {year} {1988}) }\BibitemShut {NoStop}%
\bibitem [{\citenamefont {Barnum}\ and\ \citenamefont
  {Knill}(2002)}]{Barnum02}%
  \BibitemOpen
  \bibfield  {author} {\bibinfo {author} {\bibfnamefont {H.}~\bibnamefont
  {Barnum}}\ and\ \bibinfo {author} {\bibfnamefont {E.}~\bibnamefont {Knill}},\
  }\href {https://doi.org/10.1063/1.1459754} {\bibfield  {journal} {\bibinfo  {journal} {Journal of
  Mathematical Physics}\ }\textbf {\bibinfo {volume} {43}},\ \bibinfo {pages}
  {2097–2106} \, (\bibinfo {year} {2002})}\BibitemShut {NoStop}%
\bibitem [{\citenamefont {Tsang}(2024)}]{Tsang24}%
  \BibitemOpen
  \bibfield  {author} {\bibinfo {author} {\bibfnamefont {M.}~\bibnamefont
  {Tsang}},\ } {\enquote {\bibinfo {title} {Quantum reversal: a
  general theory of coherent quantum absorbers},}\ }\href {https://doi.org/10.22331/q-2025-02-26-1650} 
  {\bibfield  {journal}
  {\bibinfo  {journal} {Quantum} \textbf {\bibinfo {volume}
  {9}}, \bibinfo {pages} {1650} }\,(\bibinfo {year} {2025})}\BibitemShut {NoStop}%
\bibitem [{\citenamefont {Johansson}\ \emph {et~al.}(2013)\citenamefont
  {Johansson}, \citenamefont {Nation},\ and\ \citenamefont {Nori}}]{Qutip}%
  \BibitemOpen
  \bibfield  {author} {\bibinfo {author} {\bibfnamefont {J.~R.}\ \bibnamefont
  {Johansson}}, \bibinfo {author} {\bibfnamefont {P.~D.}\ \bibnamefont
  {Nation}}, \ and\ \bibinfo {author} {\bibfnamefont {F.}~\bibnamefont
  {Nori}},\ }\href {https://doi.org/10.1016/j.cpc.2012.11.019} {\bibfield  {journal}
  {\bibinfo  {journal} {Comp. Phys. Comm.}\ }\textbf {\bibinfo {volume}
  {184}},\ \bibinfo {pages} {1234}\,(\bibinfo {year} {2013}) }\BibitemShut
  {NoStop}%
\bibitem [{\citenamefont {Radcliffe}(1971)}]{Radcliffe}%
  \BibitemOpen
  \bibfield  {author} {\bibinfo {author} {\bibfnamefont {J.~M.}\ \bibnamefont
  {Radcliffe}},\ }\href {https://dx.doi.org/10.1088/0305-4470/4/3/009} {\bibfield  {journal} {\bibinfo  {journal} {J.
  Phys. A: Gen. Phys.}\ }\textbf {\bibinfo {volume} {4}},\ \bibinfo {pages}
  {313}\,(\bibinfo {year} {1971}) }\BibitemShut {NoStop}%
\bibitem [{\citenamefont {van~der Vaart}(1998)}]{Vaart1998}%
  \BibitemOpen
  \bibfield  {author} {\bibinfo {author} {\bibfnamefont {A.~W.}\ \bibnamefont
  {van~der Vaart}},\ }\href {https://doi.org/10.1017/CBO9780511802256} {\emph
  {\bibinfo {title} {Asymptotic Statistics}}}\ \bibinfo  {publisher}
  {Cambridge University Press,\ (\bibinfo {year} {1998})}\BibitemShut {NoStop}%
\bibitem [{\citenamefont {Lin}(2017)}]{Li17}%
  \BibitemOpen
  \bibfield  {author} {\bibinfo {author} {\bibfnamefont {G.~D.}\ \bibnamefont
  {Lin}},\ }\href {https://doi.org/10.1186/s40488-017-0059-2} {\bibfield  {journal} {\bibinfo  {journal} {Journal
  of Statistical Distributions and Applications}\ }\textbf {\bibinfo {volume}
  {4}},\ \bibinfo {pages} {5}\,(\bibinfo {year} {2017}) } \BibitemShut {NoStop}%
\bibitem [{\citenamefont {Schmüdgen}(2020)}]{Sc20}%
  \BibitemOpen
  \bibfield  {author} {\bibinfo {author} {\bibfnamefont {K.}~\bibnamefont
  {Schmüdgen}},\ } {\enquote {\bibinfo {title} {Ten lectures on
  the moment problem},}\ }  \Eprint {https://doi.org/10.48550/arXiv.2008.12698} {arXiv:2008.12698 [math.FA]}\,(\bibinfo {year} {2020}) \BibitemShut
  {NoStop}%
\end{thebibliography}

%

\end{document}